\documentclass[12pt,a4paper,openright,twoside]{report}
\usepackage[OT1]{fontenc}    
\usepackage[italian]{babel}  
\usepackage[square]{natbib}  
\usepackage{graphicx}        
\usepackage{lmodern}         
\usepackage{enumitem}        
\usepackage{float}           
\usepackage{array}           

\renewcommand{\author}   { Andrea Simonetto }
\newcommand  {\subject}  { Deep Inference }
\renewcommand{\title}    { Indagini in Deep Inference }
\newcommand  {\keywords} { Deep Inference, Teoria della dimostrazione, Proof theory, Cut elimination }

\usepackage[pdfmark,          
	pdfauthor={\author}, pdfkeywords={\keywords},
	pdfsubject={\subject}, pdftitle={\title}]{hyperref}

\usepackage{amsmath,amssymb}  
\usepackage{amsthm}           
\usepackage{mathrsfs}         
\usepackage{cmll}             
\usepackage{bussproofs}       
\usepackage{virginialake}     
\usepackage{iidi}

\newtheorem{thm}{Teorema}[section]
\newtheorem{dfn}[thm]{Definizione}
\newtheorem{lem}[thm]{Lemma}
\newtheorem{cor}[thm]{Corollario}
\newtheorem{prop}[thm]{Proposizione}

\usepackage{indentfirst}       
\usepackage{fancyhdr}          
\hypersetup {                    
	bookmarks=true,          
	colorlinks=true         
}
\floatstyle{boxed}
\restylefloat{figure}
\makeatletter
\renewcommand\fs@boxed{\def\@fs@cfont{\bfseries}\let\@fs@capt\floatc@plain
  \def\@fs@pre{\setbox\@currbox\vbox{\hbadness10000
    \moveleft10.4pt\vbox{\advance\hsize by20.8pt
      \hrule \hbox to\hsize{\vrule\kern10pt
        \vbox{\kern10pt\box\@currbox\kern10pt}\kern10pt\vrule}\hrule}}}%
  \def\@fs@mid{\kern10pt}%
  \def\@fs@post{}\let\@fs@iftopcapt\iffalse}

\begin{document}
\begin{titlepage}
	\topmargin=0cm
	\begin{center}
		\mbox{{\large{\textsc{Alma Mater Studiorum $\cdot$ Universit\`a di Bologna}}}}
		\rule[0.1cm]{14.0cm}{0.1mm}
		\rule[0.5cm]{14.0cm}{0.6mm}
		{\small{\bf FACOLT\`A DI SCIENZE MATEMATICHE, FISICHE E NATURALI\\
		Corso di Laurea Magistrale in Informatica}}
	\end{center}
	\vspace{15mm}
	\begin{center}
		{\LARGE{\bf INDAGINI IN}}\\
		\vspace{3mm}
		{\LARGE{\bf DEEP}}\\
		\vspace{3mm}
		{\LARGE{\bf INFERENCE}}\\
		\vspace{19mm} {\large{\bf Tesi di Laurea in Tipi e Linguaggi di Programmazione}}
	\end{center}
	\vspace{40mm}\par\noindent
	\begin{minipage}[t]{0.52\textwidth}
		{\large{\bf Relatore:\\
		Chiar.mo Prof.\\
		Simone Martini}}
	\end{minipage}
	\hfill
	\begin{minipage}[t]{0.52\textwidth}\raggedleft
		{\large{\bf Presentata da:\\
		Andrea Simonetto}}
	\end{minipage}
	\vspace{20mm}
	\begin{center}
		{\large{\bf II Sessione\\
		Anno Accademico 2009/2010}}
	\end{center}
\end{titlepage}
\clearpage{\pagestyle{empty}\cleardoublepage}
\begin{titlepage}
	\thispagestyle{empty}
	\topmargin=5cm
	\raggedleft\large\emph{
		Alla venerata memoria \\
		di mio nonno, Gino Simonetto.} \\
	\vspace{3em}
	\raggedleft\large\emph{
		Wir m\"ussen wissen. Wir werden wissen. \\
		(David Hilbert)} \\
	\newpage
	\clearpage{\pagestyle{empty}\cleardoublepage}
\end{titlepage}
\pagenumbering{roman}

\phantomsection\addcontentsline{toc}{chapter}{Introduzione}
\chapter*{Introduzione}
\rhead[\fancyplain{}{\bfseries INTRODUZIONE}]{\fancyplain{}{\bfseries\thepage}}
\lhead[\fancyplain{}{\bfseries\thepage}]{\fancyplain{}{\bfseries INTRODUZIONE}}

La matematica ama parlare di s\'e stessa. Dalla teoria dei numeri -- ritenuta da Gauss la ``regina della matematica'' -- all'analisi -- l'hilbertiana ``sinfonia coerente dell'universo'' -- la matematica pura \`e abituata a porsi ad oggetto delle proprie speculazioni. Anche la \emph{logica matematica} possiede spiccate tendenze narcisistiche\footnote{Secondo Locke, la logica \`e ``l'anatomia del pensiero''. Occorre tuttavia precisare che usando il termine ``logica'', Locke intendeva quella parte di filosofia che studia il ragionamento e l'argomentazione; la logica matematica cominci\`o a fiorire solo un secolo e mezzo dopo la sua morte.}; nondimeno, vista la sua propensione ad affrontare questioni di fondamento, \`e certamente una delle branche che ama maggiormente parlare di matematica. Tra i vari contributi \`e doveroso citare, come esempi eccellenti, l'\emph{aritmetica di Peano} e la \emph{teoria assiomatica degli insiemi} di Zermelo--Fraenkel.

Uno dei concetti cardine di tutta la matematica \`e quello di \emph{dimostrazione}. Per studiare questi oggetti, i logici hanno sviluppato un'infrastruttura nota come \emph{teoria della dimostrazione}. In questa tesi faremo una panoramica sulla teoria della dimostrazione, concentrandoci su uno degli sviluppi pi\`u recenti, una tecnica nota come \emph{deep inference}.

La deep inference \`e una nuova \emph{metodologia} in proof theory utile a progettare \emph{famiglie di sistemi formali} (o \emph{formalismi}) con buone propriet\`a, quali:
\begin{itemize}
	\item l'\emph{efficienza} nella rappresentazione delle dimostrazioni: alcuni sistemi rendono disponibili dimostrazioni pi\`u brevi di quanto possano fare altri (questi aspetti sono studiati in \emph{complessit\`a delle dimostrazioni} o \emph{proof complexity});
	\item \emph{analiticit\`a}: alcuni sistemi vengono naturalmente con algoritmi di \emph{ricerca delle dimostrazioni} (o di \emph{proof search}) \emph{efficienti}, altri no, altri ancora solo con alcuni accorgimenti. L'analiticit\`a \`e la propriet\`a chiave per ottenere algoritmi di proof search efficienti;
	\item l'abilit\`a di esprimere dimostrazioni che sono matematicamente naturali, cio\`e senza artefatti sintattici ``amministrativi'' (si parla in questi casi di \emph{burocrazia delle dimostrazioni}). Uno dei problemi di ricerca principali in proof theory \`e trovare una buona corrispondenza tra le dimostrazioni e il loro significato. In particolare, il problema dell'\emph{identit\`a delle dimostrazioni} \`e prominente, e consiste nel trovare nozioni di equivalenza tra dimostrazioni non banali, supportate da semantiche appropriate alle dimostrazioni e ai sistemi formali.
\end{itemize}

I \emph{formalismi} controllano, in larga parte, la progettazione delle regole d'inferenza. Per esempio, la deduzione naturale prescrive che, per ogni connettivo, siano date due regole: una d'introduzione e una di eliminazione. In tutti i formalismi tradizionali (ma anche in quelli pi\`u moderni derivati dai primi), viene adottato un meccanismo noto come \emph{shallow inference} o \emph{inferenza di superficie}, nel quale le regole di inferenza operano sui connettivi pi\`u prossimi alla radice delle formule -- quando vengono viste come alberi, cio\`e quando ci si concentra sulla loro \emph{struttura sintattica}. 

L'inferenza di superficie \`e una metodologia molto naturale, poich\'e permette di procedere per induzione strutturale diretta sulle formule. Tuttavia non \`e ottimale riguardo alcune propriet\`a dei sistemi formali, quali quelle sopra menzionate. In particolare:
\begin{itemize}
	\item sembra sia incapace di fornire formalismi analitici che siano efficienti riguardo la complessit\`a delle dimostrazioni;
	\item le dimostrazioni tendono ad avere molta burocrazia, cio\`e rappresentazioni sintatticamente complesse degli argomenti matematici.
\end{itemize}

Inoltre la shallow inference fatica a relazionarsi con le logiche modali. Teorie logiche modali possono essere definite nei sistemi di Frege-Hilbert, ma ottenere analiticit\`a per essi \`e una sfida molto ardua, in alcuni casi ancora irrisolta. In pi\`u \`e altrettanto difficile (se non addirittura impossibile) esprimere sistemi formali per logiche non-commutative.

Nel seguito mostreremo alcuni tra i maggiori risultati di teoria della dimostrazione ottenuti mediante un formalismo deep inference, noto come \emph{calcolo delle strutture}. Il calcolo delle strutture \`e un contributo importante nello sviluppo della metodologia deep inference, per la sua semplicit\`a e la sua somiglianza coi formalismi tradizionali. Grazie al calcolo delle strutture, sono stati ottenuti i seguenti risultati:
\begin{itemize}
	\item la logica classica, intuizionista, lineare a alcune logiche modali possono essere espresse in sistemi che godono di analiticit\`a;
	\item \`e possibile esprimere logiche lineari munite di \emph{operazioni} \emph{non-commuta-tive} in sistemi analitici, e si dimostra che queste logiche non possono essere formalizzate analiticamente nel calcolo dei sequenti; inoltre questi sistemi logici sono in forte corrispondenza con le algebre di processo;
	\item sono state sviluppate tecniche nuove e generali di normalizzazione, e sono state scoperte \emph{nozioni del tutto nuove} di normalizzazione, in aggiunta alla tradizionale \emph{cut elimination};
	\item la maggior parte dei sistemi sviluppati sono costituiti interamente da regole d'inferenza \emph{locali}; una regola d'inferenza locale ha \emph{complessit\`a computazionale costante}. La localit\`a \`e una propriet\`a difficile da conseguire, e non \`e ottenibile nel calcolo dei sequenti per la logica classica;
	\item i sistemi ottenuti sono estremamente modulari; questo significa una \emph{forte indipendenza tra le regole d'inferenza};
	\item moti sistemi sono stati implementati, grazie a tecniche che producono regole d'inferenza atte a migliorare l'efficienza senza sacrificare le propriet\`a teoriche;
	\item tutti i sistemi ottenuti sono semplici, nel senso che le regole d'inferenza sono \emph{contenute e intelligibili}.
\end{itemize}

Il calcolo delle strutture \`e una generalizzazione di molti formalismi shallow inference, in particolare del calcolo dei sequenti. Questo significa che ogni dimostrazione data in questi formalismi shallow inference, pu\`o essere ``mimata'' nel calcolo delle strutture, preservandone la complessit\`a e senza perdita di propriet\`a strutturali.

\clearpage{\pagestyle{empty}\cleardoublepage}
\tableofcontents
\rhead[\fancyplain{}{\bfseries\leftmark}]{\fancyplain{}{\bfseries\thepage}}
\lhead[\fancyplain{}{\bfseries\thepage}]{\fancyplain{}{\bfseries INDICE}}
\clearpage{\pagestyle{empty}\cleardoublepage}

\listoffigures
\clearpage{\pagestyle{empty}\cleardoublepage}


\chapter{Formalismi e metodologie}
\lhead[\fancyplain{}{\bfseries\thepage}]{\fancyplain{}{\bfseries\rightmark}}
\pagenumbering{arabic}


La teoria della dimostrazione nacque sul finire del XIX secolo, ad opera di David Hilbert e dei suoi collaboratori. Fu un periodo in cui la matematica visse una crisi senza precedenti, intaccata in prossimit\`a dei suoi fondamenti logici dal manifestarsi di una serie di paradossi, proprio mentre la scuola intuizionista di Brouwer ne metteva in dubbio alcuni princ\`ipi filosofici basilari, fatti che sommati minavano alla base la maggior parte della matematica esistente. 

Tuttavia, forti evidenze empiriche suffragavano la matematica conosciuta, e molti matematici rifiutarono di abbandonarla o rifondarla: tra loro, Hilbert avanz\`o un programma di salvataggio completo. Egli propose di formulare la matematica classica come teoria formale assiomatica, e in seguito di provarne la \emph{consistenza} (ossia la non contraddittoriet\`a).

Prima della proposta di Hilbert, la consistenza di teorie assiomatiche veniva provata esibendo un ``modello'': data una teoria assiomatica, un \emph{modello} \`e un sistema di oggetti, presi da qualche altra teoria, tali da soddisfare gli assiomi, cio\`e, ad ogni oggetto o nozione primitiva della teoria assiomatica, viene fatto corrispondere un oggetto o una nozione dell'altra teoria, in modo tale che gli assiomi corrispondano a teoremi nell'altra teoria. Se l'altra teoria \`e consistente, anche quella assiomatica deve esserlo. Un esempio famoso \`e dato dalla dimostrazione di Beltrami (1868) della consistenza della geometria iperbolica: egli prov\`o che le rette nel piano non-euclideo della geometria iperbolica, potevano essere rappresentate dalle geodesiche su una superficie di curvatura costante e negativa nello spazio euclideo. Da questo concluse che il piano iperbolico dev'essere consistente, a patto che la geometria euclidea lo sia.

\`E chiaro che il metodo del modello \`e relativo. La teoria assiomatica \`e consistente solo se il suo modello lo \`e. Ma per provare l'assoluta consistenza della matematica classica, il metodo dei modelli non offriva speranze: nessuna teoria matematica era accettabile come modello, poich\'e da ognuna di esse sarebbe fatalmente riemerso il problema di partenza, cio\`e dimostrarne la consistenza.

Hilbert propose di affrontare il problema in maniera diretta: per provare la consistenza di una teoria, si deve dimostrare al suo interno una proposizione sulla teoria stessa, cio\`e un teorema su tutte le possibili dimostrazioni della teoria. La branca di matematica che si occupa di questi aspetti di formalizzazione e riflessione, venne battezzata da Hilbert ``metamatematica'' o ``teoria della dimostrazione''.

Purtroppo il sogno di Hilbert s'infranse nel 1931 con il Secondo Teorema d'Incompletezza di~\cite{God31}, che enuncia l'impossibilit\`a per sistemi abbastanza espressivi da formalizzare l'aritmetica di dimostrare la propria consistenza: purtroppo la quasi totalit\`a della matematica da salvare, passava per l'aritmetica. Tuttavia la teoria della dimostrazione sopravvisse a questo scossone, diventando importante in vari ambiti, tra i quali l'informatica.

In informatica, i dimostratori automatici di teoremi richiedono uno studio della struttura combinatoriale delle dimostrazioni, mentre nella programmazione logica la deduzione \`e usata come fondamento della computazione. Inoltre esistono forti connessioni tra sistemi logici e linguaggi di programmazione funzionali, e tecniche di proof theory sono state utilizzate per porre dei vincoli di complessit\`a computazionale ad alcuni di questi linguaggi (ad esempio in~\cite{Gir95a}).


Uno dei princ\`ipi fondamentali in proof theory \`e che la formalizzazione di una teoria richiede una totale astrazione dal significato, cio\`e un \emph{sistema formale} dovrebbe essere una mera manipolazione simbolica spogliata di ogni interpretazione semantica. Dato un sistema formale, distinguiamo il livello rigoroso del sistema stesso (o \emph{livello oggetto}), dal livello in cui esso viene studiato (il \emph{meta-livello}) espresso nel linguaggio della matematica intuitiva e informale. 

Inoltre, per essere convincenti, gli strumenti usati nelle meta-teorie dovrebbero essere ristretti a tecniche -- chiamate \emph{finitarie} dai formalisti, o, in un accezione pi\`u moderna, \emph{combinatorie} -- che impiegano solo oggetti intuitivi e processi effettivi (in accordo con la scuola intuizionista). Nessuna classe infinita di oggetti deve poter essere trattata come un ``tutto''; le prove di esistenza dovrebbero sempre esibire, almeno implicitamente, un testimone.

La proof theory \`e dunque una collezione di meta-teorie finitarie, espresse nel linguaggio ordinario e con l'ausilio di simboli matematici -- come variabili di meta-livello (o \emph{meta-variabili}) introdotte ove necessario -- tali da caratterizzare le propriet\`a dei vari sistemi formali. In questo capitolo introdurremo le nozioni di base della teoria della dimostrazione, partendo da insiemi finiti e generando quelli infiniti con procedure effettive, calcolabili.

Quella di cui abbiamo discusso finora, va oggi sotto il nome di teoria della dimostrazione \emph{strutturale}, cio\`e un'analisi combinatoriale della struttura delle dimostrazioni formali; gli strumenti centrali sono il \emph{Teorema di eliminazione del taglio} e quello di \emph{normalizzazione}.

Il percorso che seguiremo in questo capitolo \`e liberamente ispirato a~\cite{Kle52}, e si articola in quattro sezioni, le prime tre piuttosto standard, mentre la quarta aggiunta appositamente per trattare il tema della tesi, ossia la deep inference. Nell'ordine:

\begin{enumerate}
	\item si definir\`a uno strumento linguistico formale, in grado di produrre dei \emph{linguaggi oggetto} che costituiranno gli elementi base della logica da indagare;
	\item sar\`a introdotto un livello linguistico formale superiore (o \emph{meta-livello}), tale da permetterci di ragionare sui vari linguaggi oggetto, e saranno date le definizioni e gli strumenti di indagine basilari;
	\item verranno formalizzati i concetti di \emph{deduzione} e di \emph{dimostrazione}, a partire da un generico \emph{linguaggio oggetto}, usando gli strumenti del \emph{meta-livello}, e saranno presi in esame i \emph{formalismi} (i.e. le famiglie di sistemi) esprimibili con tali strumenti;
	\item si definirranno e si metteranno a confronto le due \emph{metodologie di inferenza}: di superficie e di profondit\`a (shallow \emph{versus} deep inference).
\end{enumerate}
Affronteremo il tutto sempre tenedo presente il vincolo di effettiva costruibilit\`a delle procedure e la caratterizzazione combinatoria delle tecniche e degli strumenti via via introdotti, in pieno stile formalista.

\section{Linguaggi formali} 

In questa sezione svilupperemo le basi di linguaggi formali che utilizzeremo da qui in avanti. Alcuni concetti saranno forniti in maniera intuitiva, altri in modo pi\`u preciso: per approfondimenti su linguaggi formali e grammatiche, si rimanda ad~\cite{Aho72, Aho06}.

\begin{dfn}[Alfabeti, stringhe, linguaggi, sottolinguaggi]
Un \emph{alfabeto} $\Sigma$ \`e un insieme finito non vuoto di \emph{simboli}. Una \emph{stringa su un alfabeto}\footnote{Per convenzione useremo le lettere minuscole prese dall'inizio dell'alfabeto inglese per denotare i simboli, e le lettere minuscole alla fine dell'alfabeto, di solito $w, x, y, z$, per denotare le stringhe.} $\Sigma$ \`e una sequenza finita di simboli scelti da $\Sigma$. La \emph{stringa vuota} $\epsilon$ \`e la stringa composta da zero simboli; essa \`e una stringa su qualunque alfabeto.

Siano $x = a_1 \cdots a_n$ e $y = b_1 \cdots b_m$ due stringhe su un alfabeto $\Sigma$: la loro \emph{concatenazione} (si denota giustapponendo $x$ a $y$) \`e la stringa $xy = a_1 \cdots a_n b_1 \cdots b_m$. In particolare, per ogni stringa $w$, si ha $\epsilon w = w \epsilon = w$. Dato un alfabeto $\Sigma$, definiamo:
$$
\begin{array}{lclll}
	\Sigma^0 & = & \{\epsilon\} & \qquad & \mbox{(Stringa di \emph{lunghezza} $0$)} \\
	\Sigma^{n+1} & = & \{ a w \:|\: a \in \Sigma, w \in \Sigma^n\} & \qquad & \mbox{(Stringhe di \emph{lunghezza} $n+1$)} \\
	\Sigma^* & = & \bigcup_{n \in \mathbb{N}} \Sigma^n & \qquad & \mbox{(Stringhe su $\Sigma$)}
\end{array}
$$

Un \emph{linguaggio} $\mathscr{L}$ \emph{su un alfabeto} $\Sigma$ \`e un'insieme di stringhe su $\Sigma$ (cio\`e $\mathscr{L} \subseteq \Sigma^*$). Infine, dato un linguaggio $\mathscr{L}$, si definisce \emph{sottolinguaggio di} $\mathscr{L}$ qualunque insieme di stringhe $\mathscr{L}' \subseteq \mathscr{L}$.
\end{dfn}

\begin{dfn}[Grammatiche, linguaggio generato]
Una grammatica \`e una quadrupla $G = (\Sigma, \mathcal{C}, S, \mathcal{P})$, dove:
\begin{itemize}
	\item $\Sigma$ \`e un alfabeto di \emph{simboli grammaticali};
	\item $\mathcal{C} \subseteq \Sigma$ \`e un insieme di simboli, detti \emph{categorie sintattiche} (o \emph{simboli non terminali}, in contrapposizione con gli altri simboli grammaticali $\Sigma{\smallsetminus}\mathcal{C}$, chiamati \emph{simboli terminali});
	\item $S \in \mathcal{C}$ \`e una particolare categoria sintattica, chiamata \emph{simbolo iniziale}, che rappresenta il linguaggio da definire;
	\item $\mathcal{P} \subseteq \Sigma^*{\times}\Sigma^*$ \`e un insieme di coppie di stringhe $(\alpha, \beta) \in \Sigma^*{\times}\Sigma^*$, chiamate \emph{produzioni grammaticali}. $\alpha$ \`e chiamata \emph{testa della produzione}, mentre $\beta$ \`e il \emph{corpo della produzione}.
\end{itemize}

Data una grammatica $(\Sigma, \mathcal{C}, S, \mathcal{P})$, la \emph{riscrittura ad un passo} ($\leadsto$) \`e un'applicazione di una delle produzioni in $\mathcal{P}$. Formalmente: siano $\alpha, \beta, x, y \in \Sigma^*$, allora:
$$
	x \alpha y \leadsto x \beta y \quad\mbox{ sse }\quad \mbox{esiste }(\alpha, \beta) \in \mathcal{P}
$$
La \emph{riscrittura multipasso} ($\leadsto^*$) \`e la chiusura riflessiva e transitiva di quella ad un passo:
$$
	w_1 \leadsto^* w_n \quad\mbox{ sse }\quad w_1 \leadsto \cdots \leadsto w_n \quad\mbox{per qualche $n \ge 0$}
$$
In particolare $w \leadsto^* w$ per ogni $w \in \Sigma^*$.

Il \emph{linguaggio generato da una grammatica} $G = (\Sigma, \mathcal{C}, S, \mathcal{P})$ (denotato con $\mathscr{L}_G$) \`e l'insieme delle stringhe di simboli terminali ottenibili tramite riscrittura multipasso a partire dal simbolo iniziale. In simboli:
$$
	\mathscr{L}_G = \{ w \in (\Sigma{\smallsetminus}\mathcal{C})^* \:|\: S \leadsto^* w\}
$$
\end{dfn}

\begin{dfn}[Grammatiche context-free, BNF]
Sia $G = (\Sigma, \mathcal{C}, S, \mathcal{P})$ una grammatica. Una \emph{regola di produzione} \`e \emph{context-free} se \`e della forma $(P, \beta)$ con $P \in \mathcal{C}$ e $\beta \in \Sigma^*$. Una \emph{grammatica} si dice \emph{context-free} se ogni sua regola di produzione \`e context-free.

Un modo compatto ed elegante per scrivere le regole di produzione grammaticale, \`e quello di usare la \emph{forma di Backus-Naur} o \mbox{BNF}~(\cite{Bac60}). Sia $G = (\Sigma, \mathcal{C}, S, \mathcal{P})$ una grammatica e sia:
$$
	\mathcal{P} = \{ (\alpha_1, \beta_{1,1}), \ldots, (\alpha_1, \beta_{1,n_1}), \ldots, (\alpha_k, \beta_{k,1}), \ldots, (\alpha_k, \beta_{k,n_k}), \ldots \}
$$
il suo insieme di produzioni. Allora $\mathcal{P}$ in BNF si rappresenta come segue:
\begin{eqnarray*}
	\alpha_1 & ::= & \beta_{1,1} \:|\: \beta_{1,2} \:|\: \cdots \:|\: \beta_{1,n_1} \\
	\cdots & \cdots & \cdots \\
	\alpha_k & ::= & \beta_{k,1} \:|\: \cdots \:|\: \beta_{k,n_k} \\
	\cdots & \cdots & \cdots 
\end{eqnarray*}
\end{dfn}


L'ultimo (ma non ultimo) strumento sintattico che consideriamo, serve per fare emergere le \emph{profonde simmetrie} che soggiacciono ai sistemi logici formali, ed \`e uno degli strumenti pi\`u usati in proof theory: il sequente. I sequenti sono una notazione sintattica, finalizzata ad inserire le formule logiche in ambienti adatti al ragionamento logico-deduttivo. Pi\`u precisamente:

\begin{dfn}[Sequente]\label{def:seq}
Dato un linguaggio $\mathscr{L}$, un \emph{sequente} \`e un espressione del tipo:
$$
	\Gamma \vdash \Delta
$$
dove $\Gamma, \Delta$ sono \emph{liste finite} (eventualmente vuote) di stringhe di $\mathscr{L}$ -- con le usuali operazioni definite sulle liste: $\Gamma, P$ \`e l'aggiunta di una stringa $P$ in coda ad una lista $\Gamma$, mentre $\Gamma, \Gamma'$ \`e la concatenazione delle liste $\Gamma$ e $\Gamma'$; non ci sono simboli per la lista vuota.

Il simbolo $\vdash$ \`e noto come \emph{turnstile} o \emph{tornello} e fu originariamente introdotto in~\cite{Fre79}.
\end{dfn}

L'idea intuitiva \`e che il sequente afferma (ipotizza) la deducibilit\`a di almeno una formula logica in $\Delta$ a partire dalle premesse in $\Gamma$. Se $\Gamma = P_1, \ldots, P_n$ e $\Delta = Q_1, \ldots, Q_m$, il sequente $\Gamma \vdash \Delta$ \`e da intendersi come:
\begin{center}
	\textbf{Se}\quad $P_1$ \emph{e} $\cdots$ \emph{e} $P_n$ \quad\textbf{allora}\quad $Q_1$ \emph{oppure} $\cdots$ \emph{oppure} $Q_m$
\end{center}
dove i significati di \emph{``Se \ldots~allora \ldots''}, \emph{``e''} ed \emph{``oppure''} devono essere resi espliciti in maniera formale.

Il sequente avente una lista vuota alla \emph{destra} del tornello ($\Gamma \vdash ~$), afferma l'\emph{inconsistenza delle premesse}, quello avente la lista vuota alla \emph{sinistra} del tornello ($~ \vdash \Delta$), afferma che $\Delta$ \`e un \emph{teorema}, ossia che \`e vero a prescindere da ogni premessa. Il \emph{sequente vuoto} (cio\`e avente liste vuote alla \emph{destra} ed alla \emph{sinistra} del tornello) \emph{afferma il falso} (se in un sistema logico-formale si \`e in grado di \emph{dimostrare il falso}, allora esso \`e \emph{inconsistente}).

\section{Meta-livello} 


In questa sezione introdurremo i principali strumenti meta-linguistici: in genere nei testi di logica questo aspetto \`e lasciato perlopi\`u ad un livello intuitivo. Ho cercato, con questa presentazione originale, di renderli pi\`u formali perch\'e, sebbene spesso sottovalutati, ritengo che offrano alcuni interessanti spunti di riflessione.

Osserviamo come, data una grammatica:
$$
	G = (\Sigma, \mathcal{C} = \{ S_1, \ldots, S_n \}, S_i, \mathcal{P})
$$
al variare di $i \in \{ 1, \ldots, n \}$ si producano linguaggi $\mathscr{L}_i$ diversi, seppur correlati tra loro, in funzione di quale categoria sintattica scegliamo come simbolo iniziale.

\begin{dfn}[Meta-variabili, meta-linguaggi, (sotto)formule]
Per ogni categoria sintattica $S_i$, definiamo un insieme finito $\mathscr{M}_i$ di \emph{meta-variabili} di categoria $S_i$, che sono dei ``segnaposti'' per una qualche stringa di $\mathscr{L}_i$. Grazie alle meta-variabili, possiamo imporre dei vincoli sulla forma delle stringhe di $\mathscr{L}_G$.

Il \emph{meta-livello linguistico} $\mathscr{L}_G'$ \`e quello in cui, giunti ad un certo passo di riscrittura, sostituiamo ad ogni occorrenza di simboli non terminali, una meta-variabile di categoria corrispondente. Il \emph{meta-alfabeto} \`e composto dai terminali e dalle meta-variabili:
$$
	\Sigma' = \Sigma{\smallsetminus}\mathcal{C} \cup \bigcup_{1 \le i \le n}\mathscr{M}_i
$$
mentre il \emph{meta-linguaggio} di $G$ \`e definito da:
$$
	\mathscr{L}_G' = \bigcup \{ \varphi(w) \:|\: S \leadsto^* w \}
$$
dove $\varphi(w)$ sta per ``qualunque sostituzione di metavariabili al posto dei simboli non terminali in $w$''. In simboli, se $a$ denota un terminale:
\begin{eqnarray*}
	\varphi & : & \Sigma^* \rightarrow \mathscr{P}(\Sigma'^*) \\
	\varphi(\epsilon) & = & \{ \epsilon \} \\
	\varphi(a w) & = & \{ a y \:|\: y \in \varphi(w) \} \\
	\varphi(S_i w) & = & \{ \sigma y \:|\: \sigma \in \mathscr{M}_i, y \in \varphi(w) \}
\end{eqnarray*}

Usiamo l'appellativo \emph{formula} per riferirci alle stringhe del meta-linguaggio $\mathscr{L}_G'$. Data una formula $w$ appartenente a $\mathscr{L}_G'$, per \emph{sottoformula} s'intende una qualunque porzione di $w$, che sia a sua volta compresa in $\mathscr{L}_G'$.
\end{dfn}

\`E immediato dimostrare che, data una grammatica $G$, il meta-linguaggio \`e pi\`u ricco del linguaggio, cio\`e che, per ogni $G$:
$$
	\mathscr{L}_G \subset \mathscr{L}_G'
$$
Infatti, al meta-linguaggio appartengono banalmente tutte le stringhe di $\mathscr{L}_G$ (se spingiamo la riscrittura fino a produrre stringhe di terminali, la funzione $\varphi$ non fa niente), mentre in $\mathscr{L}_G$ non ci sono le meta-variabili e quindi \`e strettamente incluso. 

Le meta-variabili si \emph{istanziano} a stringhe del linguaggio oggetto tramite \emph{unificazione}, grazie alla quale \`e anche possibile eseguire delle \emph{istanze parziali} tra meta-variabili e altre formule del meta-linguaggio.

Osserviamo che i linguaggi sono insiemi infiniti, cos\`i come il loro meta-livello genera insiemi infiniti. Tuttavia la base da cui sono generati questi insiemi (la grammatica) \`e finita e la procedura di generazione \`e concreta. Inoltre per grammatiche context-free verificare se una stringa appartiene o meno al linguaggio generato \`e un problema decidibile (in tempo polinomiale), e l'operazione di unificazione \`e anch'essa effettiva. Insomma, tutti gli strumenti dati fin qui sono \emph{finitari}, in pieno stile formalista.

Il meta-linguaggio ci permette di ragionare induttivamente (ricorsivamente) sulla struttura delle stringhe di un linguaggio. Normalmente l'induzione \`e concentrata sulla parte pi\`u esterna delle formule, cio\`e su quella di superficie: ma la metodologia deep inference aggiunge qualcosa in pi\`u.

\begin{dfn}[Contesti, saturazione, ordine]
Data una grammatica $G = (\Sigma, \mathcal{C}, S, \mathcal{P})$, il \emph{linguaggio dei contesti su} $G$, denotato con $\Xi_G$, \`e definito come il linguaggio generato dalla grammatica aumentata $(\Sigma \cup \{ \emptyctx \}, \mathcal{C}, S, \mathcal{P} \cup \{ (S, \emptyctx) \})$, dove $\emptyctx \not\in \Sigma$ \`e un nuovo simbolo terminale chiamato \emph{contesto vuoto}. Un \emph{contesto (generico)} \`e una stringa di $\mctx*{C} \in \Xi_G$ (si indica spesso con $\mctx{C}{\emptyctx}$ per enfatizzare il fatto che \`e un contesto).

Intuitivamente un contesto \`e una stringa di $\mathscr{L}_G$ con alcuni ``buchi'' (denotati da $\emptyctx$) che possono a loro volta essere riempiti con stringhe di $\mathscr{L}_G$. L'operazione di \emph{saturazione di un contesto} $\mctx*{C}$ \emph{con una stringa} $w \in \mathscr{L}_G$ si indica con $\mctx{C}{w}$, e consiste nella sostituzione testuale di $w$ al posto di tutte le occorrenze di $\emptyctx$ dentro $\mctx*{C}$; l'\emph{ordine di un contesto} (in simboli $\| \mctx*{C} \|$) \`e il numero di occorrenze di $\emptyctx$ al suo interno. Ambedue si definiscono formalmente per induzione sulla struttura di $\mctx*{C}$ come mostrato in Figura~\ref{fig:satordctx}.

\begin{figure}
\begin{center}
\begin{tabular}{lclclcl}
	\multicolumn{3}{c}{\textbf{Saturazione contesti}} & \hspace{5em} & \multicolumn{3}{c}{\textbf{Ordine contesti}} \\
	$\epsilon\{w\}$ & $=$ & $\epsilon$ & & $\| \epsilon \|$ & $=$ & $0$ \\
	$(\emptyctx y)\{w\}$ & $=$ & $w (y\{w\})$ & & $\| \emptyctx y \|$ & $=$ & $1 + \| y \|$ \\
	$(a y)\{w\}$ & $=$ & $a (y\{w\})$ & & $\| a y \|$ & $=$ & $\| y \|$
\end{tabular}
\end{center}
\caption{Definizioni induttive di \emph{saturazione} e \emph{ordine} di un contesto}
\label{fig:satordctx}
\end{figure}

Infine, sia $n \ge 0$ un numero naturale: il \emph{linguaggio dei contesti di ordine} $n$ \emph{su} $G$ (i.e. $\Xi_G^n$) \`e cos\`i definito:
$$
	\Xi_G^n = \{ \mctx*{C} \in \Xi_G \:|\: \| \mctx*{C} \| = n \}
$$
\end{dfn}

Data una grammatica $G = (\Sigma, \mathcal{C}, S, \mathcal{P})$, osserviamo che si ha $\Xi_G^0 = \mathscr{L}_G$, poich\'e per produrre le stringhe di $\Xi_G^0$ non si usa mai la regola di produzione aggiuntiva $(S, \emptyctx)$, ma solo quelle in $\mathcal{P}$, esattamente come accade per $\mathscr{L}_G$. Usando un argomento, analogo \`e possibile dimostrare che:
$$
	\mathscr{L}_G = \{ \mctx{C}{w} \:|\: \mctx*{C} \in \Xi_G^1, w \in \mathscr{L}_G \}
$$
Infatti i contesti di $\Xi_G^1$ sono generati usando una (e una sola) volta la regola di produzione aggiuntiva $(S, \emptyctx)$; in $G$, dove questa regola non \`e presente, tutto quello che \`e possibile fare \`e riscrivere $S$ con una delle altre produzioni di $\mathcal{P}$ per $S$, che equivale a sostituire \emph{quella} occorrenza di $S$ con una delle stringhe del linguaggio generato da $G$, cio\`e proprio $\mathscr{L}_G$.

Per $n > 1$ non \`e possibile ottenere risultati analoghi su $\Xi_G^n$, poich\'e se da un lato \`e possibile riscrivere occorrenze diverse di $S$ in modi diversi, dall'altro la saturazione di un contesto ammette un solo parametro in $\mathscr{L}_G$ (che viene replicato sempre uguale $n$ volte). Inoltre $n = 0$ \`e un caso triviale, perch\'e la saturazione dei contesti in $\Xi_G^0$ non produce effetti (non si fanno sostituzioni). L'unico caso degno di nota \`e $n = 1$: esso rappresenta il punto di contatto tra il concetto di saturazione di un contesto -- i.e. ``sostituzione di \emph{una} variabile (fresca)'' -- e quello pi\`u generale di riscrittura: saturazione di un contesto e riscrittura multipasso coincidono per $n \le 1$, cio\`e quando il processo di riscrittura \`e sostituito da quello di saturazione \emph{al pi\`u in un singolo punto}.

I contesti di ordine $1$ su una grammatica sono uno strumento molto potente che ci consente di focalizzare l'attenzione su una porzione specifica di una stringa del linguaggio che dipende dalla sua struttura sottostante, astraendoci dal resto. Per la loro rilevanza, d'ora in poi quando parleremo di \emph{contesti} intenderemo sempre quelli di ordine $1$. 

Anche i contesti sono strumenti di meta-livello, perch\'e trascendono il linguaggio oggetto, per permetterci di ragionare su esso. Per di pi\`u, un contesto pu\`o essere saturato con una qualunque formula di meta-livello: considerare arbitrari contesti ci permette di ragionare su classi di formule molto estese, ossia su formule \emph{immerse} in contesti arbitrari, cio\`e ad \emph{arbitrari livelli di profondit\`a}, concetto cardine di tutta la deep inference. Usando i contesti e il meta-linguaggio, possiamo ragionare per induzione strutturale sulla stringhe del linguaggio \emph{a qualsiasi livello di profondit\`a}. Inoltre, anche i contesti si possono \emph{unificare} con una procedura effettiva.

Al meta-livello ragioniamo su sequenti definiti sul meta-linguaggio. Le definizioni e le tecniche viste in precedenza per singole formule, si estendono in maniera naturale alle liste di formule e ai sequenti: in particolare, \`e possibile l'\emph{unificazione di sequenti} e considerare sequenti composti da formule \emph{immerse in contesti arbitrari} (con procedure effettive).

\section{Sistemi formali e formalismi} 

Le dimostrazioni sono l'oggetto di studio della proof theory; in questa sezione esplicitiamo la nozione di dimostrazione. Il termine ``dimostrare'' deriva dal latino \emph{demonstrare}, composto dalla radice \emph{de-} (di valore intensivo) e da \emph{monstrare} (``mostrare'', ``far vedere''), da cui il significato di \emph{rendere manifesto con fatti e con segni certi}. In matematica una dimostrazione \`e un \emph{processo di deduzione} che, partendo da \emph{premesse} assunte come valide (ipotesi) o da proposizioni dimostrate in virt\`u di tali premesse, determina la necessaria validit\`a di una nuova \emph{proposizione} in funzione della (sola) \emph{coerenza formale} del ragionamento. Le proposizioni saranno dunque stringhe appartenenti ad un linguaggio formale; il processo di deduzione sar\`a scandito dalla corretta applicazione di alcune regole di base in qualche modo riconosciute come elementari e la coerenza formale dovr\`a essere opportunamente formalizzata e funger\`a da argomento a sostegno della bont\`a delle regole scelte.

\begin{dfn}[Regole, derivazioni, dimostrazioni]
Un \emph{sistema formale} \`e una coppia $(\mathscr{L}, \mathscr{S})$ composta da un \emph{linguaggio} $\mathscr{L}$ (generato da qualche grammatica $G$, tipicamente -- ma non necessariamente -- context-free) e da un insieme di regole d'inferenza (o \emph{sistema di deduzione}) $\mathscr{S}$. Date le formule $P_1, \ldots, P_n, Q \in \mathscr{L}'$, una \emph{regola di inferenza} $\mrule{\rho}$ \`e un'espressione della forma:
\begin{center}
	\AxiomC{$P_1$}
	\AxiomC{$\cdots$}
	\AxiomC{$P_n$}
	\RightLabel{$\mrule{\rho}$}
	\TrinaryInfC{$Q$}
	\DisplayProof{}
\end{center}
dove $P_1, \ldots, P_n$ sono chiamate \emph{premesse della regola $\mrule{\rho}$} mentre $Q$ ne \`e la \emph{conclusione}. Una regola di inferenza senza premesse (i.e. avente $n=0$) \`e chiamata \emph{assioma}, mentre, per $n>0$, \`e detta \emph{regola d'inferenza propria}. In genere le premesse e la conclusione di $\mrule{\rho}$ sono formule (o sequenti) aventi una certa forma di superficie -- ed eventualmente, nell'approccio deep inference, immerse in arbitrari contesti. A questo modo di procedere, cio\`e di specificare la \emph{forma} delle (eventuali) ipotesi e della conclusione delle regole d'inferenza, ci si riferisce spesso in letteratura col termine \emph{schema} (p.e. dicendo ``schema d'assioma'').

Un \emph{passo d'inferenza} o \emph{applicazione} o \emph{istanza} di una regola d'inferenza $\mrule{\rho}$ \`e un'espressione della forma:
\begin{center}
	\AxiomC{$P_1'$}
	\AxiomC{$\cdots$}
	\AxiomC{$P_n'$}
	\RightLabel{$\mrule{\rho}$}
	\TrinaryInfC{$Q'$}
	\DisplayProof{}
\end{center}
dove $P_1', \ldots, P_n', Q' \in \mathscr{L}'$ sono formule ottenute rispettivamente per unificazione (anche parziale, al meta-livello) con $P_1, \ldots, P_n, Q \in \mathscr{L}'$. Le stringhe $P_1', \ldots, P_n'$ sono chiamate \emph{premesse dell'applicazione di $\mrule{\rho}$} mentre $Q'$ ne \`e la \emph{conclusione}. Indicheremo anche il nome della regola accanto alla barra orizzontale di derivazione, quando questo sar\`a d'aiuto alla comprensione e non sar\`a fonte d'ambiguit\`a.

Una \emph{derivazione} $\Phi$ da una lista di premesse $P_1, \ldots, P_n$ ad una conclusione $Q$ \`e un albero di istanze di regole in $\mathscr{S}$, avente $Q$ come radice e $P_1, \ldots, P_n$ come foglie, e indicato con:
$$
	\vltreeder{\Phi, \mathscr{S}}{Q}{P_1}{\cdots}{P_n}
$$
Nel seguito ometteremo $\Phi$ e/o $\mathscr{S}$ quando questo non comporter\`a ambiguit\`a.

Infine, una \emph{dimostrazione} \`e una derivazione avente per come premesse $P_1, \ldots, P_n$ solo istanze di assiomi. La indicheremo con:
$$
	\vlderivation{\vlpd{\Phi}{\mathscr{S}}{Q}}
$$
omettendo $\Phi$ e/o $\mathscr{S}$ quando e se necessario.
\end{dfn}

La \emph{derivabilit\`a} in un sistema formale (``da un insieme di formule $\Gamma$ \`e possibile derivare la formula $P$'', o anche ``$P$ \`e derivabile da $\Gamma$'') \`e un concetto sintattico, cos\`i come lo \`e la \emph{dimostrabilit\`a} -- in contrapposizione al concetto di \emph{verit\`a} e a quello di \emph{modello}, che sono invece di natura semantica. Finora non abbiamo mai parlato di verit\`a: in questa sede non ci occuperemo degli aspetti semantici legati ai sistemi formali, che non sono oggetto di studio di proof theory, rimandando per questi ad~\cite{Abr92, Bar77, Cha73}.

Il concetto di derivabilit\`a si estende anche alle regole dei sistemi formali: un regola \`e derivabile quando \`e ottenibile tramite altre regole. Ma una regola pu\`o anche essere \emph{ammissibile} (o \emph{eliminabile}): questo accade quando la sua presenza all'interno del sistema non altera l'insieme di formule dimostrabili, ossia eliminando la regola dal sistema, si riescono a dimostrare \emph{le stesse cose}. Questo vale anche per le regole derivabili, ma mentre in quel caso era sufficiente sostituire la regola con la sua derivazione, per regole ammissibili occorre ristrutturare l'albero di prova.

\begin{dfn}[Regole derivabili e ammissibili]
Una regola $\mrule{\rho}$ \`e \emph{derivabile} per un sistema $\mathscr{S}$ se, per ogni istanza di $\mrule{\rho}$:
$$
	\AxiomC{$P_1$}
	\AxiomC{$\cdots$}
	\AxiomC{$P_n$}
	\RightLabel{$\mrule{\rho}$}
	\TrinaryInfC{$Q$}
	\DisplayProof{}
$$
esiste una derivazione:
$$
	\vltreeder{\mathscr{S}}{Q}{P_1}{\cdots}{P_n}
$$

Una regola $\mrule{\rho}$ \`e \emph{ammissibile} (o \emph{eliminabile}) per un sistema $\mathscr{S}$ se, per ogni dimostrazione $\vlderivation{\vlpd{}{\mathscr{S}\cup\{\mrule{\rho}\}}{Q}}$ esiste una dimostrazione $\vlderivation{\vlpd{}{\mathscr{S}}{Q}}$.
\end{dfn}

In teoria della dimostrazione ci concentriamo (al meta-livello) sulle propriet\`a dei sistemi formali, cio\`e le propriet\`a di cui godono le dimostrazioni espresse in qualche sistema formale. Ma \emph{quale} sistema formale? Seguendo~\cite{TroSch96}, possiamo raggrupparli in alcune grandi famiglie chiamate \emph{formalismi}:
\begin{itemize}
	\item \emph{Sistemi assiomatici} o \emph{sistemi alla Frege-Hilbert}~(\cite{HilAck28, Fre79}): in questo approccio accettiamo un numero molto ristretto di regole d'inferenza proprie (p.e. nella logica proposizionale solo una, il \emph{modus ponens}) mentre il resto del sistema deduttivo sar\`a composto da assiomi; le derivazioni in questo formalismo furono originariamente concepite per rispecchiare le dimostrazioni espresse in linguaggio naturale, obiettivo ambizioso e scarsamente raggiunto da questo approccio in cui le dimostrazioni tendono invece ad essere molto dettagliate e ``pedanti''. \`E l'approccio pi\`u datato e grazie ad esso \`e stato possibile formalizzare con successo sistemi deduttivi rilevanti quali: la logica classica, quella intuizionista ed alcune logiche modali;
	\item \emph{Deduzione naturale}: introdotta nel celebre~\cite{Gen35} (assieme, come vedremo, al \emph{calcolo dei sequenti}) questa famiglia di sistemi \`e pensata per mimare il ragionamento logico-deduttivo umano (da qui l'aggettivo ``naturale''). Non ci sono assiomi e le formule possono essere composte e decomposte (usando il gergo tecnico, rispettivamente \emph{introdotte} ed \emph{eliminate}).

	Un'operazione comune nella pratica matematica \`e quella di \emph{ragionare per assunzioni}: la deduzione naturale mette a disposizione un artificio per compiere questa operazione, e i sistemi espressi in deduzione naturale godono di buone propriet\`a, relativamente semplici da dimostrare (vedi il classico~\cite{Pra65});

	\item \emph{Calcolo dei sequenti}: dovuto a Gentzen, questo \`e lo strumento preferito in teoria della dimostrazione per le ottime propriet\`a di cui gode. Fa tipicamente uso di pochi assiomi (p.e. nella logica proposizionale solo uno, l'assioma \emph{identit\`a}) e molte regole d'inferenza proprie, che permettono di comporre nuove formule a partire dalle premesse (usando la terminologia della deduzione naturale, sono presenti le regole di \emph{introduzione} ma non quelle di \emph{eliminazione}).
\end{itemize}

Volendo aderire ad una visione ``proof theoretical'', ci concentreremo in seguito sul calcolo dei sequenti. \`E tuttavia doveroso osservare che questi formalismi consentono di definire sistemi formali aventi il medesimo potere espressivo: in altre parole, nessuno prevale a priori sugli altri, dipende dal \emph{setting} in cui ci poniamo. 

Inoltre questi sono solo i formalismi ``classici''; esistono altre famiglie di sistemi formali, che possono essere usate per mettere in evidenza altri aspetti importanti del processo deduttivo e delle dimostrazioni. Tra questi \`e doveroso citare le \emph{Proof Nets}, introdotte in~\cite{Gir87} allo scopo di far emergere alcune simmetrie dei sistemi formali che erano ``oscurate'' dal calcolo dei sequenti.

Le tre famiglie sopra descritte adottano tutte shallow inference come specifica delle regole d'inferenza (per quanto questa sia una scelta del tutto arbitraria). Ma l'approccio deep inference apre la via ad (almeno) un quarto formalismo:
\begin{itemize}
	\item \emph{Calcolo delle strutture}: introdotto in~\cite{Gug02}, i sistemi deduttivi in calcolo delle strutture constano di un piccolo numero di assiomi e di regole d'inferenza proprie, e godono di una notevole quantit\`a di propriet\`a, sostanzialmente estendendo quelle studiate per il calcolo dei sequenti. Le nuove prospettive aperte nell'ambito del calcolo delle strutture, ne fanno uno strumento di grande interesse e in continuo sviluppo da parte della comunit\`a scientifica.
\end{itemize}

\newpage

\section{Metodologie: shallow \emph{versus} deep inference} 

Le metodologie guidano la progettazione dei sistemi deduttivi e il processo di inferenza: le due metodologie conosciute allo stato dell'arte sono \emph{shallow} e \emph{deep inference} e le esamineremo a turno. 

\begin{dfn}[Shallow inference]
Per \emph{shallow inference} o \emph{inferenza di superficie}, intendiamo la \emph{metodologia} d'inferenza che interpreta l'insieme delle regole d'inferenza come \emph{schemi} che disciplinano il comportamento della deduzione \emph{in funzione del connettivo principale} delle formule.
\end{dfn}

\begin{figure}[tbhp]
\begin{minipage}[t]{.64\textwidth}
	\begin{center}
	\textbf{Regole d'inferenza}
	$$
		P \vdash P \quad \mrule{ax}
	$$
	\begin{tabular}{cc}
		\AxiomC{$\Gamma, P \vdash R$}
		\RightLabel{$\mrule[l.1]{\wedge}$}
		\UnaryInfC{$\Gamma, P \wedge Q \vdash R$}
		\DisplayProof{} &
		\AxiomC{$\Gamma, Q \vdash R$}
		\RightLabel{$\mrule[l.2]{\wedge}$}
		\UnaryInfC{$\Gamma, P \wedge Q \vdash R$}
		\DisplayProof{} \\\\
		\AxiomC{$\Gamma, P \vdash Q$} %
		\RightLabel{$\mrule[r]{\rightarrow}$} %
		\UnaryInfC{$\Gamma \vdash P \rightarrow Q$} %
		\DisplayProof{} &
		\AxiomC{$\Gamma \vdash P$} %
		\AxiomC{$\Gamma \vdash Q$} %
		\RightLabel{$\mrule[r]{\wedge}$} %
		\BinaryInfC{$\Gamma \vdash P \wedge Q$} %
		\DisplayProof{} \\\\
		\multicolumn{2}{c}{ %
			\AxiomC{$\Gamma \vdash P$} %
			\AxiomC{$\Gamma, Q \vdash R$} %
			\RightLabel{$\mrule[l]{\rightarrow}$} %
			\BinaryInfC{$\Gamma, P \rightarrow Q \vdash R$} %
			\DisplayProof{}}
	\end{tabular}
	\end{center}
\end{minipage}
\begin{minipage}[t]{.35\textwidth}
	\begin{center}
	\textbf{Gramm. linguaggio}
	$$
		P ::= a \:|\: P \wedge P \:|\: P \rightarrow P
	$$
	\footnotesize{(con $a \in \mathcal{A}$ infinit\`a numerabile di simboli proposizionali)}
	\end{center}
	\begin{center}
	\textbf{Struttura formula}
	\begin{picture}(120,85)
		\thinlines
		\put(58,71){\makebox(0,0){$\rightarrow$}}
		\put(28,44){\line(4,3){30}}
		\put(25,38){\makebox(0,0){$\wedge$}}
		\put(25,32){\line(-1,-1){18}}
		\put(25,32){\line(1,-1){18}}
		\put(4,7){\makebox(0,0){$a$}}
		\put(46,7){\makebox(0,3){$b$}}
		\put(88,44){\line(-4,3){30}}
		\put(91,38){\makebox(0,0){$\wedge$}}
		\put(91,32){\line(-1,-1){18}}
		\put(91,32){\line(1,-1){18}}
		\put(70,7){\makebox(0,3){$b$}}
		\put(112,7){\makebox(0,0){$a$}}
	\end{picture}
	\end{center}
\end{minipage}
\caption{Sistema formale in shallow inference ed esempio di formula}
\label{fig:sfef}
\end{figure}

Essendo l'approccio pi\`u datato, \`e anche il pi\`u usato in letteratura, come in \emph{deduzione naturale} (i sistemi \textsf{NK} ed \textsf{NJ} usano shallow inference) e nel \emph{calcolo dei sequenti} (sistemi \textsf{LK}, \textsf{LJ}). Ad esempio: per derivare $\vdash (a \wedge b) \rightarrow (b \wedge a)$ con le regole d'inferenza in Figura~\ref{fig:sfef}, consideriamo la struttura della formula $(a \wedge b) \rightarrow (b \wedge a)$ ed osserviamo che il connettivo principale \`e $\rightarrow$. A questo punto l'unica regola applicabile (i.e. istanziabile, ricordiamo che le regole d'inferenza sono \emph{schemi}) in shallow inference \`e $\mrule[r]{\rightarrow}$. In questo modo otteniamo:
\begin{center}
	\AxiomC{$a \wedge b \vdash b \wedge a$}
	\UnaryInfC{$\vdash (a \wedge b) \rightarrow (b \wedge a)$}
	\DisplayProof{}
\end{center}
Procedendo in maniera analoga osserviamo che ci sono tre regole applicabili per derivare $a \wedge b \vdash b \wedge a$, cio\`e:
\begin{itemize}
	\item $\mrule[l.1]{\wedge}$ produce la derivazione:
	\begin{center}
		\AxiomC{$a \vdash b \wedge a$}
		\UnaryInfC{$a \wedge b \vdash b \wedge a$}
		\UnaryInfC{$\vdash (a \wedge b) \rightarrow (b \wedge a)$}
		\DisplayProof{}
	\end{center}
	da cui \`e applicabile solo $\mrule[r]{\wedge}$ che produce una derivazione bloccata (cio\`e un albero le cui foglie non sono istanze di assiomi, n\'e sono derivabili dalle regole del sistema);
	\item $\mrule[l.2]{\wedge}$ analogo al precedente;
	\item $\mrule[r]{\wedge}$ produce la derivazione:
	\begin{center}
		\alwaysNoLine
		\AxiomC{$\mrule{1}$}
		\UnaryInfC{$a \wedge b \vdash b$}
		\AxiomC{$\mrule{2}$}
		\UnaryInfC{$a \wedge b \vdash a$}
		\alwaysSingleLine
		\BinaryInfC{$a \wedge b \vdash b \wedge a$}
		\UnaryInfC{$\vdash (a \wedge b) \rightarrow (b \wedge a)$}
		\DisplayProof{}
	\end{center}
	in cui \`e possibile applicare $\mrule[l.1]{\wedge}$ o $\mrule[l.2]{\wedge}$ sia alla formula $\mrule{1}$ che alla $\mrule{2}$. L'unica combinazione che porta ad una conclusione -- cio\`e in cui ogni foglia \`e un'istanza di $\mrule{ax}$ -- \`e un'applicazione di $\mrule[l.2]{\wedge}$ a $\mrule{1}$ e di $\mrule[l.1]{\wedge}$ a $\mrule{2}$, ottenendo cos\`i:
	\begin{center}
		\AxiomC{$b \vdash b$}
		\UnaryInfC{$a \wedge b \vdash b$}
		\AxiomC{$a \vdash a$}
		\UnaryInfC{$a \wedge b \vdash a$}
		\BinaryInfC{$a \wedge b \vdash b \wedge a$}
		\UnaryInfC{$\vdash (a \wedge b) \rightarrow (b \wedge a)$}
		\DisplayProof{}
	\end{center}
\end{itemize}

La procedura descritta nell'esempio \`e nota come \emph{proof search} ed \`e automatizzabile (p.e. si pu\`o basare sulla risoluzione come avviene in \textsf{PROLOG}) per sistemi \emph{in cui la regola di taglio \`e ammissibile}.

\begin{dfn}[Deep inference] Per \emph{deep inference} o \emph{inferenza di profondit\`a} intendiamo la metodologia in cui le regole d'inferenza si possono applicare ad arbitrari contesti, e quindi ad arbitrari livelli di profondit\`a, in contrapposizione a quanto avviene nell'inferenza di superficie o \emph{shallow inference}. Il ruolo dei contesti \`e quello di permettere l'accesso alla struttura delle formule senza dover usare alberi di derivazione (cio\`e senza decomposizione strutturale delle formule). Per questa ragione le regole d'inferenza in deep inference hanno al pi\`u una premessa: pertanto le derivazioni prendono la forma di liste. Per enfatizzare il fatto che le derivazioni sono \emph{alberi degeneri} (i.e. ogni nodo ha al pi\`u un figlio), usiamo la notazione: 
$$
	\vlderd{\Phi}{\mathscr{S}}{Q}{P}
$$
per indicare una derivazione $\Phi$ che usa le regole in $\mathscr{S}$, e avente premessa $P$ e conclusione $Q$.
\end{dfn}

Dimostriamo l'analogo della formula di prima, usando il sistema deep inference in Figura~\ref{fig:sf_cos}. Invece dell'implicazione, qui abbiamo solo la negazione sugli atomi, quindi la formula di prima diventa: $(\mneg{a} \vee \mneg{b}) \vee (b \wedge a)$.

\begin{figure}[t!]
\begin{minipage}[t]{\textwidth}
	\begin{minipage}[t]{.61\textwidth}
		\begin{center}
		\textbf{Regole logiche}
		$$
			\mt \qquad \mrule{ax}
		$$
		\end{center}
		\begin{center}
		\begin{tabular}{ccc}
			$\vlinf{}{\mvlrule{id}}{\mctx{C}{a \vee \mneg{a}}}{\mctx{C}{\mt}}$
			& &
			$\vlinf{}{\mvlrule{s}}{\mctx{C}{(P \wedge Q) \vee R}}{\mctx{C}{P \wedge (Q \vee R)}}$ \\\\
		\end{tabular}
		\end{center}
	\end{minipage}
	\begin{minipage}[t]{.38\textwidth}
		\begin{center}
		\textbf{Grammatica linguaggio}
		$$
			P ::= \mt \:|\: \mf \:|\: a \:|\: \mneg{a} \:|\: P \vee P \:|\: P \wedge P
		$$
		\footnotesize{(con $a \in \mathcal{A}$ infinit\`a numerabile di simboli proposizionali)}
		\end{center}
	\end{minipage}
\end{minipage}
\begin{minipage}[t]{\textwidth}
	\begin{minipage}[t]{.61\textwidth}
		\begin{center}
		\textbf{Regole strutturali}
		\end{center}
		\begin{tabular}{ccc}
			$\vlinf{}{\mvlrule[\mt]{\wedge}}{\mctx{C}{P \wedge \mt}}{\mctx{C}{P}}$
			& &
			$\vlinf{}{\mvlrule[\mf]{\vee}}{\mctx{C}{P \vee \mf}}{\mctx{C}{P}}$ \\\\
			$\vlinf{}{\mvlrule[{com}]{\wedge}}{\mctx{C}{Q \wedge P}}{\mctx{C}{P \wedge Q}}$
			& &
			$\vlinf{}{\mvlrule[{as}]{\wedge}}{\mctx{C}{P \wedge (Q \wedge R)}}{\mctx{C}{(P \wedge Q) \wedge R}}$ \\\\
			$\vlinf{}{\mvlrule[{com}]{\vee}}{\mctx{C}{Q \vee P}}{\mctx{C}{P \vee Q}}$
			& &
			$\vlinf{}{\mvlrule[{as}]{\vee}}{\mctx{C}{P \vee (Q \vee R)}}{\mctx{C}{(P \vee Q) \vee R}}$
		\end{tabular}
	\end{minipage}
	\begin{minipage}[t]{.38\textwidth}
		\begin{center}
		\textbf{Dimostrazione d'esempio}
		\end{center}
		$$
			\qquad
			\vlderivation{
				\vlin{}{\mvlrule[{com}]{\vee}}
					{(\mneg{a} \vee \mneg{b}) \vee (b \wedge a)}{
					\vlin{}{\mvlrule[{as}]{\vee}}
						{(b \wedge a) \vee (\mneg{a} \vee \mneg{b})}{
						\vlin{}{\mvlrule{s}}
							{((b \wedge a) \vee \mneg{a}) \vee \mneg{b}}{
							\vlin{}{\mvlrule[{com}]{\wedge}}
								{(b \wedge (a \vee \mneg{a})) \vee \mneg{b}}{
								\vlin{}{\mvlrule{s}}
									{((a \vee \mneg{a}) \wedge b) \vee \mneg{b}}{
										\vlin{}{\mvlrule{id}}
										{(a \vee \mneg{a}) \wedge (b \vee \mneg{b})}{
											\vlin{}{\mvlrule{id}}
												{(a \vee \mneg{a}) \wedge \mt}{
												\vlin{}{\mvlrule[\mt]{\wedge}}
													{\mt \wedge \mt}
													{\vlhy{\mt}}}}}}}}}
			}
		$$
	\end{minipage}
\end{minipage}
\caption{Sistema formale in deep inference ed esempio di dimostrazione}
\label{fig:sf_cos}
\end{figure}

In questo caso abbiamo raggruppato le regole d'inferenza in ``regole logiche'' (simili a quelle viste prima) e ``regole strutturali''. Quest'ultime servono a formalizzare il fatto che i connettivi di congiunzione e di disgiunzione godono della propriet\`a commutativa -- rispettivamente regole $\mrule[com]{\wedge}$ e $\mrule[com]{\vee}$ -- e di quella associativa -- regole $\mrule[as]{\wedge}$ e $\mrule[as]{\vee}$ -- e che inoltre l'atomo $\mt$ (risp. $\mf$) \`e elemento neutro per il connettivo di congiunzione (risp. disgiunzione). A parte per la propriet\`a commutativa, che \`e intrinsecamente simmetrica, per le altre bisognerebbe specificare anche le regole opposte; ad esempio, per $\mrule[as]{\wedge}$ servirebbe una regola:
$$
	\vlderivation{
		\vlin{}{\mvlrule*[{as}]{\wedge}{-1}}
			{\mctx{C}{(P \wedge Q) \wedge R}}
			{\vlhy{\mctx{C}{P \wedge (Q \wedge R)}}}
	}
$$
Per questa ragione, in deep inference, si \`e soliti \emph{sostituire le regole strutturali con una relazione d'equivalenza} tra formule.

Come possiamo vedere, la dimostrazione della medesima formula di prima \`e completamente mutata: innanzitutto osserviamo che le regole induttive, nel calcolo delle strutture, hanno sempre e solo una premessa ed una conclusione. Questo fa s\`i che le dimostrazioni si sviluppino solo in altezza, collassando il lavoro strutturale svolto dagli alberi, all'interno dei contesti. Inoltre il sistema formale \`e meno rigido di quello visto nell'esempio precedente, nel senso che la metodologia deep inference permette un'applicazione pi\`u capillare delle regole d'inferenza, e quindi in generale un maggior grado di libert\`a e di non-determinismo.

\chapter{Logica classica proposizionale}
In questo capitolo presenteremo una serie di definizioni e risultati tradizionali in proof theory di logica classica proposizionale, e studieremo le propriet\`a formali del suo corrispettivo in deep inference, chiamato Sistema \textsf{SKS}.

Il Sistema che andremo a studiare (chiamato \textsf{LKp}) \`e il frammento proposizionale del Sistema \textsf{LK} di~\cite{Gen35}, il cui nome \`e l'acronimo di \emph{Logik Klassische} ossia ``logica classica'' (mentre la ``p'' sta appunto per proposizionale). Il linguaggio $\mathscr{L}_{\mathsf{LKp}}$ \`e quello dei sequenti in Definizione~\ref{def:seq}, aventi per formule quelle generate dalla grammatica:
$$
	G_{\mathsf{LKp}} = (\{ \neg, \wedge, \vee, \rightarrow, [a{-}z], S \}, \{ S \}, S, \mathcal{P})
$$
dove $[a{-}z]$ \`e una notazione abbreviata per i caratteri dell'alfabeto inglese, e le produzioni in $\mathcal{P}$, sono:
$$
	S ::= [a{-}z]{+} \:|\: \neg S \:|\: S \wedge S \:|\: S \vee S \:|\: S \rightarrow S
$$
dove $[a{-}z]{+}$ appartiene alla \emph{notazione EBNF} (\emph{BNF estesa}) e significa semplicemente ``qualunque stringa non vuota di caratteri alfabetici''. Solitamente s'introduce un'ulteriore generalizzazione, considerando la produzione $(S,a)$ dove $a$ \`e una meta-variabile appartenente ad un insieme $\mathcal{A}$ composto da un'infinit\`a numerabile di stringhe (alfabetiche, alfa-numeriche, indicizzate con apici, pedici, \ldots), chiamate genericamente ``simboli proposizionali''. Osserviamo che le stringhe di $\mathcal{A}$ sarebbero facilmente ottenibili da una grammatica context-free, questa semplificazione serve solo ad alleggerire la notazione, mentre preserva intatto il carattere finitario del linguaggio oggetto.

Il sistema deduttivo $\mathscr{S}_{\mathsf{LKp}}$ \`e dato usando forma di superficie dei sequenti. Le regole d'inferenza si possono dividere in quattro gruppi: assiomi, taglio, regole strutturali e regole logiche. 

Le regole strutturali sono di fondamentale importanza, perch\'e permettono di manipolare l'ordine ed il numero delle formule del sequente. Sono tre:
\begin{enumerate}
	\item L'\emph{ordine} delle premesse (e delle conclusioni) \emph{non \`e rilevante}. Da qui otteniamo le regole di \emph{permutazione}:
	\begin{center}
		\AxiomC{$\Gamma, P, Q, \Gamma' \vdash \Delta$}
		\RightLabel{$\mrule[l]{perm}$}
		\UnaryInfC{$\Gamma, Q, P, \Gamma' \vdash \Delta$}
		\DisplayProof{}
		\qquad
		\AxiomC{$\Gamma \vdash \Delta, P, Q, \Delta'$}
		\RightLabel{$\mrule[r]{perm}$}
		\UnaryInfC{$\Gamma \vdash \Delta, Q, P, \Delta'$}
		\DisplayProof{}
	\end{center}
	\item Assumere due volte la stessa premessa (o la stessa conclusione) \`e equivalente ad assumerla una volta sola. Questa osservazione ci conduce alle regole di \emph{contrazione}:
	\begin{center}
		\AxiomC{$\Gamma, P, P \vdash \Delta$}
		\RightLabel{$\mrule[l]{cont}$}
		\UnaryInfC{$\Gamma, P \vdash \Delta$}
		\DisplayProof{}
		\qquad
		\AxiomC{$\Gamma \vdash Q, Q, \Delta$}
		\RightLabel{$\mrule[r]{cont}$}
		\UnaryInfC{$\Gamma \vdash Q, \Delta$}
		\DisplayProof{}
	\end{center}
	\item \`E sempre possibile sia aggiungere nuove ipotesi (rafforzare l'antecedente), sia aggiungere nuove conclusioni (indebolire il conseguente). In generale il sequente ne risulter\`a indebolito (in un caso serve un'ipotesi in pi\`u affinch\'e funzioni, nell'altro a parit\`a di ipotesi dimostra una cosa pi\`u vaga, con pi\`u possibili conseguenze). Questa \`e pertanto chiamata regola di \emph{indebolimento}:
	\begin{center}
		\AxiomC{$\Gamma \vdash \Delta$}
		\RightLabel{$\mrule[l]{w}$}
		\UnaryInfC{$\Gamma, P \vdash \Delta$}
		\DisplayProof{}
		\qquad
		\AxiomC{$\Gamma \vdash \Delta$}
		\RightLabel{$\mrule[r]{w}$}
		\UnaryInfC{$\Gamma \vdash Q, \Delta$}
		\DisplayProof{}
	\end{center}
\end{enumerate}

Per individuare gli assiomi, ci poniamo la seguente domanda: quando si pu\`o sostenere che un sequente $\Gamma \vdash \Delta$ \`e \emph{evidentemente} vero? Chiaramente quando $\Gamma \cap \Delta \not= \varnothing$, cio\`e quando almeno una delle premesse in $\Gamma$ compare tra le conclusioni in $\Delta$. Questo sar\`a l'unico assioma del nostro Sistema, non ci sono altri criteri evidenti per passare dalle premesse alle conclusioni senza fare inferenza. In virt\`u delle regole strutturali, sappiamo che l'ordine non conta: pertanto dimostrare che esiste un $P \in \Gamma$ tale che $P \in \Delta$, si pu\`o scrivere: $P, \Gamma \vdash P, \Delta$. Inoltre possiamo sempre applicare la regola d'indebolimento a sinistra e a destra, per ottenere alfine:
$$
	P \vdash P \quad \mrule{ax}
$$

Quella di taglio \`e un'altra regola fondamentale, che ci si aspetta che sia soddisfatta da ogni sistema deduttivo. Il suo scopo \`e garantire la \emph{componibilit\`a} delle dimostrazioni; questa propriet\`a \`e ampiamente sfruttata nella pratica matematica: per provare un teorema complesso, si pu\`o cominciare dimostrando dei lemmi pi\`u semplici, che possono essere poi composti per ottenere il risultato cercato. La formulazione della \emph{regola di taglio} \`e pertanto la seguente:
\begin{center}
	\AxiomC{$\Gamma \vdash P, \Delta$}
	\AxiomC{$\Gamma', P \vdash \Delta'$}
	\RightLabel{$\mrule{cut}$}
	\BinaryInfC{$\Gamma, \Gamma' \vdash \Delta, \Delta'$}
	\DisplayProof{}
\end{center}

Infine le regole logiche sono quelle che specificano il comportamento dei connettivi logici. Come abbiamo gi\`a avuto modo di menzionare, nel calcolo dei sequenti \`e solo possibile \emph{introdurre} nuovi connettivi ma mai di \emph{eliminarli}. Questo fatto \`e alla base di una propriet\`a molto importante, detta \emph{della sottoformula} (vedi Definizione ~\ref{def:subformula}). Le regole d'introduzione dei connettivi saranno classificate in \emph{destre} (indicate con una ``$r$'' a pedice) o \emph{sinistre} (indicate con ``$l$'') a seconda che permettano d'introdurre il connettivo a destra oppure a sinistra del tornello. Vediamole rapidamente, ci sono quattro connettivi nel nostro linguaggio:
\begin{enumerate}
	\item \textbf{Congiunzione}: introdurre una congiunzione a sinistra significa rafforzare la premessa aggiungendo un'ipotesi. Come abbiamo detto in precedenza, il significato intuitivo del sequente va specificato formalmente, e queste regole chiarificano come le formule a sinistra del turnstile siano da considerarsi in congiunzione tra loro:
	\begin{center}
		\AxiomC{$\Gamma, P \vdash \Delta$}
		\RightLabel{$\mrule[l.1]{\wedge}$}
		\UnaryInfC{$\Gamma, P \wedge Q \vdash \Delta$}
		\DisplayProof{}
		\qquad
		\AxiomC{$\Gamma, Q \vdash \Delta$}
		\RightLabel{$\mrule[l.2]{\wedge}$}
		\UnaryInfC{$\Gamma, P \wedge Q \vdash \Delta$}
		\DisplayProof{}
	\end{center}
	A destra invece, cio\`e per concludere che una congiunzione $P \wedge Q$ vale, sotto un certo insieme di ipotesi, dobbiamo aver dimostrato separatamente i due rami $P$ e $Q$ a partire dalle stesse assunzioni, cio\`e:
	\begin{center}
		\AxiomC{$\Gamma \vdash P, \Delta$}
		\AxiomC{$\Gamma \vdash Q, \Delta$}
		\RightLabel{$\mrule[r]{\wedge}$}
		\BinaryInfC{$\Gamma \vdash P \wedge Q, \Delta$}
		\DisplayProof{}
	\end{center}

	\item \textbf{Disgiunzione}: il ragionamento e le regole seguono in maniera perfettamente simmetrica quanto visto per la congiunzione. Non c'\`e da sorprendersi, poich\'e la disgiunzione \`e il connettivo duale alla congiunzione, e poich\'e il sequente \`e fatto in modo da rispettare naturalmente tale simmetria. A destra del tornello, le formule sono da considerarsi in disgiunzione tra loro, e quindi abbiamo:
	\begin{center}
		\AxiomC{$\Gamma \vdash P, \Delta$}
		\RightLabel{$\mrule[r.1]{\vee}$}
		\UnaryInfC{$\Gamma \vdash P \vee Q, \Delta$}
		\DisplayProof{}
		\qquad
		\AxiomC{$\Gamma \vdash Q, \Delta$}
		\RightLabel{$\mrule[r.2]{\vee}$}
		\UnaryInfC{$\Gamma \vdash P \vee Q, \Delta$}
		\DisplayProof{}
	\end{center}
	mentre a sinistra, se da un set comune di ipotesi $\Gamma$ unito ad un'ipotesi $P$ riusciamo a concludere che valgono certe conclusioni $\Delta$, e indipendentemente, dallo stesso set di premesse $\Gamma$ unito stavolta ad una formula $Q$, siamo in grado di concludere le medesime conclusioni $\Delta$, allora da $\Gamma$ e $P \vee Q$ possiamo concludere che vale $\Delta$, cio\`e:
	\begin{center}
		\AxiomC{$\Gamma, P \vdash \Delta$}
		\AxiomC{$\Gamma, Q \vdash \Delta$}
		\RightLabel{$\mrule[l]{\vee}$}
		\BinaryInfC{$\Gamma, P \vee Q \vdash \Delta$}
		\DisplayProof{}
	\end{center}

	\item \textbf{Implicazione}: se la virgola a sinistra e a destra del tornello si comportano rispettivamente come una congiunzione e come una disgiunzione, il tornello stesso \`e l'implicazione. Questo fatto \`e reso evidente dalla regola d'introduzione destra della freccia. Infatti, la regola \`e:
	\begin{center}
		\AxiomC{$\Gamma, P \vdash Q, \Delta$}
		\RightLabel{$\mrule[r]{\rightarrow}$}
		\UnaryInfC{$\Gamma \vdash P \rightarrow Q, \Delta$}
		\DisplayProof{}
	\end{center}
	cio\`e afferma che se otteniamo una certa conclusione $Q$ supponendo $\Gamma$ e $P$, con solo $\Gamma$ \`e possibile concludere che ``se vale $P$ allora $Q$'', cio\`e proprio $P \rightarrow Q$. Le altre conclusioni in $\Delta$ non giocano alcun ruolo intuitivo per questa regola, se non di preservare una certa omogeneit\`a nella forma del sequente. Abbiamo visto come una formula $P$ sia in grado di passare da sinistra a destra del tornello, tramutandosi in un'implicazione. Anche il passaggio inverso \`e possibile:
	\begin{center}
		\AxiomC{$\Gamma \vdash P, \Delta$}
		\AxiomC{$\Gamma, Q \vdash \Delta$}
		\RightLabel{$\mrule[l]{\rightarrow}$}
		\BinaryInfC{$\Gamma, P \rightarrow Q \vdash \Delta$}
		\DisplayProof{}
	\end{center}
	Qui si \`e dimostrato che da $\Gamma$ si deriva $P$ e, indipendentemente, che da $\Gamma$ unito all'ipotesi aggiuntiva $Q$ si conclude $\Delta$. Allora da $\Gamma$ e supponendo che $P$ implichi $Q$ \`e possibile concludere $\Delta$.
	\item \textbf{Negazione}: in virt\`u di quanto visto finora, il comportamento della negazione dovrebbe risultare piuttosto semplice. Infatti, se consideriamo una singola formula, passare da una parte all'altra del tornello significa introdurre una negazione (la negazione di una formula \`e equivalente ad un'implicazione in cui dalla validit\`a della formula si conclude il falso). Formalmente:
	\begin{center}
		\AxiomC{$\Gamma \vdash P, \Delta$}
		\RightLabel{$\mrule[l]{\neg}$}
		\UnaryInfC{$\Gamma, \neg P \vdash \Delta$}
		\DisplayProof{}
		\qquad
		\AxiomC{$\Gamma, P \vdash \Delta$}
		\RightLabel{$\mrule[r]{\neg}$}
		\UnaryInfC{$\Gamma \vdash \neg P, \Delta$}
		\DisplayProof{}
	\end{center}
\end{enumerate}

Facciamo un esempio di \emph{regola derivabile} nel Sistema \textsf{LKp}: scriviamo la regola \emph{destra di congiunzione} e la regola di \emph{destra di congiunzione generalizzata}:
\begin{center}
	\AxiomC{$\Gamma \vdash P, \Delta$}
	\AxiomC{$\Gamma \vdash Q, \Delta$}
	\RightLabel{$\mrule[r]{\wedge}$}
	\BinaryInfC{$\Gamma \vdash P \wedge Q, \Delta$}
	\DisplayProof{}
	\qquad
	\AxiomC{$\Gamma \vdash P, \Delta$}
	\AxiomC{$\Gamma' \vdash Q, \Delta'$}
	\RightLabel{$\mrule*[r]{\wedge}{gen}$}
	\BinaryInfC{$\Gamma, \Gamma' \vdash P \wedge Q, \Delta, \Delta'$}
	\DisplayProof{}
\end{center}
\begin{itemize}
	\item $\mrule[r]{\wedge}$ \`e banalmente derivabile da $\mrule*[r]{\wedge}{gen}$, infatti basta porre $\Gamma'=\Gamma$ e $\Delta'=\Delta$ per ottenere:
	\begin{center}
		\AxiomC{$\Gamma \vdash P, \Delta$}
		\AxiomC{$\Gamma \vdash Q, \Delta$}
		\RightLabel{$\mrule*[r]{\wedge}{gen}$}
		\BinaryInfC{$\Gamma, \Gamma \vdash P \wedge Q, \Delta, \Delta$}
		\alwaysDoubleLine
		\UnaryInfC{$\Gamma \vdash P \wedge Q, \Delta$}
		\DisplayProof{}
	\end{center}
	dove la doppia barra orizzontale indica un certo numero applicazioni di regole strutturali, in questo caso \emph{permutazione} e \emph{contrazione}. D'ora in avanti useremo sempre questa convenzione.
	\item $\mrule*[r]{\wedge}{gen}$ \`e derivabile da $\mrule[r]{\wedge}$:
	\begin{center}
		\AxiomC{$\Gamma \vdash P, \Delta$}
		\alwaysDoubleLine
		\UnaryInfC{$\Gamma, \Gamma' \vdash P, \Delta, \Delta'$}
		\AxiomC{$\Gamma' \vdash Q, \Delta'$}
		\UnaryInfC{$\Gamma, \Gamma' \vdash Q, \Delta, \Delta'$}
		\RightLabel{$\mrule[r]{\wedge}$}
		\alwaysSingleLine
		\BinaryInfC{$\Gamma, \Gamma' \vdash P \wedge Q, \Delta, \Delta'$}
		\DisplayProof{}
	\end{center}
\end{itemize}

Ragionamenti analoghi valgono per la regole sinistre di disgiunzione e implicazione (generalizzate). Per quanto riguarda le regole \emph{ammissibili}, avremo modo nel seguito di dimostrare \emph{l'ammissibilit\`a della regola di taglio} in \textsf{LKp}.

\begin{dfn}[Propriet\`a della sottoformula]\label{def:subformula}
Si dice che \emph{una regola d'inferenza $\mrule{\rho}$ gode della propriet\`a della sottoformula} sse per ogni sua istanza:
\begin{center}
	\AxiomC{$P_1$}
	\AxiomC{$\cdots$}
	\AxiomC{$P_n$}
	\TrinaryInfC{$Q$}
	\DisplayProof{} 
\end{center}
si ha che $P_1, \ldots, P_n$ sono sottoformule di $Q$. Questa definizione si estende naturalmente alle regole del calcolo dei sequenti, imponendo che tutte le formule nelle premesse (sia a destra che a sinistra del turnstile) siano sottoformule di quelle presente nel sequente conclusione.

Inoltre si dice che \emph{un sistema formale gode della propriet\`a della sottoformula} quando tutte le sue regole d'inferenza ne godono.
\end{dfn}

La propriet\`a della sottoformula \`e molto rilevante, perch\'e conferisce ai sistemi una natura ``costruttiva'', il che ha molte importanti ripercussioni sulla meccanizzazione del processo inferenziale e sulla proof search. Un risultato classico \`e il seguente:

\begin{thm}[Consistenza]
Sia dato un sistema formale, espresso mediante il calcolo dei sequenti, non triviale (che non contiene il sequente vuoto tra gli assiomi). Allora, se gode della propriet\`a della sottoformula, esso \`e consistente (cio\`e non permette di derivare il sequente vuoto).
\end{thm}
\begin{proof}
La dimostrazione \`e immediata, poich\'e, se il sequente vuoto non \`e fra gli assiomi del sistema, dev'essere derivato con una regola d'inferenza propria $\mrule{\rho}$. Ma per la propriet\`a della sottoformula, la regola d'inferenza $\mrule{\rho}$ pu\`o avere per premesse solo sottoformule di quelle nel sequente vuoto, cio\`e non pu\`o avere premesse, ma $\mrule{\rho}$ \`e propria per ipotesi. Pertanto il sequente vuoto non \`e derivabile ed il sistema \`e consistente.
\end{proof}

Da una rapida ispezione alle regole del Sistema \textsf{LKp}, ci accorgiamo che la propriet\`a della sottoformula vale per tutte le regole d'inferenza tranne che per la regola di taglio. Infatti $\mrule{cut}$ introduce un'arbitraria formula $P$ tra le premesse. Onde preservare la propriet\`a della sottoformula, seguiamo i passi di Gentzen, dimostrando uno dei teoremi centrali della proof theory, noto come ``Gentzen Hauptsatz'', che ci garantisce che la regola di taglio \`e ammissibile all'interno del Sistema.

\section{Eliminazione del taglio}

Ci accingiamo a dimostrare una propriet\`a essenziale per la logica classica (e non solo), chiamata Hauptsatz, o teorema di eliminazione del taglio. L'Hauptsatz presumibilmente traccia il confine tra la logica e la nozione pi\`u ampia di sistema formale. Per sottolinearne l'importanza, Girard usa il motto:
\begin{quote}
	\emph{``A sequent calculus without cut elimination is like a car without engine''}~--~\cite{Gir95}
\end{quote}

\begin{dfn}[Grado, altezza derivazioni]
Il \emph{grado di una formula} $\delta(P)$ \`e definito per induzione strutturale come segue:
\begin{itemize}
	\item $\delta(a) = 1$ \qquad per $a$ simbolo proposizionale
	\item $\delta(P \wedge Q) = \delta(P \vee Q) = \delta(P \rightarrow Q) = 1 + \max\{\delta(P), \delta(Q)\}$
	\item $\delta(\neg P) = 1 + \delta(P)$
\end{itemize}

Il \emph{grado di un'applicazione della regola di taglio} \`e definito come il grado della formula che elimina.

Il \emph{grado} $\delta(\Phi)$ \emph{di una derivazione} \`e il massimo tra i gradi delle regole di taglio che vi compaiono. In particolare $\delta(\Phi) = 0$ se $\Phi$ non fa uso della regola di taglio.

Infine, l'\emph{altezza} $h(\Phi)$ \emph{di una derivazione} \`e quella associata all'albero $\Phi$: se la regola conclusiva di $\Phi$ ha come premesse le derivazioni $\Phi_1, \ldots, \Phi_n$, allora $h(\Phi) = 1 + \max\{h(\Phi_1), \ldots, h(\Phi_n)\}$ (mentre se $n = 0$, cio\`e se $\Phi$ \`e istanza di un assioma, allora $h(\Phi) = 0$).
\end{dfn}

\begin{lem}~\label{lem:cut_LK1}
Sia $\Phi$ una derivazione della forma seguente:
\begin{center}
	\AxiomC{$P_1$}
	\AxiomC{$\cdots$}
	\AxiomC{$P_n$}
	\RightLabel{$\mrule[l]{\rho}$}
	\TrinaryInfC{$\Gamma \vdash P, \Delta$}
	\AxiomC{$P_{n+1}$}
	\AxiomC{$\cdots$}
	\AxiomC{$P_{n+m}$}
	\RightLabel{$\mrule[r]{\rho}$}
	\TrinaryInfC{$\Gamma', P \vdash \Delta'$}
	\RightLabel{$\mrule{cut}$}
	\BinaryInfC{$\Gamma, \Gamma' \vdash \Delta, \Delta'$}
	\DisplayProof{}
\end{center}
in cui $\mrule[l]{\rho}$ (la premessa di sinistra del cut) \`e una regola logica ``destra'', mentre $\mrule[r]{\rho}$ (premessa destra del cut) \`e una regola logica ``sinistra'', tali da introdurre entrambe la formula $P$. Allora esiste una derivazione:
$$
	\vltreeder{\Psi}{\Gamma, \Gamma' \vdash \Delta, \Delta'}{P_1'}{\cdots}{P_k'}
$$
con $\{P_1', \ldots, P_k'\} \subseteq \{P_1, \ldots, P_{n+m}\}$ e tale che $\delta(\Psi) < \delta(\Phi)$.
\end{lem}
\begin{proof}
Procediamo per casi sulla premessa sinistra dell'applicazione di $\mrule{cut}$. Il fatto di concentrarci sulla premessa sinistra \`e del tutto irrilevante, poich\'e grazie alla simmetria delle regole logiche del Sistema \textsf{LKp}, se nella premessa sinistra s'introduce un certo connettivo (con una regola logica ``destra''), questo dovr\`a essere introdotto anche nella premessa di destra (con una regola logica simmetrica ``sinistra'').
\begin{enumerate}
	\item\label{proof:cut_lem1:case:1} $\mrule[r]{\wedge}$ e $\mrule[l.1]{\wedge}$: qui abbiamo $P = Q \wedge R$.
	\begin{center}
		\AxiomC{$\Gamma \vdash Q, \Delta$}
		\AxiomC{$\Gamma \vdash R, \Delta$}
		\RightLabel{$\mrule[r]{\wedge}$}
		\BinaryInfC{$\Gamma \vdash Q \wedge R, \Delta$}
		\AxiomC{$\Gamma', Q \vdash \Delta'$}
		\RightLabel{$\mrule[l.1]{\wedge}$}
		\UnaryInfC{$\Gamma', Q \wedge R \vdash \Delta'$}
		\RightLabel{$\mrule{cut}$}
		\BinaryInfC{$\Gamma, \Gamma' \vdash \Delta, \Delta'$}
		\DisplayProof{}
	\end{center}
	La derivazione $\Phi$ sopra si trasforma in $\Psi$ come segue:
	\begin{center}
		\AxiomC{$\Gamma \vdash Q, \Delta$}
		\AxiomC{$\Gamma', Q \vdash \Delta'$}
		\RightLabel{$\mrule{cut}$}
		\BinaryInfC{$\Gamma, \Gamma' \vdash \Delta, \Delta'$}
		\DisplayProof{}
	\end{center}
	Osserviamo che $\delta(\Psi) = \delta(\Phi) - \delta(R)$, cio\`e il grado di $\Psi$ \`e diminuito di un fattore $\delta(R)>0$.

	\item\label{proof:cut_lem1:case:2} $\mrule[r]{\wedge}$ e $\mrule[l.2]{\wedge}$: in maniera simmetrica, qui abbiamo:
	\begin{center}
		\AxiomC{$\Gamma \vdash Q, \Delta$}
		\AxiomC{$\Gamma \vdash R, \Delta$}
		\RightLabel{$\mrule[r]{\wedge}$}
		\BinaryInfC{$\Gamma \vdash Q \wedge R, \Delta$}
		\AxiomC{$\Gamma', R \vdash \Delta'$}
		\RightLabel{$\mrule[l.2]{\wedge}$}
		\UnaryInfC{$\Gamma', Q \wedge R \vdash \Delta'$}
		\RightLabel{$\mrule{cut}$}
		\BinaryInfC{$\Gamma, \Gamma' \vdash \Delta, \Delta'$}
		\DisplayProof{}
	\end{center}
	che si trasforma nuovamente in $\Psi$ di grado inferiore:
	\begin{center}
		\AxiomC{$\Gamma \vdash R, \Delta$}
		\AxiomC{$\Gamma', R \vdash \Delta'$}
		\RightLabel{$\mrule{cut}$}
		\BinaryInfC{$\Gamma, \Gamma' \vdash \Delta, \Delta'$}
		\DisplayProof{}
	\end{center}

	\item $\mrule[r.1]{\vee}$ e $\mrule[l]{\vee}$: qui abbiamo $P = Q \vee R$. Questo \`e il duale del caso \ref{proof:cut_lem1:case:1}:
	\begin{center}
		\AxiomC{$\Gamma \vdash Q, \Delta$}
		\RightLabel{$\mrule[r.1]{\vee}$}
		\UnaryInfC{$\Gamma \vdash Q \vee R, \Delta$}
		\AxiomC{$\Gamma', Q \vdash \Delta'$}
		\AxiomC{$\Gamma', R \vdash \Delta'$}
		\RightLabel{$\mrule[l]{\vee}$}
		\BinaryInfC{$\Gamma', Q \vee R \vdash \Delta'$}
		\RightLabel{$\mrule{cut}$}
		\BinaryInfC{$\Gamma, \Gamma' \vdash \Delta, \Delta'$}
		\DisplayProof{}
	\end{center}
	che si trasforma in $\Psi$ come segue:
	\begin{center}
		\AxiomC{$\Gamma \vdash Q, \Delta$}
		\AxiomC{$\Gamma', Q \vdash \Delta'$}
		\RightLabel{$\mrule{cut}$}
		\BinaryInfC{$\Gamma, \Gamma' \vdash \Delta, \Delta'$}
		\DisplayProof{}
	\end{center}
	col grado di $\Psi$ diminuito di un fattore $\delta(R)$.
	\item $\mrule[r.2]{\vee}$ e $\mrule[l]{\vee}$: $P = Q \vee R$, caso simmetrico al precendente e duale a \ref{proof:cut_lem1:case:2}, produciamo una derivazione $\Psi$ avente grado pari a $\delta(\Phi) - \delta(Q)$, a partire da $\Phi$:
	\begin{center}
		\AxiomC{$\Gamma \vdash R, \Delta$}
		\RightLabel{$\mrule[r.2]{\vee}$}
		\UnaryInfC{$\Gamma \vdash Q \vee R, \Delta$}
		\AxiomC{$\Gamma', Q \vdash \Delta'$}
		\AxiomC{$\Gamma', R \vdash \Delta'$}
		\RightLabel{$\mrule[l]{\vee}$}
		\BinaryInfC{$\Gamma', Q \vee R \vdash \Delta'$}
		\RightLabel{$\mrule{cut}$}
		\BinaryInfC{$\Gamma, \Gamma' \vdash \Delta, \Delta'$}
		\DisplayProof{}
	\end{center}
	nel modo seguente:
	\begin{center}
		\AxiomC{$\Gamma \vdash R, \Delta$}
		\AxiomC{$\Gamma', R \vdash \Delta'$}
		\RightLabel{$\mrule{cut}$}
		\BinaryInfC{$\Gamma, \Gamma' \vdash \Delta, \Delta'$}
		\DisplayProof{}
	\end{center}

	\item $\mrule[r]{\neg}$ e $\mrule[l]{\neg}$: qui abbiamo $P = \neg Q$. La derivazione $\Phi$ pertanto \`e:
	\begin{center}
		\AxiomC{$\Gamma, Q \vdash \Delta$}
		\RightLabel{$\mrule[r]{\neg}$}
		\UnaryInfC{$\Gamma \vdash \neg Q, \Delta$}
		\AxiomC{$\Gamma'\vdash Q, \Delta'$}
		\RightLabel{$\mrule[l]{\neg}$}
		\UnaryInfC{$\Gamma', \neg Q \vdash \Delta'$}
		\RightLabel{$\mrule{cut}$}
		\BinaryInfC{$\Gamma, \Gamma' \vdash \Delta, \Delta'$}
		\DisplayProof{}
	\end{center}
	Costruiamo $\Psi$ scambiando le premesse di $\Phi$ e applicando direttamente il taglio, per ottenere una derivazione di grado $\delta(\Phi) - 1$, come segue:
	\begin{center}
		\AxiomC{$\Gamma'\vdash Q, \Delta'$}
		\AxiomC{$\Gamma, Q \vdash \Delta$}
		\RightLabel{$\mrule{cut}$}
		\BinaryInfC{$\Gamma', \Gamma \vdash \Delta', \Delta$}
		\alwaysDoubleLine
		\UnaryInfC{$\Gamma, \Gamma' \vdash \Delta, \Delta'$}
		\DisplayProof{}
	\end{center}

	\item $\mrule[r]{\rightarrow}$ e $\mrule[l]{\rightarrow}$: $P = Q \rightarrow R$. Allora $\Phi$:
	\begin{center}
		\AxiomC{$\Gamma, Q \vdash R, \Delta$}
		\RightLabel{$\mrule[r]{\rightarrow}$}
		\UnaryInfC{$\Gamma \vdash Q \rightarrow R, \Delta$}
		\AxiomC{$\Gamma' \vdash Q, \Delta'$}
		\AxiomC{$\Gamma', R \vdash \Delta'$}
		\RightLabel{$\mrule[l]{\rightarrow}$}
		\BinaryInfC{$\Gamma', Q \rightarrow R \vdash \Delta'$}
		\RightLabel{$\mrule{cut}$}
		\BinaryInfC{$\Gamma, \Gamma' \vdash \Delta, \Delta'$}
		\DisplayProof{}
	\end{center}
	si trasforma in $\Psi$ come segue:
	\begin{center}
		\AxiomC{$\Gamma' \vdash Q, \Delta'$}
		\AxiomC{$\Gamma, Q \vdash R, \Delta$}
		\RightLabel{$\mrule{cut}$}
		\BinaryInfC{$\Gamma', \Gamma \vdash \Delta', R, \Delta$}
		\alwaysDoubleLine
		\UnaryInfC{$\Gamma, \Gamma' \vdash R, \Delta, \Delta'$}
		\AxiomC{$\Gamma', R \vdash \Delta'$}
		\UnaryInfC{$\Gamma, \Gamma', R \vdash \Delta, \Delta'$}
		\alwaysSingleLine
		\RightLabel{$\mrule{cut}$}
		\BinaryInfC{$\Gamma, \Gamma' \vdash \Delta, \Delta'$}
		\DisplayProof{}
	\end{center}
	osserviamo che in quest'ultimo caso il problema \`e stato risolto usando \emph{due} tagli, entrambi di grado inferiore.
\end{enumerate}
\end{proof}

\begin{dfn}[Rimozione]
Sia $P$ una formula e $\Gamma$ una lista di formule: allora $\Gamma{\smallsetminus}P$ denota $\Gamma$ in cui \emph{tutte le occorrenze} della formula $P$ sono state \emph{rimosse}.
\end{dfn}

Il seguente lemma dice che una (eventuale) applicazione della regola di taglio finale pu\`o essere eliminata. La sua complessa formulazione tiene conto delle regole strutturali che possono interferire col taglio.

\begin{lem}\label{lem:cut_LK2}
Sia $P$ una formula di grado $d$, e siano $\Phi, \Phi'$ rispettivamente le dimostrazioni di $\Gamma \vdash \Delta$ e di $\Gamma' \vdash \Delta'$ ambedue di grado minore di $d$. Allora \`e possibile costruire una dimostrazione $\Psi$ di $\Gamma, \Gamma'{\smallsetminus}P \vdash \Delta{\smallsetminus}P, \Delta'$ di grado minore di $d$.
\end{lem}
\begin{proof}
$\Psi$ \`e costruito per induzione su $h(\Phi) + h(\Phi')$, ma sfortunatamente non in maniera simmetrica rispetto $\Phi$ e $\Phi'$: ad un certo punto la preferenza sar\`a data a $\Phi$ od a $\Phi'$, e $\Psi$ sar\`a irreversibilmente affetta da questa scelta.

Siano $\Phi$ e $\Phi'$ rispettivamente:
$$
\begin{array}{ccc}
	\vlderivation{
		\vliiin{}{\mrule{\rho}}
			{\Gamma \vdash \Delta}
			{\vlpd{\Phi_1}{}{\Gamma_1 \vdash \Delta_1}}
			{\vlhy{\cdots}}
			{\vlpd{\Phi_n}{}{\Gamma_n \vdash \Delta_n}}
	}
	& \qquad &
	\vlderivation{
		\vliiin{}{\mrule{\rho'}}
			{\Gamma' \vdash \Delta'}
			{\vlpd{\Phi_1'}{}{\Gamma_1' \vdash \Delta_1'}}
			{\vlhy{\cdots}}
			{\vlpd{\Phi_m'}{}{\Gamma_m' \vdash \Delta_m'}}
	}
\end{array}
$$
e siano $i \in \{1, \ldots, n\}$ e $j \in \{1, \ldots, m\}$. Ci sono vari casi da considerare:
\begin{enumerate}
	\item $\Phi$ \`e un assioma. Ci sono due sottocasi:
	\begin{enumerate}
		\item $\Phi$ prova $P \vdash P$. Allora la dimostrazione $\Psi$ di $P, \Gamma'{\smallsetminus}P \vdash \Delta'$ \`e ottenuta da $\Phi'$ mediante l'applicazione di regole strutturali.
		\item $\Phi$ prova $Q \vdash Q$, con $Q \not= P$. Anche in questo caso applichiamo regole strutturali a $\Phi'$ per ottenere $Q, \Gamma'{\smallsetminus}Q \vdash Q, \Delta'$.
	\end{enumerate}
	\item $\Phi'$ \`e un assioma. Questo caso \`e del tutto analogo al precedente; \`e interessante notare che se $\Phi$ e $\Phi'$ sono entrambi assiomi, abbiamo arbitrariamente privilegiato $\Phi$ (e questo potrebbe avere delle ripercussioni sulla complessit\`a di $\Psi$).
	\item $\mrule{\rho}$ \`e una regola strutturale. L'ipotesi induttiva per $\Phi_1$ e $\Phi'$ ci danno una dimostrazione $\Psi_1$ per $\Gamma_1, \Gamma'{\smallsetminus}P \vdash \Delta_1{\smallsetminus}P, \Delta'$. Allora $\Psi$ \`e ottenuto da $\Psi_1$ mediante regole strutturali. Questo \`e possibile perch\'e, qualunque sia la regola strutturale $\mrule{\rho}$, questa gode della propriet\`a della sottoformula, e quindi $\Gamma_1$ \`e composto esclusivamente di sottoformule di $\Gamma$, cos\`i come $\Delta_1{\smallsetminus}P$ \`e composto solo di sottoformule di $\Delta{\smallsetminus}P$. Quindi per ottenere il sequente conclusivo di $\Psi$, non dovr\`a essere tolta alcuna formula presente nella conclusione di $\Psi_1$, ma al massimo lo si dovr\`a \emph{indebolire}.
	\item $\mrule{\rho'}$ \`e una regola strutturale: analogo al precedente. 
	\item $\mrule{\rho}$ \`e una regola logica, tranne una regola logica destra che introduce $P$. L'ipotesi induttiva per $\Phi_i$ e $\Phi'$ ci da $n$ dimostrazioni $\Psi_i$ di $\Gamma_i, \Gamma'{\smallsetminus}P \vdash \Delta_i{\smallsetminus}P, \Delta'$. Poich\'e la regola $\mrule{\rho}$ non introduce nuove occorrenze di $P$ a destra del tornello, questa \`e applicabile alle $\Psi_i$ per ottenere $\Psi$: $\Gamma, \Gamma'{\smallsetminus}P \vdash \Delta{\smallsetminus}P, \Delta'$.
	\item $\mrule{\rho'}$ \`e una regola logica: analogo al precedente. 
	\item Sia $\mrule{\rho}$ che $\mrule{\rho'}$ sono regole logiche: $\mrule{\rho}$ \`e una regola logica destra che introduce $P$, mentre $\mrule{\rho'}$ \`e una regola logica sinistra che introduce $P$. Questo \`e l'ultimo caso rimanente, nonch\'e l'unico interessante, ed \`e simmetrico. Per ipotesi induttiva, applicata a:
	\begin{enumerate}
		\item $\Phi_i$ e $\Phi'$, otteniamo le dimostrazioni $\Psi_i$ di $\Gamma_i, \Gamma'{\smallsetminus}P \vdash \Delta_i{\smallsetminus}P, \Delta'$; ora, applicando $\mrule{\rho}$ alle $\Psi_i$, e usando delle regole strutturali, otteniamo la dimostrazione $\Upsilon$ di $\Gamma, \Gamma'{\smallsetminus}P \vdash P, \Delta{\smallsetminus}P, \Delta'$;
		\item $\Phi$ e $\Phi_j'$, otteniamo le dimostrazioni $\Psi_j'$ di $\Gamma, \Gamma_j'{\smallsetminus}P \vdash \Delta{\smallsetminus}P, \Delta_j'$; ora, applicando $\mrule{\rho'}$ alle $\Psi_j'$, e con l'ausilio di regole strutturali, otteniamo la dimostrazione $\Upsilon'$ di $\Gamma, \Gamma'{\smallsetminus}P, P \vdash \Delta{\smallsetminus}P, \Delta'$.
	\end{enumerate}
	Ora abbiamo due dimostrazioni, $\Upsilon$ e $\Upsilon'$, che si concludono come richiesto, se non per un'occorrenza di troppo della formula $P$. Applicando la regola di taglio ad $\Upsilon$ e $\Upsilon'$, otteniamo una dimostrazione $\Upsilon''$ di:
$$
	\vlderivation{
		\vliin{}{\mrule{cut}}
			{\Gamma, \Gamma'{\smallsetminus}P, \Gamma, \Gamma'{\smallsetminus}P \vdash \Delta{\smallsetminus}P, \Delta', \Delta{\smallsetminus}P, \Delta'}
			{\vlpd{\Upsilon}{}{\Gamma, \Gamma'{\smallsetminus}P \vdash P, \Delta{\smallsetminus}P, \Delta'}}
			{\vlpd{\Upsilon'}{}{\Gamma, \Gamma'{\smallsetminus}P, P \vdash \Delta{\smallsetminus}P, \Delta'}}
	}
$$
che con semplici manipolazioni strutturali \`e riducibile a: 
$$
	\Gamma, \Gamma'{\smallsetminus}P \vdash \Delta{\smallsetminus}P, \Delta'
$$
Tuttavia il grado del taglio usato in $\Upsilon''$ \`e troppo elevato (\`e proprio di grado $d$). Ma questo \`e precisamente il caso in cui si applica il Lemma~\ref{lem:cut_LK1}, grazie al quale il taglio in $\Upsilon''$ pu\`o essere rimpiazzato con una derivazione di grado minore di $d$, e avente la stessa conclusione, dalla quale, mediante regole strutturali, possiamo ottenere $\Psi$.
\end{enumerate}
\end{proof}

Il prossimo lemma, che ci condurr\`a al risultato finale, afferma che \`e sempre possibile trasformare una dimostrazione in modo tale da diminuirne il grado. Formalmente:

\begin{lem}\label{lem:cut_LK3}
Sia $\Phi$ una dimostrazione di grado $d > 0$ per un certo sequente. Allora \`e possibile costruire una dimostrazione $\Psi$ per il medesimo sequente, avente grado inferiore.
\end{lem}
\begin{proof}
Per induzione sull'altezza $h(\Phi)$ della dimostrazione iniziale. Sia $\mrule{\rho}$ l'ultima regola applicata in $\Phi$ e siano $\Phi_i$ le premesse di $\mrule{\rho}$. Abbiamo due casi:
\begin{enumerate}
	\item $\mrule{\rho}$ non \`e un taglio di grado $d$. Per ipotesi induttiva, abbiamo $\Psi_i$ di grado minore di $d$, a cui possiamo applicare $\mrule{\rho}$ per ottenere $\Psi$;
	\item $\mrule{\rho}$ \`e un taglio di grado $d$:
	$$
		\vlderivation{
			\vliin{}{\mrule{cut}}
				{\Gamma, \Gamma' \vdash \Delta, \Delta'}
				{\vlpd{\Phi_1}{}{\Gamma \vdash P, \Delta}}
				{\vlpd{\Phi_2}{}{\Gamma', P \vdash \Delta'}}
		}
	$$
	Osserviamo che poich\'e il grado di questo $\mrule{cut}$ \`e $d$, abbiamo $\delta(P) = d$. Per ipotesi induttiva:
	$$
	\begin{array}{ccc}
		\vlderivation{\vlpd{\Psi_1}{}{\Gamma \vdash P, \Delta}} 
		& \qquad & 
		\vlderivation{\vlpd{\Psi_2}{}{\Gamma', P \vdash \Delta'}} 
	\end{array}
	$$
	hanno grado minore di $d$, e possiamo applicarvi il Lemma~\ref{lem:cut_LK2} per produrre $\Gamma, \Gamma'{\smallsetminus}P \vdash \Delta{\smallsetminus}P, \Delta'$ di grado inferiore a $d$; con alcune applicazioni di regole strutturali, otteniamo infine $\Psi$.
\end{enumerate}
\end{proof}

\begin{thm}[Gentzen Hauptsatz]\label{thm:haupt_lk}
La regola di taglio \`e ammissibile nel Sistema \textsf{LKp}.
\end{thm}
\begin{proof}
\`E sufficiente iterare l'applicazione del lemma precedente per trasformare una dimostrazione di grado strettamente positivo, in una di grado nullo, e quindi esente da applicazioni della regola di taglio.
\end{proof}

Il processo di eliminazione dei tagli fa esplodere l'altezza delle dimostrazioni. Infatti il Lemma~\ref{lem:cut_LK2} fa crescere l'altezza della prova in modo lineare nel caso peggiore (di un fattore $\kappa = 4$, senza considerare le applicazioni delle regole strutturali). Il Lemma~\ref{lem:cut_LK3} comporta una crescita esponenziale nel caso pessimo, cio\`e ridurre il grado di $1$ pu\`o accrescere l'albero di prova da $h$ a $\kappa^h$, poich\'e usando il Lemma~\ref{lem:cut_LK2} moltiplichiamo per $\kappa$ ad ogni unit\`a di altezza.

Quindi, mettendo tutto assieme, applicare l'Hauptsatz comporta una crescita iperesponenziale. Partendo da una dimostrazione di grado $d$ e altezza $h$ se ne ottiene una avente altezza $\mathcal{H}(d, h)$, dove:
\begin{eqnarray*}
	\mathcal{H}(0, h) & = & h \\
	\mathcal{H}(d+1, h) & = & \kappa^{\mathcal{H}(d, h)}
\end{eqnarray*}

L'Hauptsatz -- in varie forme, come la normalizzazione nel $\lambda$-calcolo -- \`e utilizzabile come fondamento teorico per la computazione. Per esempio, consideriamo un editor di testo: pu\`o essere visto come un insieme di lemmi generici (corrispondenti alle varie procedure di formattazione, impaginazione, \ldots) che possono essere applicati a input concreti, come una pagina scritta da qualche utente. Il numero di input possibili \`e chiaramente infinito e infatti i lemmi sono fatti per trattare infiniti casi; ma quando eseguiamo il programma su un certo input -- ad esempio per produrre in output una visualizzazione del testo -- i riferimenti a queste infinit\`a scompaiono. Concretamente, questa eliminazione dell'infinito \`e effettuata sostituendo sistematicamente le variabili (gli input dei lemmi) con il testo inserito dall'utente, in altre parole, eseguendo il programma. 

Questo \`e esattamente quello che fa l'algoritmo di eliminazione del taglio. Ecco perch\'e la struttura della procedure di cut elimination \`e cos\`i importante (osservazione fatta nel Lemma~\ref{lem:cut_LK2}). La strategia adottata nel ristrutturare la dimostrazione, effettuando le sostituzioni, produce delle scelte che sono, in generale, irreversibili. Questo pu\`o essere un problema, e si pu\`o risolvere ad esempio usando, al posto del calcolo dei sequenti, la deduzione naturale, che gode della propriet\`a di \emph{confluenza} (o \emph{propriet\`a di Church-Rosser}), la quale garantisce che \emph{le scelte fatte sono sempre reversibili}. Purtroppo la deduzione naturale soffre di altri problemi, e specialmente non gode della propriet\`a della sottoformula, e non si relaziona bene con la simmetria classica (ha molte premesse ma una sola conclusione). L'approccio deep inference pu\`o offrire diversi vantaggi nei confronti di ambedue questi formalismi.

\section{Deep inference e simmetria} 

L'eliminazione del taglio \`e un'idea centrale della proof theory. Se spostiamo tutto alla destra del turnstile e applichiamo qualche regola strutturale, la regola di taglio diventa:
$$
	\vliinf{}{\mvlrule*{cut}{1}}
		{\vdash \Delta, \Delta'}
		{\vdash P, \Delta}
		{\vdash \neg P, \Delta'}
$$
Quando letta dal basso verso l'alto, la regola di taglio introduce una formula arbitraria $P$, insieme alla sua negazione $\neg P$. Osserviamo ora la regola d'identit\`a, manipolata nuovamente per portare tutto a destra del tornello:
$$
	\vdash P, \neg P \qquad \mrule*{id}{1}
$$
Ci accorgiamo che, quando letta dall'alto al basso, anch'essa introduce una formula arbitraria assieme alla sua negazione. \`E chiaro che le due regole sono intimamente correlate. Tuttavia, la loro dualit\`a \`e oscurata dal fatto che le simmetrie verticali sono nascoste nel calcolo dei sequenti: le derivazioni sono alberi, e gli alberi sono verticalmente asimmetrici.

Per rivelare la dualit\`a tra le due regole, occorre ripristinare questa simmetria verticale. La forma ad albero delle derivazioni nel calcolo dei sequenti \`e dovuta alla presenza di regole d'inferenza con due premesse. Per esempio la regola destra di congiunzione, nella versione ad un lato diventa:
$$
	\vliinf{}{\mvlrule*[r]{\wedge}{1}}
		{\vdash P \wedge Q, \Delta, \Delta'}
		{\vdash P, \Delta}
		{\vdash Q, \Delta'}
$$
in cui \`e presente un'asimmetria: due premesse ma solo una conclusione. O per meglio dire: un connettivo nella conclusione, ma nessuno tra le premesse.

Questa asimmetria pu\`o essere riparata. Sappiamo che la virgola a destra del turnstile corrisponde alla disgiunzione, e che i diversi rami dell'albero di derivazione corrispondono a congiunzioni; pertanto la regola $\mrule*[r]{\wedge}{1}$ pu\`o essere riscritta come:
$$
	\vlinf{}{\mvlrule*[r]{\wedge}{1.1}}
		{\vdash (P \wedge Q) \vee \Delta \vee \Delta'}
		{\vdash (P \vee \Delta) \wedge (Q \vee \Delta')}
$$

In tal modo andiamo ad indentificare una parte del livello oggetto (i connettivi tra le formule) con il meta-livello (i rami dell'albero di derivazione). Cos\`i facendo rendiamo il sistema ``incompleto'', poich\'e uno degli scopi degli alberi di derivazione \`e quello di permettere alle regole d'inferenza di essere applicate in profondit\`a, seguendo la struttura sintattica delle formule. Consideriamo la derivazione:
$$
	\vlderivation{
		\vltr{}
			{\vdash \mctx{C}{P \wedge Q}}
			{\vlhy{\cdots}}
			{\vlhy{}}
			{
				\vliin{}{(\wedge_r^1)}
					{\vdash P \wedge Q, \Delta, \Delta'}
					{\vlhy{\vdash P, \Delta}}
					{\vlhy{\vdash Q, \Delta'}}
			}
	}
$$
in cui la conclusione contiene la sottoformula $P \wedge Q$. Leggendola dal basso all'alto, il motivo per cui la regola $\mrule*[r]{\wedge}{1}$ pu\`o essere applicata, \`e che le applicazioni nella derivazione sottostante decompongono il contesto $\mctx{C}{\emptyctx}$ e ne distribuiscono il contenuto tra le foglie dell'albero di derivazione.

Se vogliamo eliminare la forma ad albero delle derivazioni per ottenere un sistema completamente simmetrico, dobbiamo in qualche modo riconferire alle derivazioni l'abilit\`a di accedere alle sottoformule: questo pu\`o essere fatto direttamente, usando la metodologia deep inference. In questo modo, l'assioma d'identit\`a e la regola di taglio diventano:
$$
	\vlinf{}{\mvlrule{id}}{P \vee \neg P}{\mt}
	\qquad\qquad
	\vlinf{}{\mvlrule{cut}}{\mf}{P \wedge \neg P}
$$
da cui \`e evidente il carattere duale delle due: una pu\`o essere ottenuta dall'altra scambiando e negando la premessa e la conclusione. A questa nozione di dualit\`a ci si riferisce con l'aggettivo \emph{contrappositiva}.

Avremo modo di osservare una profonda simmetria, tutte le regole d'inferenza si raggrupperanno in coppie duali, come identit\`a e taglio. Questa dualit\`a si estender\`a naturalmente alle derivazioni: per ottenere la duale di una derivazione, baster\`a negare ogni formula e ``girare la derivazione sottosopra'', cio\`e leggerla dal basso verso l'alto.

\subsection{Sistema SKS generalizzato} 

Presentiamo un sistema formale per la logica classica proposizionale, che da una parte segue la tradizione del calcolo dei sequenti, in particolare possiede una regola di taglio e la sua ammissibilit\`a \`e dimostrata, mentre dall'altra, in contrasto col calcolo dei sequenti, ha regole che si applicano a profondit\`a arbitraria nelle formule e le derivazioni sono alberi degeneri (i.e. liste, le regole hanno al pi\`u una premessa). In questo Sistema potremo osservare una simmetria verticale nelle regole che mancava nel calcolo dei sequenti.

\begin{dfn}[Linguaggio \textsf{SKSg}, equivalenza] Sia $\mathcal{P}$ un insieme infinito enumerabile di \emph{simboli proposizionali}. L'insieme degli \emph{atomi} $\mathcal{A}$ \`e cos\`i definito:
$$
	\mathcal{A} = \{p, \mneg{p} \:|\: p \in \mathcal{P} \}
$$
dove $\mneg{\cdot}$ \`e una \emph{funzione di negazione primitiva sui simboli proposizionali}. La negazione si estende facilmente a tutti gli atomi definendo $\mneg{\mneg{p}} = p$ per ogni simbolo proposizionale negato $\mneg{p}$.

Siano $\mt, \mf \not\in \mathcal{A}$ simboli costanti o \emph{unit\`a} (che denotano rispettivamente \emph{il vero} ed \emph{il falso}) e sia $a \in \mathcal{A}$. Il \emph{linguaggio di \textsf{SKSg}} \`e definito dalle seguenti regole BNF di produzione:
$$
\begin{array}{llll}
	T & ::= & \mt \:|\: \mf \:|\: a & \quad\mbox{(termini)} \\
	P & ::= & T \:|\: \mand{P, P} \:|\: \mpar{P, P} & \quad\mbox{(formule)} 
\end{array}
$$
dove $\mand{P_1, P_2}$ e $\mpar{P_1, P_2}$ denotano rispettivamente la \emph{congiunzione} e la \emph{disgiunzione} delle formule $P_1$ e $P_2$. Come prima, dato un contesto $\mctx{C}{\emptyctx}$ e una formula $P$, indichiamo con $\mctx{C}{P}$ la formula ottenuta saturando il contesto $\mctx*{C}$ con la formula $P$. Ad esempio, sia $\mctx{C}{\emptyctx} = \mpar{a, \mand{\emptyctx, c}}$: allora $\mctx{C}{\mneg{b}} = \mpar{a, \mand{\mneg{b}, c}}$ mentre $\mctx{C}{\mand{b_1, b_2}} = \mpar{a, \mand{\mand{b_1, b_2}, c}}$; in quest'ultimo caso possiamo adottare la convenzione di \emph{omettere le parentesi graffe attorno ai termini composti} e pertanto di scrivere semplicemente $\mctx*{C}\mand{b_1, b_2}$.

Consideriamo due formule equivalenti quando appartengono alla relazione indotta dalle equazioni in Figura~\ref{fig:skseq}.

Infine, una formula \`e in \emph{forma normale negata} quando la negazione occorre solo sui simboli proposizionali.
\end{dfn}

\begin{figure}[t!]
\begin{minipage}[t]{.5\textwidth}
	\textbf{Associativit\`a}
	\begin{center}
		$\mpar{\mpar{P, Q}, R} = \mpar{P, \mpar{Q, R}}$ \\
		$\mand{\mand{P, Q}, R} = \mand{P, \mand{Q, R}}$ 
	\end{center}
	\textbf{Unit\`a}
	\begin{center}
	\begin{tabular}{ccc}
		$\mpar{\mt, \mt} = \mt$ & \quad & $\mpar{\mf, P} = P$ \\
		$\mand{\mf, \mf} = \mf$ & \quad & $\mand{\mt, P} = P$
	\end{tabular}
	\end{center}
	\textbf{Chiusura contestuale}
	\begin{center}
		\AxiomC{$P = Q$}
		\UnaryInfC{$\mctx{C}{P} = \mctx{C}{Q}$}
		\DisplayProof{}
	\end{center}
	\vspace{.25em}
\end{minipage}
\begin{minipage}[t]{.5\textwidth}
	\textbf{Commutativit\`a}
	\begin{center}
		$\mpar{P, Q} = \mpar{Q, P}$ \\
		$\mand{P, Q} = \mand{Q, P}$
	\end{center}
	\textbf{Negazione}
	\begin{center}
	\begin{tabular}[c]{ccc}
		$\mneg{\mf}$ & $=$ & $\mt$ \\
		$\mneg{\mt}$ & $=$ & $\mf$ \\
		$\mneg{\mpar{P, Q}}$ & $=$ & $\mand{\mneg{P}, \mneg{Q}}$ \\
		$\mneg{\mand{P, Q}}$ & $=$ & $\mpar{\mneg{P}, \mneg{Q}}$ \\
		$\mneg{\mneg{P}}$ & $=$ & $P$
	\end{tabular}
	\end{center}
\end{minipage}
\textbf{Equivalenza}
\begin{center}
	$P = P$
	\qquad\qquad
	\AxiomC{$P = Q$}
	\UnaryInfC{$Q = P$}
	\DisplayProof{}
	\qquad\qquad
	\AxiomC{$P = Q$}
	\AxiomC{$Q = R$}
	\BinaryInfC{$P = R$}
	\DisplayProof{}
\end{center} 
\caption{Equivalenza tra formule di \textsf{SKSg}}
\label{fig:skseq}
\end{figure}

In virt\`u della propriet\`a associativa, adottiamo la seguente convenzione:
\begin{eqnarray*}
	\mpar{P_1, P_2, \ldots, P_n} & = & \mpar{P_1, \mpar{P_2, \mpar{\ldots, P_n} \ldots}} \\
	\mand{P_1, P_2, \ldots, P_n} & = & \mand{P_1, \mand{P_2, \mand{\ldots, P_n} \ldots}}
\end{eqnarray*}
cio\`e scriviamo rispettivamente liste di disgiunzioni e congiunzioni senza curarci di come le sottoformule siano associate tra loro, assumendo quando non specificato che associno a destra.

Osserviamo inoltre che l'equivalenza ci permette di spingere la negazione all'interno delle formule fino a livello degli atomi, scambiando ogni volta congiunzione e disgiunzione in stile De Morgan.

\begin{dfn}[Dualit\`a, simmetria]
Il \emph{duale di una regola} d'inferenza si ottiene scambiando la premessa con la conclusione e negando ambedue. Ad esempio:
$$
	\vlinf{}{\mruledn{i}}{\mpar{P, \mneg{P}}}{\mt}
	\qquad\qquad
	\vlinf{}{\mruleup{i}}{\mf}{\mand{P, \mneg{P}}}
$$

Un \emph{sistema deduttivo} \`e \emph{simmetrico} se per ogni regola d'inferenza esso contiene anche la duale.
\end{dfn}

Il Sistema deduttivo \textsf{SKSg} \`e riportato in Figura~\ref{fig:sksg_cos}. Il suo nome \`e un acronimo, in cui la prima ``S'' indica che \`e simmetrico, la ``K'' sta per ``Klassisch'' (come nel Sistema \textsf{LK}) e la ``S'' finale dice che il Sistema \`e espresso nel calcolo delle strutture (il termine ``struttura'' \`e usato per indicare una lista di formule in congiunzione o in disgiunzione). La ``g'' minuscola indica che il Sistema \`e \emph{generalizzato}, che significa che le regole non sono ristrette alla forma atomica.

\`E possibile dimostrare~(\cite{Bru04}) che questo sistema formale cattura tutte le dimostrazioni esprimibili in \textsf{LKp}, passando per un sistema intermedio chiamato calcolo dei sequenti ``ad un lato'' o calcolo dei sequenti di Gentzen-Sch\"utte~(\cite{Sch50, TroSch96}).

\begin{figure}
\begin{tabular}{ccccc}
	$\mt \quad \mrule{ax}$ &
	$\vlinf{}{\mruledn{i}}{\mctx*{C}\mpar{P, \mneg{P}}}{\mctx{C}{\mt}}$ & 
	$\vlinf{}{\mruledn{w}}{\mctx{C}{P}}{\mctx{C}{\mf}}$ &
	$\vlinf{}{\mruledn{c}}{\mctx{C}{P}}{\mctx*{C}\mpar{P, P}}$ &
	$\vlinf{}{\mvlrule{s}}{\mctx*{C}\mpar{\mand{P, Q}, R}}{\mctx*{C}\mand{P, \mpar{Q, R}}}$ \\\\
	\begin{minipage}[c]{6em}\centering(nessuna regola per $\mf$)\end{minipage} &
	$\vlinf{}{\mruleup{i}}{\mctx{C}{\mf}}{\mctx*{C}\mand{P, \mneg{P}}}$ &
	$\vlinf{}{\mruleup{w}}{\mctx{C}{\mt}}{\mctx{C}{P}}$ &
	$\vlinf{}{\mruleup{c}}{\mctx*{C}\mand{P, P}}{\mctx{C}{P}}$ &
	~
\end{tabular}
\caption{Sistema deduttivo \textsf{SKSg}}
\label{fig:sksg_cos}
\end{figure}

Le regole $\mrule{s}$, $\mruledn{w}$ e $\mruledn{c}$ sono chiamate rispettivamente \emph{scambio}, \emph{indebolimento} e \emph{contrazione}. Le duali portano lo stesso nome, con l'aggiunta del prefisso ``co-'', ad esempio $\mruleup{w}$ \`e chiamata \emph{co-indebolimento}. La regola di scambio \`e duale a s\'e stessa, o \emph{auto-duale}. La sua funzione \`e quella di modellare il comportamento duale di congiunzione e disgiunzione. La sua semantica intesa non \`e facile da cogliere: per una spiegazione dettagliata, si rimanda alla Sezione~\ref{sec:lbv_opi} del presente volume.

Mentre \`e immediato osservare la corrispondenza tra $\mruledn{w}$ e $\mruledn{c}$ in \textsf{SKS} e le regole di indebolimento e contrazione nel calcolo dei sequenti, le loro duali non hanno corrispettivi in \textsf{LKp}. Il loro ruolo \`e quello di assicurare la simmetria del Sistema; se non siamo interessati alla simmetria, si pu\`o dimostrare che queste regole (e anche il taglio, cio\`e tutte le regole aventi la freccia rivolta verso l'alto) sono ammissibili. Infatti la nozione di dimostrazione \`e inerentemente asimmetrica: il duale di una dimostrazione \emph{non \`e} una dimostrazione, bens\`i \`e una derivazione che si conclude con l'unit\`a $\mf$, ossia una \emph{refutazione}.

Il primo meta-teorema che andremo a dimostrare, ci d\`a una caratterizzazione del Sistema \textsf{SKSg}, mettendo in relazione il concetto di derivazione con quello di dimostrazione. 

\begin{thm}[Deduzione]~\\
Esiste una derivazione $\;\toks0={3.0}\vlderivation{\vldf{\Psi}{\mbox{\scriptsize{\textsf{SKSg}}}}{Q}{\vlhy{P}}{\the\toks0}}\;$ se e solo se esiste una dimostrazione $\;\toks0={3.0}\vlderivation{\vlpf{\Phi}{\mbox{\scriptsize{\textsf{SKSg}}}}{\mpar{\mneg{P}, Q}}{\the\toks0}}\;$.
\end{thm}
\begin{proof} 
La dimostrazione $\Phi$ pu\`o essere ottenuta, data una derivazione $\Psi$, come segue:
$$
	\toks0={3.0}\vlderivation{
		\vldf{\mpar{\mneg{P}, \Psi}}{\scriptsize{\textsf{SKSg}}}{\mpar{\mneg{P}, Q}}{
			\vlin{}{\mruledn{i}}{\mpar{\mneg{P}, P}}{\vlhy{\mt}}
		}{\the\toks0}
	}
$$
Osserviamo che grazie alla metodologia deep inference, \`e stato possibile racchiudere l'intera derivazione $\Psi$ all'interno del contesto $\mctx{C}{\emptyctx} = \mpar{\mneg{P}, \emptyctx}$.

La derivazione $\Psi$ si ottiene da $\Phi$ come segue:
$$
	\toks0={3.0}\vlderivation{
		\vlin{}{\mruleup{i}}{\mpar{\mf, Q} = Q}{
			\vlin{}{\mrule{s}}{\mpar{\mand{P, \mneg{P}}, Q}}{
				\vldf{\mand{P, \Phi}}{\mbox{\scriptsize{\textsf{SKSg}}}}{\mand{P, \mpar{\mneg{P}, Q}}}{\vlhy{P = \mand{P, \mt}}}{\the\toks0}
			}
		}
	}
$$
Anche in questo caso la trasformazione ha avuto successo perch\'e \`e stato possibile usare la dimostrazione $\Phi$ nel contesto $\mctx{C'}{\emptyctx} = \mand{P, \emptyctx}$.
\end{proof}

\subsection{Localit\`a: il Sistema SKS} 

Le regole d'inferenza che duplicano una quantit\`a illimitata di informazione sono problematiche dal punto di vista della complessit\`a e dell'implementazione, ad esempio, della proof search. Nel calcolo dei sequenti, la regola di contrazione:
\begin{center}
	\AxiomC{$\Gamma, P, P \vdash \Delta$}
	\RightLabel{$\mrule[l]{cont}$}
	\UnaryInfC{$\Gamma, P \vdash \Delta$}
	\DisplayProof{}
\end{center}
quando letta dall'alto al basso duplica una formula $P$ di dimensione arbitraria. Qualunque sia il meccanismo effettivo che compie questa duplicazione, esso necessita di una visione \emph{globale} delle copie di $P$ presenti: se ad esempio pensiamo di implementare la contrazione su un sistema distribuito, in cui ogni processore ha una quantit\`a limitata di memoria locale, la formula $P$ potrebbe essere replicata in processori diversi. In questo caso nessun processore avrebbe una visone globale delle copie di $P$, e bisognerebbe usare un meccanismo \emph{ad hoc} per gestire questa situazione. Chiamiamo \emph{locali} le regole d'inferenza che non necessitano di una visione globale su formule di dimensione arbitraria, e \emph{non-locali} le altre. 

Mentre \`e possibile utilizzare tecniche per risolvere questa situazione nelle implementazioni, una questione interessante \`e trovare un approccio teorico che sia in grado di eliminare le regole non-locali. Questo \`e possibile, riducendo le regole non-locali alla loro forma atomica. Ad esempio, l'identit\`a:
$$
	\vlinf{}{\mruledn{i}}{\mctx*{C}\mpar{P, \mneg{P}}}{\mctx{C}{\mt}}
	\qquad\mbox{\`e sostituita dalla regola}\qquad
	\vlinf{}{\mruledn{ai}}{\mctx*{C}\mpar{a, \mneg{a}}}{\mctx{C}{\mt}}
$$
dove $a$ \`e un simbolo proposizionale.

Operazioni analoghe possono essere fatte anche nel calcolo dei sequenti; l'unica regola problematica \`e, appunto, la contrazione. Essa non pu\`o semplicemente essere ristretta alla forma atomica nel Sistema \textsf{SKSg}. Il problema si risolve inserendo nel Sistema una nuova regola, introdotta in~\cite{BruTiu01} e chiamata \emph{mediale}:
$$
	\vlinf{}{\mvlrule{m}}{\mctx*{C}\mand{\mpar{P,R}, \mpar{Q,S}}}{\mctx*{C}\mpar{\mand{P,Q}, \mand{R,S}}}
$$
Questa regola non ha analoghi nel calcolo dei sequenti, ma \`e chiaramente corretta, poich\'e \`e derivabile da $\{\mruledn{c}, \mruledn{w}\}$:
$$
	\vlderivation{
		\vlin{}{\mruledn{c}}{\mctx*{C}\mand{\mpar{P,R}, \mpar{Q,S}}}{
			\vlin{}{\mruledn{w}}{\mctx*{C}\mpar{\mand{\mpar{P,R}, \mpar{Q,S}}, \mand{\mpar{P,R}, \mpar{Q,S}}}}{
				\vlin{}{\mruledn{w}}{\mpar{\mand{P, \mpar{Q,S}}, \mand{\mpar{P,R}, \mpar{Q,S}}}}{
					\vlin{}{\mruledn{w}}{\mpar{\mand{P, Q}, \mand{\mpar{P,R}, \mpar{Q,S}}}}{
						\vlin{}{\mruledn{w}}{\mpar{\mand{P, Q}, \mand{R, \mpar{Q,S}}}}{
							\vlhy{\mpar{\mand{P, Q}, \mand{R, S}}}
						}
					}
				}
			}
		}
	}
$$

\begin{figure}
\begin{tabular}{ccccccc}
	$\mt \quad \mrule{ax}$ & \quad\quad &
	$\vlinf{}{\mruledn{ai}}{\mctx*{C}\mpar{a, \mneg{a}}}{\mctx{C}{\mt}}$ & \quad\quad &
	$\vlinf{}{\mruledn{aw}}{\mctx{C}{a}}{\mctx{C}{\mf}}$ & \quad\quad &
	$\vlinf{}{\mruledn{ac}}{\mctx{C}{a}}{\mctx*{C}\mpar{a, a}}$ \\\\
	\begin{minipage}[c]{6em}\centering(nessuna regola per $\mf$)\end{minipage} & \qquad\quad &
	$\vlinf{}{\mruleup{ai}}{\mctx{C}{\mf}}{\mctx*{C}\mand{a, \mneg{a}}}$ & \quad\quad &
	$\vlinf{}{\mruleup{aw}}{\mctx{C}{\mt}}{\mctx{C}{a}}$ & \quad\quad &
	$\vlinf{}{\mruleup{ac}}{\mctx*{C}\mand{a, a}}{\mctx{C}{a}}$ 
\end{tabular}
\vspace{1em}
$$
	\vlinf{}{\mvlrule{s}}{\mctx*{C}\mpar{\mand{P, Q}, R}}{\mctx*{C}\mand{P, \mpar{Q, R}}}
	\qquad\qquad\qquad
	\vlinf{}{\mvlrule{m}}{\mctx*{C}\mand{\mpar{P,R}, \mpar{Q,S}}}{\mctx*{C}\mpar{\mand{P,Q}, \mand{R,S}}}
$$
\caption{Regole del Sistema \emph{locale} \textsf{SKS}}
\label{fig:sks_cos}
\end{figure}

Il prossimo teorema ci garantisce la derivabilit\`a del Sistema locale \textsf{KS} in Figura~\ref{fig:sks_cos}:

\begin{thm}\label{thm:sksg_atomic}
Le regole $\mruledn{i}$, $\mruledn{w}$ e $\mruledn{c}$ sono derivabili, rispettivamente da $\{\mruledn{ai}, \mrule{s}\}$, $\{\mruledn{aw}, \mrule{s}\}$, $\{\mruledn{ac}, \mrule{m}\}$. Dualmente, le regole $\mruleup{i}$, $\mruleup{w}$ e $\mruleup{c}$ sono risp. derivabili da $\{\mruleup{ai}, \mrule{s}\}$, $\{\mruleup{aw}, \mrule{s}\}$, $\{\mruleup{ac}, \mrule{m}\}$.
\end{thm}
\begin{proof}
Data un'istanza di una delle seguenti regole:
$$
	\vlinf{}{\mruledn{i}}{\mctx*{C}\mpar{P, \mneg{P}}}{\mctx{C}{\mt}}
	\qquad,\qquad
	\vlinf{}{\mruledn{w}}{\mctx{C}{P}}{\mctx{C}{\mf}}
	\qquad,\qquad
	\vlinf{}{\mruledn{c}}{\mctx{C}{P}}{\mctx*{C}\mpar{P, P}}
$$
costruiamo una nuova derivazione per induzione strutturale su $P$:
\begin{itemize}
	\item $P$ \`e un atomo. Allora l'istanza di una regola generale \`e anche un'istanza della corrispettiva in forma atomica.
	\item $P = \mt$ o $P = \mf$. Allora l'istanza di una regola generale \`e un'istanza della relazione d'equivalenza, con l'eccezione dell'indebolimento quando $P = \mf$. Allora la regola d'indebolimento generale \`e sostituita da:
	$$
		\vlinf{}{\mvlrule{s}}{\mctx*{C}\mpar{\mand{\mf, \mt}, \mt} = \mctx{C}{\mt}}
			{\mctx{C}{\mf} = \mctx*{C}\mand{\mf, \mpar{\mt, \mt}}}
	$$
	\item $P = \mpar{Q, R}$. Per ipotesi induttiva, usando rispettivamente le sole regole $\{\mruledn{ai}, \mrule{s}\}$, $\{\mruledn{aw}, \mrule{s}\}$ e $\{\mruledn{ac}, \mrule{m}\}$, abbiamo:
	$$
		\vlderd{\Phi_Q^{\mruledn{i}}}{}{\mctx*{C}\mpar{Q, \mneg{Q}}}{\mctx{C}{\mt}}
		\quad
		\vlderd{\Phi_R^{\mruledn{i}}}{}{\mctx*{C}\mpar{R, \mneg{R}}}{\mctx{C}{\mt}}
		\quad,\quad
		\vlderd{\Phi_Q^{\mruledn{w}}}{}{\mctx{C}{Q}}{\mctx{C}{\mf}}
		\quad
		\vlderd{\Phi_R^{\mruledn{w}}}{}{\mctx{C}{R}}{\mctx{C}{\mf}}
		\quad,\quad
		\vlderd{\Phi_Q^{\mruledn{c}}}{}{\mctx{C}{Q}}{\mctx*{C}\mpar{Q, Q}}
		\quad
		\vlderd{\Phi_R^{\mruledn{c}}}{}{\mctx{C}{R}}{\mctx*{C}\mpar{R, R}}
	$$
	da cui \`e possibile derivare:
	$$
		\vlderivation{
			\vlin{}{\mvlrule{s}}{\mctx*{C}\mpar{\mand{\mneg{Q}, \mneg{R}}, \mpar{Q, R}}}{
				\vlin{}{\mvlrule{s}}{\mctx*{C}\mpar{\mand{\mneg{R}, \mpar{\mneg{Q}, Q}}, R}}{
					\vldd{\Phi_Q^{\mruledn{i}}}{}{\mctx*{C}\mand{\mpar{\mneg{Q}, Q}, \mpar{\mneg{R}, R}}}{
						\vldd{\Phi_R^{\mruledn{i}}}{}{\mctx*{C}\mpar{\mneg{R}, R}}
							{\vlhy{\mctx{C}{\mt}}}
					}
				}
			}
		}
		\qquad,\qquad
		\vlderivation{
			\vldd{\Phi_R^{\mruledn{w}}}{}{\mctx*{C}\mpar{Q, R}}{
				\vldd{\Phi_Q^{\mruledn{w}}}{}{\mctx*{C}\mpar{Q, \mf}}{
					\vlhy{\mctx{C}{\mf} = \mctx*{C}\mpar{\mf, \mf}}
				}
			}
		}
		\qquad,\qquad
		\vlderivation{
			\vldd{\Phi_R^{\mruledn{c}}}{}{\mctx*{C}\mpar{Q, R}}{
				\vldd{\Phi_Q^{\mruledn{c}}}{}{\mctx*{C}\mpar{Q, R, R}}{
					\vlhy{\mctx*{C}\mpar{Q, Q, R, R}}
				}
			}
		}
	$$

	\item $P = \mand{Q, R}$. L'ipotesi induttiva \`e identica a quella del caso precedente,	da cui \`e possibile derivare:
	$$
		\vlderivation{
			\vlin{}{\mvlrule{s}}{\mctx*{C}\mpar{\mand{Q, R}, \mpar{\mneg{Q}, \mneg{R}}}}{
				\vlin{}{\mvlrule{s}}{\mctx*{C}\mpar{\mand{R, \mpar{Q, \mneg{Q}}}, \mneg{R}}}{
					\vldd{\Phi_Q^{\mruledn{i}}}{}{\mctx*{C}\mand{\mpar{Q, \mneg{Q}}, \mpar{R, \mneg{R}}}}{
						\vldd{\Phi_R^{\mruledn{i}}}{}{\mctx*{C}\mpar{R, \mneg{R}}}
							{\vlhy{\mctx{C}{\mt}}}
					}
				}
			}
		}
		\qquad,\qquad
		\vlderivation{
			\vldd{\Phi_R^{\mruledn{w}}}{}{\mctx*{C}\mand{Q, R}}{
				\vldd{\Phi_Q^{\mruledn{w}}}{}{\mctx*{C}\mand{Q, \mf}}{
					\vlhy{\mctx{C}{\mf} = \mctx*{C}\mand{\mf, \mf}}
				}
			}
		}
		\qquad,\qquad
		\vlderivation{
			\vldd{\Phi_R^{\mruledn{c}}}{}{\mctx*{C}\mand{Q, R}}{
				\vldd{\Phi_Q^{\mruledn{c}}}{}{\mctx*{C}\mand{Q, \mpar{R, R}}}{
					\vlin{}{\mvlrule{m}}{\mctx*{C}\mand{\mpar{Q, Q}, \mpar{R, R}}}{
						\vlhy{\mctx*{C}\mpar{\mand{Q, R}, \mand{Q, R}}}
					}
				}
			}
		}
	$$
\end{itemize}

I casi duali si dimostrano allo stesso modo, ``girando sottosopra'' le dimostrazioni e negando tutte le formule.
\end{proof}

Questo \`e un risultato molto significativo e difficilmente ottenibile usando il calcolo dei sequenti. Inoltre la dimostrazione \`e costruttiva e modulare, caratteristiche che ci permetteranno in seguito di utilizzare regole generalizzate con la consapevolezza di poterle sempre sostituire con una procedura effettiva con le loro versioni atomiche.

\subsection{Rompere la simmetria: il Sistema KS} 

Dimostriamo che nel Sistema \textsf{SKS} le regole con la freccia rivolta in alto $\mruleup{\rho}$ sono \emph{ammissibili} (e quindi in particolare anche la regola di taglio lo \`e). Il Sistema risultante dall'eliminazione delle regole $\mruleup{\rho}$ \`e chiamato Sistema \textsf{KS}, ed \`e riportato in Figura~\ref{fig:ks_cos}. 

In questa sezione seguiamo la dimostrazione di~\cite{Bru04}, a cui ho apportato alcune modifiche personali di carattere tecnico.

\begin{lem}\label{lem:sks_cut1}
Ogni regola di \textsf{SKS} \`e derivabile usando solo la sua duale, \emph{identit\`a}, \emph{taglio} e \emph{switch}.
\end{lem}
\begin{proof}
Le regole $\mrule{s}$ e $\mrule{m}$ sono auto-duali, e pertanto banalmente derivabili. Un'istanza di una regola $\vlinf{}{\mruleup{\rho}}{\mctx{C}{Q}}{\mctx{C}{P}}$ pu\`o essere sostituita da:
$$
	\vlderivation{
		\vlin{}{\mruleup{i}}{\mctx{C}{Q}}{
			\vlin{}{\mruledn{\rho}}{\mctx*{C}\mpar{\mand{P, \mneg{P}}, Q}}{
				\vlin{}{\mvlrule{s}}{\mctx*{C}\mpar{\mand{P, \mneg{Q}}, Q}}{
					\vlin{}{\mruledn{i}}{\mctx*{C}{\mand{P, \mpar{\mneg{Q}, Q}}}}{\vlhy{\mctx{C}{P}}}
				}
			}
		}
	}
$$
e lo stesso vale per le regole $\mruledn{\rho}$.
\end{proof}

Prima di proseguire con la cut elimination, occorre stabilire una semplice proposizione, valida per la maggior parte dei sistemi espressi col calcolo delle strutture.

\begin{prop}\label{prop:sks_ctx_ins}
Per ogni struttura $P,Q$ e contesto $\mctx*{C}$, esiste una derivazione $\vlderd{}{\{\mrule{s}\}}{\mpar{\mctx{C}{P}, Q}}{\mctx*{C}\mpar{P,Q}}$.
\end{prop}
\begin{proof}
Per induzione sulla dimensione del contesto $\mctx*{C}$.
\begin{enumerate}
	\item Il caso base \`e $\mctx{C}{\emptyctx} = \emptyctx$, da cui si deriva che esiste una derivazione (vuota) per $\mpar{P,Q}$.
	\item $\mctx{C}{\emptyctx} = \mpar{R, \mctx{C'}{\emptyctx}}$. Allora, per ipotesi induttiva, esiste una derivazione: 
	$$
		\vlderd{\Pi}{\{\mrule{s}\}}{\mpar{\mctx{C'}{P}, Q}}{\mctx*{C'}\mpar{P,Q}}
	$$
	che pu\`o essere usata per costruire:
	$$
		\vlderd{\mpar{R, \Pi}}{\{\mrule{s}\}}{\mpar{R, \mctx{C'}{P}, Q}}{\mpar{R, \mctx*{C'}\mpar{P,Q}}}
	$$
	e $\mpar{R, \mctx{C'}{P}, Q}$ \`e proprio uguale a $\mpar{\mctx{C}{P},Q}$.
	\item $\mctx{C}{\emptyctx} = \mand{R, \mctx{C'}{\emptyctx}}$. Qui l'ipotesi induttiva ci d\`a:
	$$
		\vlderd{\Pi}{\{\mrule{s}\}}{\mpar{\mctx{C'}{P}, Q}}{\mctx*{C'}\mpar{P,Q}}
	$$
	che pu\`o essere usata per costruire:
	$$
		\vlderivation{
			\vlin{}{\mrule{s}}{\mpar{\mand{R, \mctx{C'}{P}}, Q} = \mpar{\mctx{C}{P},Q}}{
				\vldd{\mand{R, \Pi}}{\{\mrule{s}\}}{\mand{R, \mpar{\mctx{C'}{P}, Q}}}{
					\vlhy{\mand{R, \mctx*{C'}\mpar{P,Q}}}
				}
			}
		}
	$$
\end{enumerate}
\end{proof}

\begin{dfn}[Taglio atomico di superficie]
Un'istanza della regola di taglio atomica $\mruleup{ai}$ \`e chiamata \emph{shallow} (o \emph{taglio atomico di superficie}) quando \`e della forma:
$$
	\vlinf{}{\mruleup{ai}}{S}{\mpar{S, \mand{a, \mneg{a}}}}
$$
\end{dfn}

\begin{lem}\label{lem:sks_shallow_cut}
La regola di taglio atomica $\mruleup{ai}$ \`e derivabile usando \emph{taglio atomico di superficie} e \emph{switch}.
\end{lem}
\begin{proof}
Ogni formula $\mctx*{C}\mand{a, \mneg{a}}$ \`e equivalente a $\mctx*{C}\mpar{\mf, \mand{a, \mneg{a}}}$. Per la Proposizione~\ref{prop:sks_ctx_ins}, esiste una derivazione:
$$
	\vlderd{}{\{\mrule{s}\}}{\mpar{\mctx{C}{\mf}, \mand{a, \mneg{a}}}}{\mctx*{C}\mpar{\mf, \mand{a, \mneg{a}}}}
$$
Pertanto basta porre $S = \mctx{C}{\mf}$ per effettuare la trasformazione:
$$
	\vlinf{}{\mruleup{ai}}{\mctx{C}{\mf}}{\mctx*{C}\mand{a, \mneg{a}}}
	\qquad\rightsquigarrow\qquad
	\vlderivation{
		\vlin{}{\mruleup{ai}}{S}{
			\vldd{}{\{\mrule{s}\}}{\mpar{S, \mand{a, \mneg{a}}}}{
				\vlhy{\mctx*{C}\mand{a, \mneg{a}}}
			}
		}
	}
$$
\end{proof}

\begin{lem}\label{lem:ks_alltrue}
Ogni dimostrazione $\vlproofd{}{\mathsf{KS}}{\mctx{C}{a}}$ pu\`o essere trasformata in $\vlproofd{}{\mathsf{KS}}{\mctx{C}{\mt}}$.
\end{lem}
\begin{proof}
Risalendo la dimostrazione, sostituiamo nelle regole l'occorrenza di $a$ e le sue copie prodotte per contrazione, con l'unit\`a $\mt$. Le istanze delle regole $\mrule{s}$ e $\mrule{m}$ rimangono intatte, le istanze di $\mruledn{ac}$ si riducono ad applicazioni della relazione d'equivalenza $=$. Le altre applicazioni vengono sostituite dalle seguenti derivazioni:
$$
\begin{array}{ccc}
	\vlinf{}{\mruledn{aw}}{\mctx{C}{a}}{\mctx{C}{\mf}} & \rightsquigarrow & \vlinf{}{\mrule{s}}{\mctx*{C}\mpar{\mand{\mf,\mt},\mt} = \mctx{C}{\mt}}{\mctx{C}{\mf} = \mctx*{C}\mand{\mf, \mpar{\mt, \mt}}} \\
	~ & \qquad\qquad\qquad & ~ \\
	\vlinf{}{\mruledn{ai}}{\mctx*{C}\mpar{a, \mneg{a}}}{\mctx{C}{\mt}} & \rightsquigarrow & \vlinf{}{\mruledn{aw}}{\mctx*{C}\mpar{\mt, \mneg{a}}}{\mctx{C}{\mt} = \mctx*{C}\mpar{\mt, \mf}}
\end{array}
$$
\end{proof}

\begin{figure}{thbs}
$$
	\mt \quad \mrule{ax}
	\qquad\qquad
	\vlinf{}{\mruledn{ai}}{\mctx*{C}\mpar{a, \mneg{a}}}{\mctx{C}{\mt}}
	\qquad\qquad
	\vlinf{}{\mruledn{aw}}{\mctx{C}{a}}{\mctx{C}{\mf}}
	\qquad\qquad
	\vlinf{}{\mruledn{ac}}{\mctx{C}{a}}{\mctx*{C}\mpar{a, a}}
$$
~\\
$$
	\vlinf{}{\mvlrule{s}}{\mctx*{C}\mpar{\mand{P, Q}, R}}{\mctx*{C}\mand{P, \mpar{Q, R}}}
	\qquad\qquad\qquad
	\vlinf{}{\mvlrule{m}}{\mctx*{C}\mand{\mpar{P,R}, \mpar{Q,S}}}{\mctx*{C}\mpar{\mand{P,Q}, \mand{R,S}}}
$$
\caption{Regole del Sistema \textsf{KS}}
\label{fig:ks_cos}
\end{figure}

\begin{thm}
Ogni dimostrazione $\vlproofd{}{\mathsf{SKS}}{P}$ pu\`o essere trasformata in una dimostrazione $\vlproofd{}{\mathsf{KS}}{P}$.
\end{thm}
\begin{proof}
Grazie al Lemma~\ref{lem:sks_cut1}, sappiamo che l'unica regola da eliminare \`e il taglio $\mruleup{ai}$. Grazie al Lemma~\ref{lem:sks_shallow_cut} possiamo sostituire tutti i tagli con tagli di superficie. Partendo dall'alto, selezioniamo la prima istanza della regola di taglio:
$$
	\vlderivation{
		\vldd{\Phi}{\mathsf{KS} \cup \{\mruleup{ai}\}}{P}{
			\vlin{}{\mruleup{ai}}{R}{
				\vlpd{\Pi}{\mathsf{KS}}{\mpar{R, \mand{a, \mneg{a}}}}
			}
		}
	}
$$
Applicando due volte il Lemma~\ref{lem:ks_alltrue} a $\Pi$, otteniamo:
$$
	\vlproofd{\Pi_1}{\mathsf{KS}}{\mpar{R, a}}
	\qquad,\qquad
	\vlproofd{\Pi_2}{\mathsf{KS}}{\mpar{R, \mneg{a}}}
$$

Partendo dalla conclusione e risalendo la dimostrazione $\Pi_1$, sostituiamo l'occorrenza di $a$ e le sue copie prodotte per contrazione, con la formula $R$. Le istanze delle regole $\mrule{m}$ e $\mrule{s}$ rimangono intatte, mentre le istanze di $\mruledn{ac}$ e $\mruledn{aw}$ vengono sostituite dalle loro versioni generalizzate:
$$
\begin{array}{ccc}
	\vlinf{}{\mruledn{ac}}{\mctx{C}{a}}{\mctx*{C}\mpar{a,a}} & \rightsquigarrow & \vlinf{}{\mruledn{c}}{\mctx{C}{R}}{\mctx*{C}\mpar{R,R}} \\
	~ & \qquad\qquad\qquad & ~ \\
	\vlinf{}{\mruledn{aw}}{\mctx{C}{a}}{\mctx{C}{\mf}} & \rightsquigarrow & \vlinf{}{\mruledn{w}}{\mctx{C}{R}}{\mctx{C}{\mf}}
\end{array}
$$
Le istanze di $\mruledn{ai}$ sono sostituite da $\mctx{C}{\Pi_2}$:
$$
	\vlinf{}{\mruledn{ai}}{\mctx*{C}\mpar{a, \mneg{a}}}{\mctx{C}{\mt}}
	\quad\quad\rightsquigarrow\quad\quad\quad
	\vlderd{\mctx{C}{\Pi_2}}{\mathsf{KS}}{\mctx*{C}\mpar{R, \mneg{a}}}{\mctx{C}{\mt}}
$$
Il risultato di questa sostitituzione di $\Pi_2$ dentro $\Pi_1$ \`e una dimostrazione $\Pi_3$, grazie alla quale possiamo costruire:
$$
	\vlderivation{
		\vldd{\Phi}{\mathsf{KS} \cup \{\mruleup{ai}\}}{P}{
			\vlin{}{\mruledn{c}}{R}{
				\vlpd{\Pi_3}{\mathsf{KS}}{\mpar{R, R}}
			}
		}
	}
$$

Ora basta procedere induttivamente verso il basso per rimuovere le rimanenti istanze di $\mruleup{ai}$. Alla fine di questo procedimento, le regole generalizzate possono essere rimosse usando la procedura descritta nella dimostrazione del Teorema~\ref{thm:sksg_atomic}.
\end{proof}

\chapter{Logica lineare}

La logica lineare \`e un'estensione della logica classica ideata da Jean-Yves Girard verso la fine degli anni '80~(\cite{Gir87, GirLafTay89, Gir95}). La caratteristica peculiare della logica lineare \`e che tratta l'implicazione come \emph{fenomeno causale} anzich\'e (com'\`e pratica comune in matematica) come \emph{concetto stabile}:
$$
	\mbox{se }A\mbox{ e }A{\Rightarrow}B\mbox{ allora }B\mbox{, \emph{ma $A$ \`e ancora valida.}}
$$
Un'implicazione causale non pu\`o essere reiterata, poich\'e le condizioni iniziali sono modificate dopo il suo utilizzo; questo processo di modifica delle premesse (condizioni) \`e noto in fisica come \emph{reazione}\footnote{Quello di reazione \`e un concetto base anche della teoria dei modelli concorrenti, vedi ad esempio~\cite{Mil92, SanWal01}.}. Per esempio, se $A$ \`e ``spendere una moneta nel distributore automatico di bevande (o DAB)'' e $B$ \`e ``prendere un caff\`e'', la moneta viene persa nel processo, che quindi non si pu\`o ripetere una seconda volta. Esistono tuttavia casi, sia in matematica che nella vita reale, in cui le reazioni non esistono o sono trascurabili: ad esempio un lemma che resta sempre vero, o un tecnico che possiede la chiave del DAB e pu\`o recuperare ogni volta la sua moneta. Questi sono i casi che Girard chiama \emph{situazioni}, cio\`e condizioni durature e immutevoli (o {verit\`a stabili}), e sono comunque gestibili in logica lineare tramite speciali connettivi (gli \emph{esponenziali}, ``!'' e ``?''). Gli esponenziali esprimono la reiterabilit\`a di un'azione, ossia l'assenza di reazioni; tipicamente $!A$ significa ``spendere quante monete si vogliono''. Usiamo il simbolo $\multimap$ per denotare l'implicazione causale (o \emph{implicazione lineare}); vale la seguente equazione:
$$
	A \Rightarrow B \quad=\quad (!A) \multimap B
$$
cio\`e $B$ \`e causato da un certo numero d'iterazioni di $A$.

Una \emph{azione di tipo $A$} consister\`a nel tirare fuori una certa moneta dalla tasca di qualcuno (ci potrebbero essere diverse azioni di questo tipo, poich\'e potremmo disporre di diverse monete). Analogamente saranno disponibili un certo numero di caff\`e nel distributore automatico, perci\`o ci saranno diverse \emph{azioni di tipo $B$}.

La logica lineare apre nuovi interessanti scenari sulla visione dei connettivi classici: ad esempio esistono \emph{due} congiunzioni ($\otimes$ o ``per'', inteso in senso di moltiplicazione, ed $\with$ o ``con'') corrispondenti a due usi radicalmente differenti della parola ``e''. Ambedue le congiunzioni esprimono la disponibilit\`a di due azioni; ma nel caso di $\otimes$, saranno fatte tutt'e due, mentre nel caso di $\with$, solo una delle due sar\`a eseguita (ma noi potremo decidere quale). Ad esempio, siano $A$, $B$, $C$:
\begin{center}
\begin{tabular}{lll}
	$A$ & : & spendere una moneta nel DAB \\
	$B$ & : & prendere un caff\`e \\
	$C$ & : & prendere un t\`e 
\end{tabular}
\end{center}
Data un'azione di tipo $A \multimap B$ e una di tipo $A \multimap C$, non sar\`a possibile formare un'azione di tipo $A \multimap B \otimes C$, poich\'e per una moneta non si potr\`a mai avere ci\`o che ne costa due (sar\`a invece possibile formare un'azione di tipo $A \otimes A \multimap B \otimes C$, cio\`e avere due bevande in cambio di due monete). Comunque potremmo sempre produrre un'azione di tipo $A \multimap B \with C$ come sovrapposizione delle due. Per eseguire quest'ultima azione dovremmo prima scegliere tra le possibili azioni che vogliamo produrre e in seguito effettuare quella scelta. Questo \`e analogo a quanto accade col costrutto~$\mathsf{if}~\dots~\mathsf{then}~\dots~\mathsf{else}~\dots$~ben noto in informatica: infatti, sia la parte~$\mathsf{then}$~\dots~che quella~$\mathsf{else}$~\dots~sono disponibili, ma solo una di esse verr\`a eseguita. Per quanto ``$\with$'' abbia delle ovvie caratteristiche disgiuntive, sarebbe tecnicamente errato vederlo come disgiunzione: infatti in logica lineare sia $A \with B \multimap A$, sia $A \with B \multimap B$ sono dimostrabili.

In logica lineare, in maniera del tutto speculare, abbiamo due disgiunzioni, che sono $\oplus$ o ``pi\`u'', e $\parr$ o ``par'' (mnemonico per \emph{parallelo}). $\oplus$ \`e il duale di ``$\with$'' ed esprime la presenza di due opzioni: in questo caso per\`o, non sar\`a possibile scegliere quale delle due eseguire. La differenza tra $\with$ e $\oplus$ \`e la stessa che c'\`e in informatica tra nondeterminismo esterno ed interno. Infine $\parr$ \`e il duale di $\otimes$.

Il pi\`u importante connettivo lineare \`e la \emph{negazione lineare} $\;\mneg{\cdot}\;$ o ``nil''. Poich\'e l'implicazione lineare si pu\`o sempre riscrivere come $\mneg{A} \parr B$, ``nil'' \`e l'unica operazione negativa della logica lineare. La negazione lineare si comporta come la trasposizione in algebra lineare, esprime cio\`e \emph{dualit\`a}, ovverosia un cambio di prospettiva:
\begin{center}
	\emph{azione di tipo} $A$ = \emph{reazione di tipo} $\mneg{A}$
\end{center}
La propriet\`a principale di ``nil'' \`e che, come accade in logica classica, $\mneg{\mneg{A}}$ pu\`o essere identificato con $A$ stesso. A differenza della logica classica per\`o, la logica lineare gode di una \emph{semplice interpretazione costruttiva}. Il carattere involutivo di ``nil'' assicura il comportamento \emph{alla De Morgan} per tutti i connettivi ed i quantificatori, ad esempio:
$$
	\exists x.A \quad=\quad \mneg{(\forall x.\mneg{A})}
$$
che pu\`o sembrare insolito ad un primo sguardo, specialmente se consideriamo che l'esistenziale in logica lineare \`e un operatore \emph{effettivo}: tipicamente si dimostra $\exists x.A$ dimostrando $A[t/x]$ per un certo termine $t$. Questo comportamento di ``nil'' deriva dal fatto che $\mneg{A}$ nega (cio\`e \emph{reagisce con}) una singola azione di tipo $A$, mentre la negazione classica nega solo alcune (non specificate) iterazioni di $A$, che tipicamente porta ad una disgiunzione di lunghezza non specificata. La negazione lineare \`e da un lato pi\`u primitiva, e dall'altro pi\`u forte (e anche pi\`u difficile da trattare) di quella classica.

Grazie alla presenza degli esponenziali, la logica lineare \`e espressiva quanto quella classica o quella intuizionista. Di fatto \`e pi\`u espressiva. Qui bisogna essere cauti: \`e lo stesso problema della logica intuizionista, che \`e anch'essa ``pi\`u espressiva'' di quella classica. Tecnicamente il potere espressivo \`e equivalente: ma i connettivi della logica lineare possono esprimere in maniera primitiva cose che in logica classica possono essere espresse solo tramite complesse traduzioni \emph{ad hoc}. L'introduzione di nuovi connettivi \`e quindi la chiave di volta verso formalizzazioni pi\`u semplici ed efficaci; la restrizione a vari frammenti apre le frontiere a linguaggi con specifico potere espressivo, ad esempio con una complessit\`a computazionale nota~(\cite{Gir98, Laf02, DalBai06}).

Un notevole problema aperto \`e quello di trovare una versione convincente di logica lineare non-commutativa. Anche se molti convengono sul fatto che la non-commutativit\`a ha ragione d'esser considerata a questo livello (esistono svariati esempi di problemi intrinsecamente non-commutativi, si pensi all'operatore di prefisso del $\pi$-calcolo), semantiche non trivali di non-commutativit\`a non sono note. Unite all'introduzione di una semantica naturale, le metodologie per raggiungere un sistema non-commutativo potrebbero comportare un effettivo guadagno di potere espressivo, in relazione al caso commutativo.

\section{Calcolo dei sequenti lineari}

Definiamo la sintassi della logica lineare classica (o \textsf{CLL}, acronimo di Classical Linear Logic):

\begin{dfn}[Linguaggio \textsf{CLL}] Sia $\mathcal{P}$ un insieme infinito enumerabile di \emph{simboli proposizionali}. L'insieme degli \emph{atomi} $\mathcal{A}$ \`e cos\`i definito:
$$
	\mathcal{A} = \{ p, \mneg{p} \:|\: p \in \mathcal{P} \}
$$
dove $\mneg{\cdot}$ \`e una \emph{funzione di negazione primitiva sui simboli proposizionali}. La negazione si estende facilmente a tutti gli atomi definendo $\mneg{\mneg{p}} = p$ per ogni simbolo proposizionale negato $\mneg{p}$.

Siano $\mone, \bot, \mzero, \top \not\in \mathcal{A}$ simboli costanti o \emph{unit\`a}, e sia $a \in \mathcal{A}$. Il \emph{linguaggio \textsf{CLL} delle formule lineari classiche} \`e cos\`i definito:
$$
\begin{array}{llll}
	T & ::= & \mone \:|\: \bot \:|\: \mzero \:|\: \top \:|\: a & \quad\mbox{(termini)} \\
	P & ::= & T \:|\: P \oplus P \:|\: P \with P \:|\: P \otimes P \:|\: P \parr P \:|\: \oc P \:|\: \wn P & \quad\mbox{(formule)} 
\end{array}
$$
\end{dfn}
I connettivi $\otimes$, $\parr$, $\multimap$, insieme agli elementi neutri $\mone$ (relativamente a $\otimes$) e $\bot$ (relativamente a $\parr$) sono chiamati \emph{moltiplicativi}; i connettivi $\with$, $\oplus$, insieme agli elementi neutri $\top$ (relativamente a $\with$) e $\mzero$ (relativamente a $\oplus$) sono chiamati \emph{additivi}; i connettivi $\oc$ e $\wn$ sono chiamati \emph{esponenziali}. Questa notazione \`e stata scelta perch\'e facile da memorizzare: infatti essa suggerisce che $\otimes$ sia moltiplicativo e congiuntivo, con elemento neutro $\mone$, mentre $\oplus$ \`e additivo e disgiuntivo, con elemento neutro $\mzero$; inoltre, anche la distributivit\`a di $\otimes$ su $\oplus$ \`e suggerita dalla notazione.

La negazione si estende alle formule, come mostrato in Figura~\ref{fig:eq_lin}; inoltre l'implicazione lineare \`e definita con l'ausilio di negazione e connettivo ``par''.

\begin{figure}
$$
\begin{array}{rclcrcl}
	\mneg{\mone} &=& \bot & \qquad\qquad\qquad\qquad & \mneg{\bot} &=& \mone \\
	\mneg{\top} &=& \mzero & & \mneg{\mzero} &=& \top \\
	\mneg{P \otimes Q} &=& \mneg{P} \parr \mneg{Q} & & \mneg{P \parr Q} &=& \mneg{P} \otimes \mneg{Q} \\
	\mneg{P \with Q} &=& \mneg{P} \oplus \mneg{Q} & & \mneg{P \oplus Q} &=& \mneg{P} \with \mneg{Q} \\
	\mneg{\oc P} &=& \wn \mneg{P} & & \mneg{\wn P} &=& \oc \mneg{P}
\end{array}
$$
~\\
\centering$P \multimap Q = \mneg{P} \parr Q$
\caption{Definizione di negazione e implicazione lineari}
\label{fig:eq_lin}
\end{figure}

Procediamo mostrando in Figura~\ref{fig:ded_lin} un sistema deduttivo per la logica lineare reminiscente il calcolo dei sequenti di~\cite{Gen35}. Come visto in precedenza, un sequente \`e un'espressione $\Gamma \vdash \Delta$, in cui $\Gamma = P_1, \dots, P_n$ e $\Delta = Q_1, \dots, Q_m$ sono sequenze finite di formule. Il significato inteso di $\Gamma \vdash \Delta$ \`e:
$$
	P_1\mbox{ e }\dots\mbox{ e }P_n \quad\mbox{ implica }\quad Q_1\mbox{ oppure }\dots\mbox{ oppure }Q_m
$$
dove il senso di ``e'', ``implica'' e ``oppure'' devono essere specificati formalmente. I sequenti lineari sono ad un lato, cio\`e della forma $\vdash \Gamma$; sequenti nella forma generale $\Gamma \vdash \Delta$ si possono ``mimare'' usando $\vdash \mneg{\Gamma}, \Delta$.

\begin{figure}[t!]
\textbf{Identit\`a / taglio} \\\\
\begin{minipage}[m]{.5\textwidth}
	\centering$\vdash P, \mneg{P} \quad \mrule{id}$
\end{minipage}
\begin{minipage}[m]{.5\textwidth}
	\AxiomC{$\vdash \Gamma, P$}
	\AxiomC{$\vdash \mneg{P}, \Delta$}
	\RightLabel{$\mrule{cut}$}
	\BinaryInfC{$\vdash \Gamma, \Delta$}
	\centering\DisplayProof{} 
\end{minipage}\\\\
\textbf{Regole strutturali}
\begin{center}
	\AxiomC{$\vdash \Gamma, P, Q, \Delta$}
	\RightLabel{$\mrule{perm}$}
	\UnaryInfC{$\vdash \Gamma, Q, P, \Delta$}
	\DisplayProof{}
\end{center}
\textbf{Regole logiche} \\\\
\begin{minipage}[m]{.5\textwidth}
	\centering$\vdash \mone \quad \mrule{one}$
\end{minipage}
\begin{minipage}[m]{.5\textwidth}
	\AxiomC{$\vdash \Gamma$}
	\RightLabel{$\mrule{false}$}
	\UnaryInfC{$\vdash \Gamma, \bot$} 
	\centering\DisplayProof{}
\end{minipage}\\\\
\begin{minipage}[m]{.5\textwidth}
	\AxiomC{$\vdash \Gamma, P$}
	\AxiomC{$\vdash \Delta, Q$}
	\RightLabel{$\mrule{\otimes}$}
	\BinaryInfC{$\vdash \Gamma, \Delta, P \otimes Q$}
	\centering\DisplayProof{} 
\end{minipage}
\begin{minipage}[m]{.5\textwidth}
 	\AxiomC{$\vdash \Gamma, P, Q$}
	\RightLabel{$\mrule{\parr}$}
	\UnaryInfC{$\vdash \Gamma, P \parr Q$} 
	\centering\DisplayProof{}
\end{minipage}\\\\
\begin{minipage}[m]{.5\textwidth}
	\centering$\vdash \Gamma, \top \quad \mrule{true}$ 
\end{minipage}
\begin{minipage}[m]{.5\textwidth}
	\centering\textit{(nessuna regola per $\mzero$)}
\end{minipage}\\\\
\begin{minipage}[m]{.5\textwidth}
	\AxiomC{$\vdash \Gamma, P$}
	\AxiomC{$\vdash \Gamma, Q$}
	\RightLabel{$\mrule{\with}$}
	\BinaryInfC{$\vdash \Gamma, P \with Q$}
	\centering\DisplayProof{}
\end{minipage}
\begin{minipage}[m]{.5\textwidth}
	\begin{center}
		\AxiomC{$\vdash \Gamma, P$}
		\RightLabel{$\mrule[l]{\oplus}$}
		\UnaryInfC{$\vdash \Gamma, P \oplus Q$}
		\DisplayProof{} \\\vspace{.5em}
		\AxiomC{$\vdash \Gamma, Q$}
		\RightLabel{$\mrule[r]{\oplus}$}
		\UnaryInfC{$\vdash \Gamma, P \oplus Q$}
		\DisplayProof{}
	\end{center}
\end{minipage}\\~\vspace{.8em}~\\
\begin{minipage}[m]{.5\textwidth}
	\AxiomC{$\vdash \wn \Gamma, P$}
	\RightLabel{$\mrule{\oc}$}
	\UnaryInfC{$\vdash \wn \Gamma, \oc P$}
	\centering\DisplayProof{}
\end{minipage}
\begin{minipage}[m]{.5\textwidth}
	\AxiomC{$\vdash \Gamma$}
	\RightLabel{$\mrule{weak}$}
	\UnaryInfC{$\vdash \Gamma, \wn P$}
	\centering\DisplayProof{}
\end{minipage}\\~\vspace{.3em}~\\
\begin{minipage}[m]{.5\textwidth}
	\AxiomC{$\vdash \Gamma, P$}
	\RightLabel{$\mrule{drlc}$}
	\UnaryInfC{$\vdash \Gamma, \wn P$}
	\centering\DisplayProof{}
\end{minipage}
\begin{minipage}[m]{.5\textwidth}
	\AxiomC{$\vdash \Gamma, \wn P, \wn P$}
	\RightLabel{$\mrule{cntr}$}
	\UnaryInfC{$\vdash \Gamma, \wn P$}
	\centering\DisplayProof{}
\end{minipage}
\caption{Sistema deduttivo per \textsf{CLL}}
\label{fig:ded_lin}
\end{figure}

Nel calcolo dei sequenti lineari abbiamo rimosso le regole strutturali di indebolimento (\`e sempre possibile aggiungere una formula nella premessa o nella conclusione del sequente) e contrazione (la molteplicit\`a di una formula non conta) in virt\`u delle critiche mosse dalla scuola lineare. La possibilit\`a di utilizzare queste operazioni \`e tuttavia ripristinata grazie all'introduzione degli operatori $\oc$ e $\wn$.

Identit\`a e taglio restano invariate rispetto al Sistema \textsf{LKp}, cos\`i come la regola di permutazione. 

La situazione \`e diversa per quanto riguarda la congiunzione: come avveniva in precedenza, per dimostrare una congiunzione tra $P$ e $Q$ bisogna aver dimostrato separatamente sia $P$ che $Q$, ma in assenza della regola d'indebolimento, possiamo distinguere il caso in cui le dimostrazioni di $P$ e $Q$ siano fatte nello stesso ambiente ($\Gamma$ nella regola $\mrule{\with}$), o in ambienti diversi ($\Gamma$ e $\Delta$ in $\mrule{\otimes}$).

Un ragionamento analogo vale per la disgiunzione: in logica lineare possiamo infatti distinguere il caso in cui, nel dimostrare la disgiunzione di $P$ e $Q$, disponiamo solo di $P$ (regola $\mrule[l]{\oplus}$), solo di $Q$ (regola $\mrule[r]{\oplus}$), e quello in cui abbiamo ambedue (regola $\mrule{\parr}$).

\`E possibile introdurre nuove formule, indebolendo il sequente, a patto che queste siano ``marcate'' con l'operatore $\oc$; per questa classe di formule (chiamate formule ``perch\'e non'' o ``why not'') la molteplicit\`a non \`e rilevante: inoltre ogni formula pu\`o essere trasformata in una \emph{why not} grazie alla regola di \emph{derelizione} $\mrule{drlc}$. La regola di \emph{promozione} $\mrule{\oc}$ permette ``aumentare'' la molteplicit\`a di una formula di una quantit\`a arbitraria.

Il Sistema cos\`i ottenuto gode di buone propriet\`a, oltre ad avere una granularit\`a pi\`u fine rispetto alla logica classica. Per \textsf{CLL} \`e possibile dimostrare la \emph{cut elimination}:
\begin{thm}[Hauptsatz lineare]
La regola di taglio lineare \`e eliminabile da \textsf{CLL}.
\end{thm}
\begin{proof}
La dimostrazione segue un argomento del tutto analogo a quello visto per la logica classica nel Teorema~\ref{thm:haupt_lk}, con alcune semplificazioni dovute al fatto di non dover trattare le usuali regole strutturali.
\end{proof}

Nuovamente, la dimostrazione risultante dalla procedura di cut elimination non \`e univocamente determinata, a causa della \emph{permutazione delle regole}. Ad esempio, nella derivazione:
\begin{center}
	\AxiomC{$\vdash \Gamma, P$}
	\RightLabel{$\mrule{\rho}$}
	\UnaryInfC{$\vdash \Gamma', P$}
	\AxiomC{$\vdash \mneg{P}, \Delta$}
	\RightLabel{$\mrule{\sigma}$}
	\UnaryInfC{$\vdash \mneg{P}, \Delta'$}
	\RightLabel{$\mrule{cut}$}
	\BinaryInfC{$\vdash \Gamma', \Delta'$}
	\DisplayProof{}
\end{center}
non c'\`e nessun modo ovvio di eliminare l'applicazione di $\mrule{cut}$, poich\'e le regole $\mrule{\rho}$ e $\mrule{\sigma}$ non agiscono su $P$ e $\mneg{P}$. Quindi l'idea \`e di ``spingere il cut verso l'alto'':
\begin{center}
	\AxiomC{$\vdash \Gamma, P$}
	\AxiomC{$\vdash \mneg{P}, \Delta$}
	\RightLabel{$\mrule{cut}$}
	\BinaryInfC{$\vdash \Gamma, \Delta$}
	\RightLabel{$\mrule{\rho}$}
	\UnaryInfC{$\vdash \Gamma', \Delta$}
	\RightLabel{$\mrule{\sigma}$}
	\UnaryInfC{$\vdash \Gamma', \Delta'$}
	\DisplayProof{}
\end{center}
ma cos\`i facendo abbiamo arbitrariamente privilegiato la regola $\mrule{\rho}$ rispetto alla $\mrule{\sigma}$, mentre l'altra scelta:
\begin{center}
	\AxiomC{$\vdash \Gamma, P$}
	\AxiomC{$\vdash \mneg{P}, \Delta$}
	\RightLabel{$\mrule{cut}$}
	\BinaryInfC{$\vdash \Gamma, \Delta$}
	\RightLabel{$\mrule{\sigma}$}
	\UnaryInfC{$\vdash \Gamma, \Delta'$}
	\RightLabel{$\mrule{\rho}$}
	\UnaryInfC{$\vdash \Gamma', \Delta'$}
	\DisplayProof{}
\end{center}
sarebbe stata altrettanto legittima. La scelta compiuta in questo passo della cut elimination \`e in generale irreversibile: a meno che $\mrule{\rho}$ o $\mrule{\sigma}$ non siano successivamente eliminate, non sar\`a pi\`u possibili scambiarle. Per eliminare questa fonte di non-determinismo, fu introdotto in~\cite{Gir87} un nuovo formalismo, basato sulla teoria dei grafi, e chiamato \emph{Proof Nets}.

Il Sistema \textsf{CLL} non \`e l'unico rappresentante della classe delle logiche lineari. Vista l'ampia gamma di regole che possiede, questo Sistema pu\`o essere suddiviso in moduli con interessanti propriet\`a computazionali: il punto \`e proprio che la logica lineare \`e in grado di trattare naturalmente con le risorse (rappresentate dalla \emph{molteplicit\`a} delle formule), e per questo ci si riferisce ad essa con l'appellativo \emph{resource-conscious}; in informatica avere coscienza delle risorse significa saper distinguere varie classi di complessit\`a.

Tra i vari sottosistemi, quelli che maggiormente divergono dalla logica classica (e intuizionista), sono chiamati \textsf{LLL} (Light Linear Logic) ed \textsf{ELL} (Elementary Linear Logic)~--~\cite{Gir95a, DanJoi01}. Essi seguono dalla scoperta che, in assenza degli esponenziali, la procedura di eliminazione dei tagli pu\`o essere eseguita in tempo lineare.

Per i nostri scopi, ci occuperemo esclusivamente del frammento moltiplicativo: questo \`e il pi\`u semplice ed il pi\`u piccolo frammento di logica lineare (fu anche il primo che venne trasposto nelle Proof Nets, per via della sua semplicit\`a). Nella fattispecie tratteremo d'ora in avanti il Sistema in Figura~\ref{fig:sys_mllmix}, chiamato \textsf{MLL+mix}, cio\`e \emph{Multiplicative Linear Logic} con l'aggiunta della regola \textsf{mix} che ``fonde'' i sequenti provenienti da due diversi sottoalberi di derivazione. La negazione \`e definita dalle leggi di De Morgan:
\begin{eqnarray*}
	\mneg{P \otimes Q} & = & \mneg{P} \parr \mneg{Q} \\
	\mneg{P \parr Q} & = & \mneg{P} \otimes \mneg{Q}
\end{eqnarray*}

\begin{figure}[t!]
\centering\textbf{Grammatica di \textsf{MLL+mix}}\\
\begin{minipage}[m]{.48\textwidth}
$$
	P ::= a \:|\: \mneg{P} \:|\: P \otimes P \:|\: P \parr P
$$
\end{minipage}
\begin{minipage}[m]{.48\textwidth}
	\vspace{1.5em}
	\centering\footnotesize{(con $a \in \mathcal{A}$ infinit\`a numerabile \\ di simboli proposizionali)}
	\vspace{1em}
\end{minipage}
\begin{center}
	\textbf{Sistema deduttivo}\\
	~\\
	$\vdash P, \mneg{P} \; \mrule{id}$
	\AxiomC{$\vdash \Gamma, P$}
	\AxiomC{$\vdash Q, \Delta$}
	\RightLabel{$\mrule{\otimes}$}
	\BinaryInfC{$\vdash \Gamma, \Delta, P \otimes Q$} 
	\DisplayProof{} 
 	\AxiomC{$\vdash \Gamma, P, Q$}
	\RightLabel{$\mrule{\parr}$}
	\UnaryInfC{$\vdash \Gamma, P \parr Q$} 
	\DisplayProof{}
	\AxiomC{$\vdash \Gamma$}
	\AxiomC{$\vdash \Delta$}
	\RightLabel{$\mrule{mix}$}
	\BinaryInfC{$\vdash \Gamma, \Delta$}
	\DisplayProof{}
\end{center}
\caption{Sistema \textsf{MLL+mix}}
\label{fig:sys_mllmix}
\end{figure}

Avremo modo di osservare una natuale corrispondenza di questo Sistema ed il suo corrispettivo in deep inference: il Sistema \textsf{LBV} di~\cite{Gug02}.
\newpage

\section{Sistema LBV}

\`E il pi\`u semplice Sistema deep inference concepibile: un calcolo proposizionale composto da \emph{due operatori duali} (del tutto simili a congiunzione e disgiunzione classici), una \emph{negazione auto-duale} alla De Morgan e una \emph{unit\`a logica}. Come nel caso classico, le formule sono considerate uguali modulo una relazione di equivalenza. Le regole sono l'assioma $\mrule{ax}$ per l'unit\`a e la regola di scambio $\mrule{s}$; inoltre la regola d'identit\`a (chiamata anche \emph{regola d'interazione}) e quella di taglio (o \emph{regola di co-interazione}), nella versione generalizzata:
$$
	\vlinf{}{\mruledn{i}}{\mctx*{C}\mpar{P, \mneg{P}}}{\mctx{C}{\circ}}
	\qquad\qquad
	\vlinf{}{\mruleup{i}}{\mctx{C}{\circ}}{\mctx*{C}\mand{P, \mneg{P}}}
$$
che tuttavia, come prima, possono essere ridotte alla loro forma atomica, dando origine al Sistema di Figura~\ref{fig:sys_lbv}.

\begin{thm}[Localit\`a di \textsf{LBV+cut}]\label{thm:loc_lbv}
La regola $\mruledn{i}$ \`e derivabile da $\{\mruledn{ai}, \mrule{s}\}$. Dualmente, la regola $\mruleup{i}$ \`e derivabile da $\{\mruleup{ai}, \mrule{s}\}$.
\end{thm}
\begin{proof}
Data l'istanza $\:\vlinf{}{\mruledn{i}}{\mctx*{C}\mpar{P, \mneg{P}}}{\mctx{C}{\circ}}\:$ procediamo per induzione strutturale su $P$. Il caso duale $\mruleup{i}$ si dimostra allo stesso modo.
\begin{description}
	\item[Casi base] ~
	\begin{enumerate}
		\item $P = \circ$. Ovvio, poich\'e $\mctx*{C}\mpar{P, \mneg{P}} = \mctx{C}{\circ}$.
		\item $P$ \`e un atomo: Allora $\mruledn{i}$ \`e un'istanza di $\mruledn{ai}$.
	\end{enumerate}
	\item[Casi induttivi] ~
	\begin{enumerate}[resume]
		\item $P = \mpar{R, S}$. Per ipotesi induttiva, abbiamo due derivazioni $\Phi_R$ e $\Phi_S$:
		$$
			\vlderd{\Phi_R}{\{\mruledn{ai}, \mvlrule{s}\}}{\mctx*{C}\mpar{R, \mneg{R}}}{\mctx{C}{\circ}}
			\qquad\qquad
			\vlderd{\Phi_S}{\{\mruledn{ai}, \mvlrule{s}\}}{\mctx*{D}\mpar{S, \mneg{S}}}{\mctx{D}{\circ}}
		$$
		da cui \`e possibile ottenere:
		$$
			\vlderivation {
				\vlin{}{\mvlrule{s}}{\mctx*{C}\mpar{R, S, \mand{\mneg{R}, \mneg{S}}}}{
					\vlin{}{\mvlrule{s}}{\mctx*{C}\mpar{R, \mand{\mneg{R}, \mpar{S, \mneg{S}}}}}{
						\vldd{\Phi_S}{\{\mruledn{ai}, \mvlrule{s}\}}{\mctx*{C}\mand{\mpar{R, \mneg{R}}, \mpar{S, \mneg{S}}}}{
							\vldd{\Phi_R}{\{\mruledn{ai}, \mvlrule{s}\}}{\mctx*{C}\mpar{R, \mneg{R}} = \mctx*{C}\mand{\mpar{R, \mneg{R}}, \circ}}{
								\vlhy{\mctx{C}{\circ}}
							}
						}
					}
				}
			}
		$$
		\item Infine, il caso $P = \mand{R, S}$ \`e analogo al precedente.
	\end{enumerate}
\end{description}
\end{proof}

\begin{figure}[h!]
\vspace{0.5em}
\textbf{Sintassi} \\
\begin{minipage}[m]{.48\textwidth}
	\vspace{1em}
	\centering $P ::= \circ \:|\: a \:|\: \mneg{a} \:|\: \mpar{P, P} \:|\: \mand{P, P}$
\end{minipage}
\begin{minipage}[m]{.48\textwidth}
\begin{center}
	\vspace{1em}
	\footnotesize{(con $a \in \mathcal{A}$ infinit\`a numerabile \\ di simboli proposizionali)}
\end{center}
\end{minipage}
\\\\\textbf{Sistema deduttivo}
\begin{center}
	$\circ \quad \mrule{ax}$
	\qquad
	$\vlinf{}{\mruledn{ai}}{\mctx*{C}\mpar{a, \mneg{a}}}{\mctx{C}{\circ}}$
	\qquad
	$\vlinf{}{\mruleup{ai}}{\mctx{C}{\circ}}{\mctx*{C}\mand{a, \mneg{a}}}$
	\qquad
	$\vlinf{}{\mvlrule{s}}{\mctx*{C}\mpar{\mand{P, Q}, R}}{\mctx*{C}\mand{P, \mpar{Q, R}}}$
\end{center}
\vspace{.25em}
\begin{minipage}[t]{.5\textwidth}
\textbf{Associativit\`a}
\begin{eqnarray*}
	\mpar{\mpar{P, Q}, R} & = & \mpar{P, \mpar{Q, R}} \\
	\mand{\mand{P, Q}, R} & = & \mand{P, \mand{Q, R}} 
\end{eqnarray*}
\textbf{Commutativit\`a}
\begin{eqnarray*}
	\mpar{P, Q} & = & \mpar{Q, P} \\
	\mand{P, Q} & = & \mand{Q, P}
\end{eqnarray*}
\textbf{Unit\`a}
$$
	\mpar{P, \circ} = \mand{P, \circ} = P
$$
\end{minipage}
\begin{minipage}[t]{.49\textwidth}
\textbf{Negazione}
\begin{eqnarray*}
	\mneg{\circ} & = & \circ \\
	\mneg{\mpar{P, Q}} & = & \mand{\mneg{P}, \mneg{Q}} \\
	\mneg{\mand{P, Q}} & = & \mpar{\mneg{P}, \mneg{Q}} \\
	\mneg{\mneg{P}} & = & P
\end{eqnarray*}
\textbf{Congruenza}
\begin{center}
	$P = P$ \qquad
	\AxiomC{$P = R$}\AxiomC{$R = Q$}\BinaryInfC{$P = Q$}\DisplayProof{}
\end{center}
\begin{center}
	\AxiomC{$P = Q$}\UnaryInfC{$Q = P$}\DisplayProof{} \qquad
	\AxiomC{$P = Q$}\UnaryInfC{$\mctx{C}{P} = \mctx{C}{Q}$}\DisplayProof{} 
\end{center}
\end{minipage}
\vspace{.45em}
\caption{Sistema \textsf{LBV+cut}, equivalenza tra formule e negazione}
\label{fig:sys_lbv}
\end{figure}

Esiste una corrispondenza 1:1 tra Sistema \textsf{MLL+mix} e Sistema \textsf{LBV}.

\begin{dfn}[Trasformazioni \textsf{LBV}$\leftrightarrow$\textsf{MLL}]
$$
\begin{array}{rclcrcl}
	\multicolumn{3}{c}{\transV{\cdot} \::\: \mathscr{L}_\mathsf{MLL}\rightarrow\mathscr{L}_\mathsf{LBV}} & \qquad\qquad & \multicolumn{3}{c}{\transL{\cdot} \::\: \mathscr{L}_\mathsf{LBV}\rightarrow\mathscr{L}_\mathsf{MLL}} \\
	\transV{a} & = & a & ~ & \transL{a} & = & a \\
	\transV{P \parr Q} & = & \mpar{\transV{P}, \transV{Q}} & ~ & \transL{\mpar{P, Q}} & = & \transL{P} \parr \transL{Q} \\
	\transV{P \otimes Q} & = & \mand{\transV{P}, \transV{Q}} & ~ & \transL{\mand{P,Q}} & = & \transL{P} \otimes \transL{Q}
\end{array}
$$
Inoltre la definizione di $\transV{\cdot}$ si estende facilmente ai sequenti:
$$
	\transV{\vdash P_1, \ldots, P_n} = \mpar{\transV{P_1}, \ldots, \transV{P_n}}
$$
per $n > 0$; per $n=0$ si pone $\transV{\vdash} = \circ$.
\end{dfn}

\begin{thm}[Equivalenza di \textsf{LBV} e \textsf{MLL+mix}]~\\
\begin{enumerate}[label=\roman*)]
	\item Se il sequente $\vdash P$ \`e dimostrabile in \textsf{MLL+mix}, allora la struttura $\transV{P}$ \`e dimostrabile in \textsf{LBV}.
	\item Se la struttura $P$ (in forma normale, con $P \not= \circ$) \`e dimostrabile in \textsf{LBV}, allora il sequente $\vdash \transL{P}$ \`e dimostrabile in \textsf{MLL+mix}.
\end{enumerate}
\end{thm}

Questo teorema (dimostrato per la prima volta in~\cite{Gug02}) stabilisce una correlazione tra calcolo dei sequenti e calcolo delle strutture; come gi\`a accennato, \`e possibile conseguire un risultato analogo per il Sistema \textsf{SKS}, ma, ad esempio, la propriet\`a di localit\`a non vale per la logica classica proposizionale nel calcolo dei sequenti.

\newpage

\subsection{Eliminazione del taglio} 

L'argomento calssico per dimostrare l'eliminazione del taglio nel calcolo dei sequenti, risiede nel fatto che, quando le formule principali del taglio sono introdotte in entrambi i rami, esse determinano che regole saranno applicate immediatamente sopra a quella di taglio. Questo \`e conseguenza del fatto che le formule hanno un connettivo principale, e le regole logiche si basano solo su quello, e su nessun'altra propriet\`a delle formule.

Questo fatto non vale nel calcolo delle strutture. Per dimostrare la cut elimination nel Sistema \textsf{LBV}, occorre appoggiarsi ad un'altra propriet\`a, scoperta in~\cite{Gug02}, e chiamata \emph{scissione} o \emph{splitting}. Essa \`e una generalizzazione della tecnica vista nella dimostrazione di eliminazione del taglio per il sistema \textsf{SKS}. Si consideri la dimostrazione del sequente:
$$
	\vdash \mctx{C}{P \otimes Q}, \Gamma
$$
dove $\mctx{C}{P \otimes Q}$ \`e una formula contenente la sottoformula $P \otimes Q$. Sappiamo per certo che nella dimostrazione ci deve essere un'istanza della regola $\mrule{\otimes}$ che scinde $P$ da $Q$ assieme ai rispettivi contesti. Siamo nella seguente situazione:
$$
	\vlderivation{
		\vltr{\Psi}{\vdash \mctx{C}{P \otimes Q}, \Gamma}{\vlhy{\quad}}{
			\vliin{}{\mvlrule{\otimes}}{\vdash P \otimes Q, \Gamma', \Gamma''}{
				\vltr{\Pi_1}{\vdash P, \Gamma'}{\vlhy{~}}{\vlhy{\quad}}{\vlhy{~}}
			}{
				\vltr{\Pi_2}{\vdash Q, \Gamma''}{\vlhy{~}}{\vlhy{\quad}}{\vlhy{~}}
			}
		}{\vlhy{~}}
	}
	\qquad\mbox{corrispondente a}\qquad
	\vlderivation{
		\vldd{\Psi}{}{\mpar{\mctx*{C}\mand{P,Q}, \Gamma}}{
			\vlin{}{\mvlrule{s}}{\mpar{\mand{P,Q},\Gamma',\Gamma''}}{
				\vlin{}{\mvlrule{s}}{\mpar{\mand{\mpar{P,\Gamma'},Q},\Gamma''}}{
					\vldd{\Pi_1}{}{\mand{\mpar{P,\Gamma'},\mpar{Q,\Gamma''}}}{
						\vlpr{\Pi_2}{}{\mpar{Q,\Gamma''}}
					}
				}
			}
		}
	}
$$
La derivazione $\Psi$ implementa lo splitting, che \`e ottenuto in due passi:
\begin{enumerate}
	\item riduzione del contesto: se $\mctx{C}{P}$ \`e dimostrabile, allora $\mctx*{C}$ pu\`o essere ridotto, risalendo nella dimostrazione, ad un contesto $\mpar{\emptyctx, S}$, per un $S$ opportuno, tale che $\mpar{P, S}$ \`e dimostrabile. Nell'esempio sopra, $\mpar{\mctx{C}{\emptyctx}, \Gamma}$ \`e ridotto a $\mpar{\emptyctx, S}$ per un certo $S$;
	\item scissione di superficie: se $\mpar{\mand{R,T},P}$ \`e dimostrabile, allora $P$ pu\`o essere ridotto, risalendo nella dimostrazione, ad una struttura $\mpar{P_1, P_2}$ tali che $\mpar{R,P_1}$ e $\mpar{T,P_2}$ sono dimostrabili. Nell'esempio, $S$ \`e scisso in $\mpar{\Gamma', \Gamma''}$.
\end{enumerate}

Grazie al Teorema di splitting, abbiamo la capacit\`a di scindere un copar in due dimostrazioni, una per ogni rispettiva sottoformula: l'importanza di tale capacit\`a ai fini della cut elimination, diventa chiara se consideriamo la regola di taglio nel Sistema \textsf{LBV}:
$$
	\vlinf{}{\mruleup{ai}}{\mctx{C}{\circ}}{\mctx*{C}\mand{a, \mneg{a}}}
$$
Il contesto $\mctx*{C}$ viene scisso in due componenti $S_1$ e $S_2$ tali che esistono le dimostrazioni $\:\vlproofd{\Pi_1}{}{\mpar{a, S_1}}\:$ e $\:\vlproofd{\Pi_2}{}{\mpar{\mneg{a}, S_2}}\:$. Ora possiamo sfruttare il fatto che gli atomi $a$ e $\mneg{a}$ possono essere introdotti, rispettivamente nelle conclusioni di $\Pi_1$ e $\Pi_2$, solo mediante applicazioni della regola $\mruledn{ai}$ (fatto non vero, ad esempio, nel Sistema \textsf{KS}). A questo punto siamo in grado di isolare il segmento di dimostrazione che introduce gli atomi che verranno in seguito rimossi dal taglio, e possiamo pertanto trasformare questa sezione per bloccare il flusso degli atomi diretti al taglio sul nascere.

\begin{thm}[Shallow splitting]\label{thm:shallow_split}
Se $\mpar{\mand{R,T},P}$ \`e dimostrabile in \textsf{LBV}, allora esistono $P_1$ e $P_1$ tali che:
$$
	\vlderd{}{\mathsf{LBV}}{P}{\mpar{P_1,P_2}}
$$
e $\mpar{R,P_1}$ e $\mpar{T,P_2}$ siano entrambi dimostrabili in \textsf{LBV}.
\end{thm}
\begin{proof}
Consideriamo l'ordinamento lessicografico sui naturali:
$$
	(m', n') \prec (m, n) \quad\mbox{ sse }\quad m' < m \mbox{ oppure } (m' = m \mbox{ e } n' < n)
$$
Vogliamo procedere per induzione completa su due quantit\`a: l'altezza della dimostrazione di $\mpar{\mand{R,T},P}$ e la \emph{lunghezza delle formule}, definita induttivamente da:
$$
	\begin{array}{ccccccc}
		\msize{\circ} & = & 0 & \qquad\qquad & \msize{\mpar{P,Q}} & = & \msize{P} + \msize{Q} \\
		\msize{a} & = & 1 & \qquad\qquad & \msize{\mand{P,Q}} & = & \msize{P} + \msize{Q} \\
		\msize{\mneg{P}} & = & \msize{P}
	\end{array}
$$
Dato il meta-enunciato:
$$
\begin{array}{lll}
	C(m,n) & = & \forall R,T,P. \\
	& ~ &\quad \forall (m',n') \preceq (m,n). \\
	~ & ~ &\quad\quad \left(m' = \msizeb{\mpar{\mand{R,T},P}} \quad\bigwedge\quad \mbox{esiste } \vlproofd{}{}{\mpar{\mand{R,T},P}} \mbox{ con altezza $n'$}\;\right) \\
	~ & ~ &\quad\quad \mbox{\scalebox{1.5}{$\Rightarrow$}}\quad \exists P_1,P_2.\left(\;\vlderd{}{}{P}{\mpar{P_1,P_2}} \quad\bigwedge\:\quad \vlproofd{}{}{\mpar{R,P_1}} \:\quad\bigwedge\:\quad \vlproofd{}{}{\mpar{T,P_2}} \;\right)
\end{array}
$$
il teorema \`e equivalente a $\forall m,n.C(m,n)$. Per ipotesi induttiva possiamo suppore di avere una dimostrazione di $C(m',n')$ per ogni $(m',n') \prec (m,n)$.

La lunghezza di $\mpar{\mand{R,T},P}$ \`e $m$ e l'altezza della sua dimostrazione \`e $n$. Consideriamo l'istanza dell'ultima regola di questa dimostrazione:
$$
	\vlderivation{
		\vlin{}{\mvlrule{\rho}}{\mpar{\mand{R,T},P}}{
			\vlpd{}{}{Q}
		}
	}
$$
Procediamo per casi su $\mrule{\rho}$ (assumiamo sempre $P \not= \circ$ e $R \not= T \not= \circ$, perch\'e in questi casi il teorema vale banalmente):
\begin{enumerate}
	\item $\mrule{\rho} = \mruledn{ai}$. Questa regola si pu\`o applicare in tre diversi modi:
	\begin{enumerate}[label=\arabic{enumi}.\arabic*.]
		\item all'interno di $R$, cio\`e:
		$$
			\vlderivation{
				\vlin{}{\mruledn{ai}}{\mpar{\mand{R,T},P}}{
					\vlpd{}{}{\mpar{\mand{R',T},P}}
				}
			}
		$$
		Per ipotesi induttiva esistono $P_1$, $P_2$ tali che:
		$$
			\vlderd{}{}{P}{\mpar{P_1, P_2}} \qquad,\qquad \vlproofd{}{}{\mpar{R',P_1}} \qquad,\qquad \vlproofd{}{}{\mpar{T,P_2}}
		$$
		\`E sufficiente applicare $\mruledn{ai}$ in coda alla dimostrazione di $\mpar{R',P_1}$ per ottenere:
		$$
			\vlderivation{
				\vlin{}{\mruledn{ai}}{\mpar{R,P_1}}{
					\vlpd{}{}{\mpar{R',P_1}}
				}
			}
		$$
		\item all'interno di $T$. Analogo al caso precedente.
		\item all'interno di $P$, cio\`e:
		$$
			\vlderivation{
				\vlin{}{\mruledn{ai}}{\mpar{\mand{R,T},P}}{
					\vlpd{}{}{\mpar{\mand{R,T},P'}}
				}
			}
		$$
		Anche in questo caso procediamo per induzione diretta, a meno di un'applicazione di $\mruledn{ai}$ in:
		$$
			\vlderivation{
				\vlin{}{\mruledn{ai}}{P}{
					\vldd{}{\;\mbox{\footnotesize{(per I.H.)}}}{P'}{\vlhy{\mpar{P_1,P_2}}}
				}
			}
		$$
	\end{enumerate}
	\item $\mrule{\rho} = \mrule{s}$. Se la regola si applica all'interno di $R$, $T$ o $P$, procediamo in maniera del tutto analoga a quanto fatto nel caso $\mrule{\rho} = \mruledn{ai}$. Esistono altre due possibilit\`a:
	\begin{enumerate}[label=\arabic{enumi}.\arabic*.]
		\item $R=\mand{R',R''}$, $T=\mand{T',T''}$, $P=\mpar{P',P''}$ e:
		$$
			\vlderivation{
				\vlin{}{\mvlrule{s}}{\mpar{\mand{R',R'',T',T''}, P', P''}}{
					\vlpd{}{}{\mpar{\mand{\mpar{\mand{R', T'}, P'}, R'', T''}, P''}}
				}
			}
		$$
		Possiamo applicare l'ipotesi induttiva per ottenere:
		$$
			\vlderd{}{}{P''}{\mpar{P_1', P_2'}} \qquad,\qquad \vlproofd{\Pi_1}{}{\mpar{\mand{R',T'}, P', P_1'}} \qquad,\qquad \vlproofd{\Pi_2}{}{\mpar{\mand{R'', T''}, P_2'}}
		$$
		Ora possiamo applicare nuovamente l'ipotesi induttiva su $\Pi_1$ e su $\Pi_2$: infatti, anche se non conosciamo l'altezza di queste due dimostrazioni, sappiamo che la loro dimensione \`e inferiore a:
		$$
			\msizeb{\mpar{\mand{R',R'',T',T''}, P', P''}}
		$$
		perch\'e per ipotesi, l'istanza di $\mrule{s}$ non \`e triviale. Pertanto abbiamo per ipotesi induttiva:
		$$
			\begin{array}{ccccc}
				\vlderd{}{}{\mpar{P',P_1'}}{\mpar{P_1'',P_2''}} &\quad,\quad& \vlproofd{}{}{\mpar{R',P_1''}} &\quad,\quad& \vlproofd{}{}{\mpar{T',P_2''}} \\\\
				\vlderd{}{}{P_2'}{\mpar{P_1''',P_2'''}} &,& \vlproofd{}{}{\mpar{R'',P_1'''}} &,& \vlproofd{}{}{\mpar{T'',P_2'''}}
			\end{array}
		$$
		Ora, ponendo $P_1 = \mpar{P_1'', P_1'''}$ e $P_2 = \mpar{P_2'', P_2'''}$ otteniamo:
		$$
			\vlderivation{
				\vlde{}{}{\mpar{P',P''}}{
					\vlde{}{}{\mpar{P',P_1',P_2'}}{
						\vlde{}{}{\mpar{P_1'',P_2'',P_2'}}{
							\vlhy{\mpar{P_1'',P_2'',P_1''',P_2'''}}
						}
					}
				}
			}
			\quad,\quad
			\vlderivation{
				\vlin{}{\mrule{s}}{\mpar{\mand{R',R''}, P_1'', P_1'''}}{
					\vldd{}{}{\mpar{\mand{\mpar{R',P_1''}, R''}, P_1'''}}{
						\vlpd{}{}{\mpar{R'',P_1'''}}
					}
				}
			}
			\quad,\quad
			\vlderivation{
				\vlin{}{\mrule{s}}{\mpar{\mand{T',T''}, P_2'', P_2'''}}{
					\vldd{}{}{\mpar{\mand{\mpar{T',P_2''}, T''}, P_2'''}}{
						\vlpd{}{}{\mpar{T'',P_2'''}}
					}
				}
			}
		$$

		\item $P = \mpar{\mand{P',P''},U',U''}$ e:
		$$
			\vlderivation{
				\vlin{}{\mrule{s}}{\mpar{\mand{R,T},\mand{P',P''},U',U''}}{
					\vlpd{}{}{\mpar{\mand{\mpar{\mand{R,T},P',U'},P''},U''}}
				}
			}
		$$
		Per ipotesi induttiva abbiamo:
		$$
			\vlderd{}{}{U''}{\mpar{U_1,U_2}} \qquad,\qquad \vlproofd{}{\Pi}{\mpar{\mand{R,T},P',U',U_1}} \qquad,\qquad \vlproofd{}{}{\mpar{P'',U_2}}
		$$
		\`E di nuovo \`e possibile applicare l'ipotesi induttiva su $\Pi$ poich\'e:
		$$
			\msizeb{\mpar{\mand{R,T},P',U',U_1}} < \mpar{\mand{\mpar{\mand{R,T},P',U'},P''},U''}
		$$
		per ottenere:
		$$
			\vlderd{}{}{\mpar{P',U',U_1}}{\mpar{P_1,P_2}} \qquad,\qquad \vlproofd{}{}{\mpar{R,P_1}} \qquad,\qquad \vlproofd{}{}{\mpar{T,P_2}}
		$$
		Ora possiamo costruire:
		$$
			\vlderivation{
				\vldd{}{}{\mpar{\mand{P',P''},U',U'']}}{
					\vlin{}{\mrule{s}}{\mpar{\mand{P',P''},U',U_1,U_2}}{
						\vldd{}{}{\mpar{\mand{P',\mpar{P'',U_2}},U',U_1}}{
							\vldd{}{}{\mpar{P',U',U_1}}{\vlhy{\mpar{P_1,P_2}}}
						}
					}
				}
			}
		$$
	\end{enumerate}
\end{enumerate}
\end{proof}

\begin{thm}[Riduzione del contesto]\label{thm:ctx_reduction}
Per ogni struttura $P$ ed ogni contesto $\mctx*{C}$ tale che $\mctx{C}{P}$ \`e dimostrabile in \textsf{LBV}, esiste una struttura $C$ tale che, per ogni struttura $X$ esistono le derivazioni:
$$
	\vlderd{}{\mathsf{LBV}}{\mctx{C}{X}}{\mpar{C,X}}
	\qquad\mbox{e}\qquad
	\vlproofd{\mathsf{LBV}}{}{\mpar{C,P}}
$$
\end{thm}
\begin{proof}
Per induzione sulla dimensione di $\mctx{C}{\emptyctx}$. Il caso base \`e triviale, $C=\circ$. I casi induttivi sono:
\begin{enumerate}
	\item $\mctx{C}{\emptyctx} = \mand{\mctx{C'}{\emptyctx}, Q}$, per qualche $Q \not= \circ$. Se $\mctx{C}{P}$ \`e dimostrabile, allora devono esistere le dimostrazioni di $\mctx{C'}{P}$ e di $Q$. Applicando l'ipotesi induttiva su $\mctx{C'}{P}$, otteniamo $C$ tale che, per ogni $X$:
	$$
		\vlderivation{
			\vlde{}{}{\mand{\mctx{C'}{X}, Q}}{
				\vlde{}{}{\mctx{C'}{X}}{\vlhy{\mpar{C, X}}}
			}
		}
	$$
	e tale che $\mpar{C,P}$ \`e dimostrabile in \textsf{LBV}. Lo stesso argomento si applica quando $\mctx{C}{\emptyctx} = \mand{Q, \mctx{C'}{\emptyctx}}$ con $Q \not= \circ$.
	\item $\mctx{C}{\emptyctx} = \mpar{\mctx{C'}{\emptyctx}, Q}$, per qualche $Q \not= \circ$. Assumiamo che $\mctx{C'}{\emptyctx}$ non sia un par: questa ipotesi non \`e limitativa, perch\'e \`e sempre possibile far ``rientrare'' il parallelo in $Q$, lasciando $\mctx*{C'}$ come copar. Se alla fine di questo processo otteniamo $\mctx{C'}{\emptyctx} = \emptyctx$, il teorema \`e banalmente provato. Quindi $\mctx{C'}{\emptyctx} = \mand{\mctx{C''}{\emptyctx},Q'}$ con $Q' \not= \circ$. Per il Teorema~\ref{thm:shallow_split}, esistono:
	$$
		\vlderd{}{\mathsf{LBV}}{Q}{\mpar{Q_1,Q_2}}
		\qquad,\qquad
		\vlproofd{\Pi}{\mathsf{LBV}}{\mpar{\mctx{C''}{P}, Q_1}}
		\qquad,\qquad
		\vlproofd{}{\mathsf{LBV}}{\mpar{Q', Q_2}}
	$$
	Ora, applicando l'ipotesi induttiva su $\Pi$, otteniamo:
	$$
		\vlderivation{
			\vldd{}{}{\mpar{\mand{\mctx{C''}{X}, Q'}, Q} = \mand{\mctx{C'}{X}, Q} = \mctx{C}{X}} {
				\vlin{}{\mrule{s}}{\mpar{\mand{\mctx{C''}{X}, Q'}, Q_1, Q_2}}{
					\vldd{}{}{\mpar{\mand{\mpar{Q',Q_2}, \mctx{C''}{X}}, Q_1}}{
						\vldd{}{}{\mpar{\mctx{C''}{X}, Q_1}}{\vlhy{\mpar{C, X}}}
					}
				}
			}
		}
		\qquad\mbox{e}\quad
		\vlproofd{}{}{\mpar{C, P}}
	$$
	Analogamente si dimostra $\mctx{C'}{\emptyctx} = \mand{Q',\mctx{C''}{\emptyctx}}$ e si usa lo stesso argomento per $\mctx{C}{\emptyctx} = \mpar{Q,\mctx{C'}{\emptyctx}}$.
\end{enumerate}
\end{proof}

\begin{cor}[Splitting]\label{cor:splitting}
Per ogni struttura $P$ e $Q$ e contesto $\mctx*{C}$, se $\mctx*{C}\mand{P,Q}$ \`e dimostrabile in \textsf{LBV}, allora esistono due strutture $S_1$ e $S_2$ tali che, per ogni struttura $X$, esistono le derivazioni:
$$
	\vlderd{}{\mathsf{LBV}}{\mctx{C}{X}}{\mpar{X,S_1,S_2}}
	\qquad,\qquad
	\vlproofd{}{\mathsf{LBV}}{\mpar{P, S_1}}
	\qquad,\qquad
	\vlproofd{}{\mathsf{LBV}}{\mpar{Q, S_2}}
$$
\end{cor}
\begin{proof}
Prima si applica il Teorema~\ref{thm:shallow_split}, poi il Teorema~\ref{thm:ctx_reduction}.
\end{proof}

Infine, prima di passare alla cut elimination, occorre enunciare un ultimo semplice risultato.

\begin{prop}\label{prop:lbv_ctx_ins}
Per ogni struttura $P, Q$ e contesto $\mctx*{C}$, esiste una derivazione:
$$
	\vlderd{}{\{\mvlrule{s}\}}{\mpar{\mctx{C}{P}, Q}}{\mctx*{C}\mpar{P, Q}}
$$
\end{prop}
\begin{proof}
Per induzione sulla dimensione del contesto $\mctx*{C}$. Questa dimostrazione \`e uguale a quella della Proposizione~\ref{prop:sks_ctx_ins}, che stabilisce la stessa proposizione per il Sistema \textsf{KS}.
\end{proof}

\begin{thm}[Cut elimination]
La regola $\mruleup{ai}$ \`e ammissibile in \textsf{LBV}.
\end{thm}
\begin{proof}
Consideriamo la dimostrazione:
$$
	\vlderivation{
		\vlin{}{\mruleup{ai}}{\mctx{C}{\circ}}{
			\vlpd{}{\mathsf{LBV}}{\mctx*{C}\mand{a, \mneg{a}}}
		}
	}
$$
Per il Corollario~\ref{cor:splitting}, esistono $S_1$ e $S_2$ tali che esistono le derivazioni:
$$
	\vlderd{}{\mathsf{LBV}}{\mctx{C}{\circ}}{\mpar{S_1,S_2}}
	\qquad,\qquad
	\vlproofd{\Pi_1}{\mathsf{LBV}}{\mpar{a, S_1}}
	\qquad,\qquad
	\vlproofd{\Pi_2}{\mathsf{LBV}}{\mpar{\mneg{a}, S_2}}
$$
Vogliamo individuare, nella dimostrazione $\Pi_1$, il punto in cui l'atomo $a$ viene introdotto. Certamente deve esistere un contesto $\mctx*{C'}$ tale che $S_1 = \mctx{C'}{\mneg{a}}$. Inoltre, deve esistere un contesto $\mctx*{C''}$ tale che:
$$
	\vlderivation{
		\vldd{\Pi_1'}{}{\mpar{a, \mctx{C'}{\mneg{a}}}}{
			\vlin{}{\mruledn{ai}}{\mctx*{C''}\mpar{a,\mneg{a}}}{
				\vlpd{\Pi_1''}{}{\mctx{C''}{\circ}}
			}
		}
	}
$$
sia la dimostrazione $\Pi_1$ in cui abbiamo individuato l'applicazione della regola $\mruledn{ai}$. Ora, sostituendo in $\Pi_1'$ le occorrenze di $a$ e $\mneg{a}$ con $\circ$, otteniamo una dimostrazione $\Psi_1'$, grazie alla quale \`e possibile dimostrare:
$$
	\vlderivation{
		\vldd{\Psi_1'}{}{\mctx{C'}{\circ}}{
			\vlpd{\Pi_1''}{}{\mctx{C''}{\circ}}
		}
	}
$$
Analogamente possiamo trasformare la dimostrazione di $\Pi_2$ in una dimostrazione di $\mctx{D'}{\circ}$ dove $S_2 = \mctx{D'}{a}$. Ora possiamo concludere, esibendo la seguente dimostrazione:
$$
	\vlderivation{
		\vlde{}{}{\mctx{C}{\circ}}{
			\vldd{\Phi}{}{\mpar{\mctx{C'}{\mneg{a}}, \mctx{D'}{a}}}{
				\vlin{}{\mruledn{ai}}{\mctx{C'}{\mctx*{D'}\mpar{a, \mneg{a}}}}{
					\vlde{}{}{\mctx{C'}{\mctx{D'}{\circ}}}{
						\vlpr{}{}{\mctx{C'}{\circ}}
					}
				}
			}
		}
	}
$$
in cui $\Phi$ \`e ottenuta applicando due volte la Proposizione~\ref{prop:lbv_ctx_ins}.

Possiamo ripetere induttivamente l'argomento per ogni dimostrazione di $\mathsf{LBV} \cup \{\mruleup{ai}\}$, partendo dall'alto, ed eliminare una per una tutte le istanze di $\mruleup{ai}$.
\end{proof}

\subsection{Un'interpretazione operazionale}\label{sec:lbv_opi}

Quando lette dal basso all'alto, cio\`e nel verso della \emph{proof search}, le regole d'inferenza possono essere direzionate, trasponendole in \emph{regole di riscrittura}. Questo \`e reso possibile dal fatto che, nel calcolo delle strutture, ogni regola ha sempre \emph{al pi\`u una premessa}. Riscriviamo le regole del Sistema \textsf{LBV} sotto questa nuova prospettiva; la regola d'interazione atomica diventa:
$$
	\mpar{a, \mneg{a}} \rightarrow \circ \qquad \mrule{ai}
$$
che dice che due atomi di polarit\`a opposta messi in parallelo possono interagire. Abbiamo omesso la chiusura contestuale; nei sistemi di riscrittura si \`e soliti fattorizzare la chiusura contestuale con una regola:
\begin{center}
	\AxiomC{$P \rightarrow Q$}
	\RightLabel{$\mrule{di}$}
	\UnaryInfC{$\mctx{C}{P} \rightarrow \mctx{C}{Q}$}
	\DisplayProof{}
\end{center}
che per noi corrisponde moralmente all'impiego della metodologia deep inference.

L'assioma del Sistema \textsf{LBV} non viene trasposto: infatti in questa interpretazione, significa solo che l'unit\`a non pu\`o essere riscritta, e che essa rappresenta l'unico \emph{valore} del Sistema. Procedendo nel processo di riscrittura, potremo in generale imbatterci in situazioni in cui nessuna regola \`e applicabile; l'unico caso ``accettato'' \`e quello di una computazione che termina sul simbolo $\circ$. Negli altri casi, il processo di riscrittura non avr\`a individuato una dimostrazione, bens\`i una derivazione \emph{bloccata}.

L'equivalenza tra formule, definita come in Figura~\ref{fig:sys_lbv}, \`e nota nel mondo dei sistemi di riscrittura, come \emph{riscrittura modulo} una certa relazione d'equivalenza, come riportato in~\cite{BaaNip98}. Un modo di esprimerla \`e usare la regola:
\begin{center}
	\AxiomC{$P = P'$}
	\AxiomC{$P' \rightarrow Q'$}
	\AxiomC{$Q' = Q$}
	\RightLabel{$\mrule{eq}$}
	\TrinaryInfC{$P \rightarrow Q$}
	\DisplayProof{}
\end{center}

La riscrittura modulo associativit\`a, commutativit\`a e identit\`a \`e delicata, perch\'e non \`e sempre terminante. Per garantire la terminazione occorre imporre dei vincoli sull'applicabilit\`a della regola $\mrule{eq}$, come mostrato in~\cite{BaiPetWil89}. Inoltre~\cite{Kah06} ha sviluppato una tecnica per ridurre il non-determinismo dettato dalla fine grana delle regole di \textsf{LBV}.

Se la regola di switch \`e di difficile comprensione quando la si considera come regola d'inferenza, la sua trasposizione in regola di riscrittura offre una prospettiva molto pi\`u intuitiva. Nell'approccio operazionale, consideriamo le due formule all'interno di un ``par'' come due processi paralleli, che girano simultaneamente e che si ``conoscono'' a vicenda (cio\`e che hanno modo di individuarsi, ad esempio possiedono le reciproche coordinate all'interno di una rete), e pertanto sono in grado di interagire. I processi in ``copar'' girano anch'essi in parallelo, ma non hanno la possibilit\`a di comunicare, perch\'e non possiedono l'informazione sulle reciproche coordinate. Oltre all'interazione, l'unica altra operazione possibile per i processi, \`e ``dimenticarsi'': un processo all'interno di un ``par'' pu\`o decidere di eliminare l'informazione sulle coordinate di un'altro processo, o in altre parole, \emph{disconoscerlo}. L'effetto di questo comportamento, \`e inibire ogni possibilit\`a di interazione tra i processi coinvolti.

Questo meccanismo d\`a luogo a una serie di casi. Siano $P$ e $Q$ due processi in un'ambiente ``par'':
$$
	\mpar{P, Q}{\;}\rightarrow{\;}?
$$
Cosa possono fare $P$ e $Q$ a parte interagire?
\begin{enumerate}[label=\arabic*.]
	\item possono lavorare per conto loro, cio\`e, ai fini del comportamento concorrente, non fare niente. Questo caso va contemplato per completezza dell'interpretazione: $\mpar{P, Q} \rightarrow \mpar{P, Q}$;
	\item $P$ pu\`o disconoscere $Q$ o viceversa: in entrambi in casi $\mpar{P, Q} \rightarrow \mand{P,Q}$;
	\item\label{lbv_opi_sw1} se $P$ \`e una composizione parallela $\mpar{P_1, P_2}$:
	\begin{enumerate}[label=\arabic{enumi}.\arabic*.]
		\item $P_1$ pu\`o disconoscere $Q$: $\mpar{P_1, P_2, Q} \rightarrow \mpar{\mand{P_1,Q}, P_2}$;
		\item $P_2$ pu\`o disconoscere $Q$: $\mpar{P_1, P_2, Q} \rightarrow \mpar{\mand{P_2,Q}, P_1}$;
		\item se $P_1$ \`e una composizione parallela $\mpar{P_{11}, P_{12}}$ \ldots
		\item se $P_1$ \`e un ``copar'' $\mand{P_{11}, P_{12}}$ \ldots
		\item se $P_2$ \`e una composizione parallela $\mpar{P_{21}, P_{22}}$ \ldots
		\item se $P_2$ \`e un ``copar'' $\mand{P_{21}, P_{22}}$ \ldots
	\end{enumerate}
	\item\label{lbv_opi_sw2} se $P$ \`e un ``copar'' $\mand{P_1, P_2}$:
	\begin{enumerate}[label=\arabic{enumi}.\arabic*.]
		\item $P_1$ pu\`o disconoscere $Q$: $\mpar{\mand{P_1, P_2}, Q} \rightarrow \mand{\mpar{P_2, Q}, P_1}$;
		\item $P_2$ pu\`o disconoscere $Q$: $\mpar{\mand{P_1, P_2}, Q} \rightarrow \mand{\mpar{P_1, Q}, P_2}$;
		\item se $P_1$ \`e una composizione parallela $\mpar{P_{11}, P_{12}}$ \ldots
		\item \ldots
	\end{enumerate}
	\item se $P$ \`e una composizione parallela $\mpar{P_1, P_2}$, si procede in maniera simmetrica rispetto al caso \ref{lbv_opi_sw1};
	\item se $P$ \`e un ``copar'' $\mand{P_1, P_2}$, simmetricamente rispetto a \ref{lbv_opi_sw2}.
\end{enumerate}

\begin{figure}[t!]
\vspace{.7em}
\begin{minipage}[t]{.7\textwidth}\end{minipage}
\begin{minipage}[t]{.93\textwidth}
$\mmerge{P}{Q} = \{ \mpar{P, Q}, \mand{P, Q} \} \cup \mmerge*{P}{Q} \cup \mmerge*{Q}{P}$ \
\vspace{.6em} \\
dove: \
\vspace{.4em} \\
\begin{tabular}{lrcl}
	& $\mmerge*{\circ}{Q} = \mmerge*{a}{Q}$ & $=$ & $\varnothing$ \\
	& $\mmerge*{\mpar{S, T}}{Q}$ & $=$ & $\{ \mpar{S, X} \:|\: X \in \mmerge{T}{Q} \} \cup \{ \mpar{X, T} \:|\: X \in \mmerge{S}{Q} \}$ \\
	& $\mmerge*{\mand{S, T}}{Q}$ & $=$ & $\{ \mand{S, X} \:|\: X \in \mmerge{T}{Q} \} \cup \{ \mand{X, T} \:|\: X \in \mmerge{S}{Q} \}$
\end{tabular}
\end{minipage}
\vspace{.4em}
\caption{Definizione dell'operatore di \emph{merge}}
\label{fig:merge}
\end{figure}

Questa struttura, chiaramente ricorsiva, porta alla definizione di~\cite{Gug02} di \emph{merge set}, riportata in Figura~\ref{fig:merge}. Due processi $P$ e $Q$ immersi in un contesto parallelo, possono muovere ad un processo $R$ appartenente al merge set, come prescritto dalla \emph{regola di merge}:
$$
	\mpar{P, Q} \rightarrow R \;\mbox{ se } R \in \mmerge{P}{Q} \qquad \mrule{g}
$$

Questa regola \`e abbastanza pesante, poich\'e necessita il ricalcolo del merge set ogni volta che dev'essere applicata. Ecco perch\'e facciamo ricorso alla regola di switch:
$$
	\mpar{\mand{P, Q}, R} \rightarrow \mand{P, \mpar{Q, R}} \qquad \mrule{s}
$$

\begin{thm}
La regola di merge \`e eliminabile in presenza di $\mrule{s}$.
\end{thm}
\begin{proof}
Siano $P$, $Q$ ed $R$ processi, tali che $R \in \mmerge{P}{Q}$. Allora:
$$
	\mpar{P, Q} \rightarrow R \qquad \mrule{g}
$$
Vogliamo dimostrare che esiste un cammino:
$$
	\mpar{P, Q} \rightarrow^* R
$$
composto di sole applicazioni di $\mrule{s}$, di $\mrule{eq}$ e di $\mrule{di}$.

Procediamo per induzione strutturale su $R$, usando la definizione di merge set di Figura~\ref{fig:merge}:
\begin{enumerate}[label=\arabic*.]
	\item $R = \circ$. Allora dev'essere $P = Q = \circ$, e quindi $\mpar{P,Q} = R$;
	\item $R = a$. Allora dev'essere $P = a$ e $Q = \circ$ o $P = \circ$ e $Q = a$. In entrambi i casi $\mpar{P, Q} = a = R$;
	\item $R = \mpar{P, Q}$ vale banalmente, in quanto $\mpar{P, Q} \rightarrow^* \mpar{P, Q}$ per ogni $P$, $Q$;
	\item $R = \mand{P, Q}$. Allora:
	\begin{center}
		\AxiomC{$\mpar{P, Q} = \mpar{\mand{P, \circ}, Q}$}
		\AxiomC{$\mpar{\mand{P, \circ}, Q} \rightarrow \mand{P, \mpar{\circ, Q}}$}
		\AxiomC{$\mand{P, \mpar{\circ, Q}} = \mand{P, Q}$}
		\TrinaryInfC{$\mpar{P, Q} \rightarrow \mand{P, Q}$}
		\DisplayProof{}
	\end{center}
	\item $R = \mpar{R', R''}$. Ci sono due sottocasi da considerare:
	\begin{enumerate}[label=\arabic{enumi}.\alph*.]
		\item $P = \mpar{R', P'}$ e $R'' \in \mmerge{P'}{Q}$. Per ipotesi induttiva, sappiamo che $\mpar{P', Q} \rightarrow^* R''$ senza usare la regola $\mrule{g}$. Allora possiamo concludere, esibendo:
		\begin{center}
			\AxiomC{$\mpar{P', Q} \rightarrow^* R''$\;\;\:}
			\RightLabel{$\mrule{di}$}
			\UnaryInfC{$\mpar{R', P', Q} \rightarrow^* \mpar{R', R''}$}
			\DisplayProof{}
		\end{center}
		e osservando che $\mpar{P, Q} = \mpar{R', P', Q}$;
		\item $P = \mpar{P'', R''}$ e $R' \in \mmerge{P''}{Q}$. Allora, come prima:
		\begin{center}
			\AxiomC{$\mpar{P'', Q} \stackrel{\mathsf{I.H.}}{\rightarrow^*} R'$\;\;\;}
			\RightLabel{$\mrule{di}$}
			\UnaryInfC{$\mpar{P'', Q, R''} \rightarrow^* \mpar{R', R''}$}
			\DisplayProof{}
		\end{center}
		e $\mpar{P, Q} = \mpar{P'', Q, R''}$.
	\end{enumerate}
	\item $R = \mand{R', R''}$. Di nuovo, seguendo la definizione di merge set, ci sono due sottocasi possibili:
	\begin{enumerate}[label=\arabic{enumi}.\alph*.]
		\item $P = \mand{R', P'}$ e $R'' \in \mmerge{P'}{Q}$. Poich\'e $\mpar{P,Q} = \mpar{\mand{R', P'}, Q}$, possiamo applicare la regola di switch per ottenere:
		$$
			\mpar{\mand{R', P'}, Q} \rightarrow \mand{R', \mpar{P', Q}}
		$$ 
		da cui, per ipotesi induttiva immersa nel contesto $\mand{R', \emptyctx}$, possiamo concludere esibendo:
		\begin{center}
			\AxiomC{$\mpar{P', Q} \stackrel{\mathsf{I.H.}}{\rightarrow^*} R''$}
			\RightLabel{$\mrule{di}$}
			\UnaryInfC{$\mand{R', \mpar{P', Q}} \rightarrow^* \mand{R', R''}$}
			\DisplayProof{}
		\end{center}
		\item $P = \mand{P'', R''}$ e $R' \in \mmerge{P''}{Q}$. Allora, grazie alle regole $\mrule{s}$ ed $\mrule{eq}$:
		$$
			\mpar{P, Q} = \mpar{\mand{R'', P''}, Q} \rightarrow \mand{R'', \mpar{P'', Q}} = \mand{\mpar{P'', Q}, R''}
		$$
		da cui concludiamo, grazie all'ipotesi induttiva e a:
		\begin{center}
			\AxiomC{$\mpar{P'', Q} \stackrel{\mathsf{I.H.}}{\rightarrow^*} R'$}
			\RightLabel{$\mrule{di}$}
			\UnaryInfC{$\mand{\mpar{P'', Q}, R''} \rightarrow^* \mand{R', R''}$}
			\DisplayProof{}
		\end{center}
	\end{enumerate}
\end{enumerate}
\end{proof}

Grazie a questo risultato, possiamo appoggiarci sulla pi\`u pratica (nonch\'e locale) regola di switch, eliminando quella di merge. Il sistema cos\`i ottenuto, costituisce una personale interpretazione operazionale del Sistema \textsf{LBV}, reminiscente le algebre di processo, che pu\`o costituire un (ulteriore) ponte tra il mondo della proof theory e quello dei modelli concorrenti, oltre ad un punto di partenza per indagini riguardanti le propriet\`a di complessit\`a della proof search.

\clearpage{\pagestyle{empty}\cleardoublepage}
\phantomsection\addcontentsline{toc}{chapter}{Conclusioni}
\chapter*{Conclusioni}
\rhead[\fancyplain{}{\bfseries CONCLUSIONI}]{\fancyplain{}{\bfseries\thepage}}
\lhead[\fancyplain{}{\bfseries\thepage}]{\fancyplain{}{\bfseries CONCLUSIONI}}

La deep inference offre una prospettiva nuova e moderna in teoria della dimostrazione. Grazie a questa metodologia, il lavoro strutturale svolto dagli alberi in shallow inference, viene collassato nell'uso dei contesti, che sono un concetto fondamentale in questo approccio. A questo proposito, \`e interessante osservare come questo metta in relazione diretta il modo di operare tipico in proof theory (con alberi di derivazione) con il mondo dei sistemi di riscrittura. Infatti le derivazione nel calcolo delle strutture possono essere linearizzate, leggendole dal basso in alto (verso della proof search), e le regole d'inferenza si possono vedere come regole di riscrittura; quali corrispondenze si possono trovare in questo senso? A quali sistemi di riscrittura corrispondono i sistemi in calcolo delle strutture, e di quali propriet\`a godono?

Inoltre, osserviamo come, nelle procedure di cut elimination per il calcolo delle strutture, l'attenzione sia posta sugli atomi da eliminare, a conseguenza del fatto che questi sistemi godono di localit\`a. Una sotto-procedura invariante nella cut elimination \`e la discesa nella dimostrazione alla ricerca del taglio, per poi risalire seguendo il flusso degli atomi coinvolti. Da questa osservazione nasce un nuovo filone di ricerca in deep inference che tratta i cosiddetti ``flussi atomici'' o ``atomic flows''; per una introduzione all'argomento, vedere~\cite{Gund09}.

Infine, esistono molti problemi rilevanti riguardanti la complessit\`a delle dimostrazioni, alcuni dei quali ancora aperti, altri gi\`a risolti, ad esempio in~\cite{Jer09, BruGug09, BruGugGunPar09}. Il calcolo delle strutture \`e un formalismo molto espressivo, ma difficile da trattare a causa del forte non-determinismo che comporta la fine grana (i.e. la vasta applicabilit\`a) delle sue regole; \cite{Kah06}~ha ideato una tecnica capace di ridurre questo non-determinismo.

Tutte le fonti e le informazioni riguardanti le ricerche in deep inference, sono reperibili online sulla pagina di Alessio Guglielmi, uno dei maggiori promotori di questo approccio, all'indirizzo:
\begin{center}
	\url{http://alessio.guglielmi.name/res/cos/}
\end{center}


\clearpage{\pagestyle{empty}\cleardoublepage}
\rhead[\fancyplain{}{\bfseries BIBLIOGRAFIA}]{\fancyplain{}{\bfseries\thepage}}
\lhead[\fancyplain{}{\bfseries\thepage}]{\fancyplain{}{\bfseries BIBLIOGRAFIA}}
\phantomsection\addcontentsline{toc}{chapter}{Bibliografia}
\bibliographystyle{humanbio}
\bibliography{iidi}


\clearpage{\pagestyle{empty}\cleardoublepage}
\chapter*{Ringraziamenti}
\thispagestyle{empty}

Ringrazio anzitutto i mei genitori Carmen e Roberto, senza i quali tutto questo non sarebbe stato possibile. Con loro, ringrazio tutta la mia famiglia per l'amore che mi hanno dato dacch\'e sono al mondo.

Ringrazio i miei amici per le ore passate a discutere insieme, per aver ascoltato pazientemente i miei vaneggianti sproloqui, ma soprattutto per avermi dato la certezza di aver sempre qualcuno su cui contare.

Infine ringrazio i miei professori, per avermi ascoltato e per la pazienza che hanno avuto nel sopportare questo tremendo rompiscatole. Senza i vostri insegnamenti, ma non solo, senza il vostro esempio, non sarei quello che sono.

Grazie di cuore a tutti quanti, grazie a chi ha sempre creduto in me, grazie a chi non ci ha creduto mai, grazie agli amici ma anche ai nemici, grazie al contributo di tutti perch\'e mi \`e stato indispensabile per raggiungere, oggi, questo risultato.

\end{document}